\documentclass{article}

 \usepackage[main, final]{neurips_2026}

\usepackage[utf8]{inputenc}
\usepackage[T1]{fontenc}
\usepackage{url}
\usepackage{booktabs}
\usepackage{amsfonts}
\usepackage{microtype}
\usepackage{xcolor}
\usepackage{amsmath}
\usepackage{amsthm}
\usepackage{amssymb}
\usepackage{algorithm}
\usepackage{algpseudocode}
\usepackage{float}
\usepackage{graphicx}
\usepackage{multirow}
\usepackage{bm}
\usepackage{hyperref}

\newtheorem{theorem}{Theorem}
\newtheorem{lemma}{Lemma}
\newtheorem{proposition}{Proposition}
\newtheorem{corollary}[theorem]{Corollary}

\newtheorem{condition}{Condition}
\newtheorem{remark}{Remark}

\newcommand{\one}[1]{\mathbf{1}\left\{ #1 \right\}}

\def\argmin{\mathop{\arg\min}}

\def\P{\mathbb{P}}
\newcommand{\E}{\mathbb{E}}
\newcommand{\R}{\mathbb{R}}
\newcommand{\var}{\mathrm{Var}}
\newcommand{\cov}{\mathrm{Cov}}
\newcommand{\supp}{\mathrm{supp}}

\title{High-dimensional Gaussian Graphical Model Testing\\ for Long-Memory Time Series}

\author{%
  Percy S.~Zhai \\
  The University of Chicago\\
  Chicago, IL 60637 \\
  \texttt{percy.zhai@chicagobooth.edu} \\
  \And
  Ping-Shou Zhong \\
  University of Illinois Chicago \\
  Chicago, IL 60607 \\
  \texttt{pszhong@uic.edu} \\
  \AND
  Wei Biao Wu \\
  The University of Chicago \\
  Chicago, IL 60637 \\
  \texttt{wbwu@uchicago.edu} \\
}

\begin{document}

\maketitle

\begin{abstract}

Many real-world high-dimensional time series exhibit long-memory, but Gaussian graphical model testing in this regime remains understudied.
We develop a direct, data-adaptive test statistic for assessing conditional independence in the graph structure of stationary Gaussian time series.
We establish a finite-sample, Berry--Esseen type Gaussian approximation bound for the statistic, which applies to both short-memory and long-memory time series.
The testing procedure uses a block bootstrap method, on which we provide a finite-sample validity result including in the ultra-high-dimensional scenario, and can be extended to comparing graphical structures in two-sample tests.
We also develop a consistency-empowered correction to the statistic and show that such tests attain asymptotic consistency in both size and power.
Our proposed method is applied to real-world fMRI data to understand functional connectivities within brain in different periods.

\end{abstract}

\section{Introduction}

Testing on Gaussian graphical models for high-dimensional time series remains relatively understudied as opposed to its i.i.d.~counterparts.
Existing work has largely been confined to the short-memory regime, where the autocorrelation decays fast enough.
The main difficulty, though, is that real-world time series data usually exhibits long-range temporal dependence (i.e., long-memory).
This paper proposes a novel method for testing Gaussian graphical models of stationary high-dimensional time series that allows for long-range dependence.

Gaussian graphical models encode conditional independence through the zero pattern of the precision matrix \citep{lauritzen1996graphical}.
In the i.i.d.~setting, a large body of work has studied precision matrix estimation, including neighborhood selection \citep{meinshausen2006high}, penalized likelihood \citep{friedman2008sparse}, and constrained $\ell_1$ methods \citep{cai2011constrained}.
The corresponding inference literature is more recent.
One direction focuses on edgewise or low-dimensional inference for entries of $\Omega$, where multiple-testing procedures \citep{liu2013gaussian} or debiased precision estimators \citep{ren2015asymptotic, jankova2015confidence} are used to test individual conditional dependencies.
Another direction develops global or structural tests for precision matrices, including neighborhood tests \citep{verzelen2009tests, zhao2024high}, tests for special structures such as bandedness \citep{cheng2017bandedness}, and two-sample differential-network tests \citep{xia2015testing, zhao2022fudge}.
\citet{le2023universally} proposed direct plug-in and consistency-empowered tests for the adequacy of a prespecified Gaussian graph.
More recently, \citet{lin2025goodness} proposed exact finite-sample goodness-of-fit tests for high-dimensional Gaussian graphical models via exchangeable sampling.
These methods provide important tools for i.i.d.~Gaussian observations, but require independence or exchangeability.

For time series, graph-related work has mostly focused on estimation rather than graph-specification testing.
\citet{chen2013covariance} studied covariance and precision matrix estimation for high-dimensional time series using functional dependence measure.
\citet{qiu2016joint} considered joint estimation of multiple graphical models from high-dimensional time series with dependent observations.
\citet{shu2019estimation} developed covariance and precision estimation for processes with slowly decaying temporal dependence,
and \citet{baek2023local} studied high-dimensional long-run variance and long-run precision estimation under general dependence, including long-range dependence.
Apart from these valuable estimation methods, a data-driven goodness-of-fit test for long-memory high-dimensional time series is largely unexplored.

Meanwhile, a separate line of work develops high-dimensional inference for short-memory time series.
\citet{zhang2017gaussian} developed Gaussian approximation for high-dimensional stationary time series, and \citet{zhang2014bootstrapping} studied bootstrap inference for high-dimensional time series.
These results are designed for short-range dependence and do not easily extend into the long-memory regime.
Only recently did \citet{zhai2026simultaneous} establish Gaussian approximation and block bootstrap validity for sample covariance errors in high-dimensional long-memory time series.
The same work also developed precision matrix inference, but only in the low-dimensional regime where the sample covariance matrix is invertible.

This paper fills the gap between these two literatures.
We use a graph-specific statistic adapted from \citet{le2023universally} based on off-support local inverse, and demonstrate that the first-order behavior of the statistic is governed by the sample covariance error, which further enables us to import the machinery in \citet{zhai2026simultaneous}, without the need of estimating the full precision matrix.
The new graph-specific steps are the linearization and plug-in error control in Propositions~\ref{prop:deterministic.reduction} and~\ref{prop:score.perturb}, together with verification of the imported assumptions for the transformed process $Y_t=HX_t$.
Our Theorem \ref{thm:ga.direct} provides a finite-sample Gaussian approximation result on the proposed statistic.
We then propose a testing method based on block-bootstrap, and show its finite-sample validity in Theorem \ref{thm:block.bootstrap.direct}.
We further add a consistency-empowered correction for the case where the hypothesized graph is larger than the true graph.
The correction screens precision estimates near zero on the null support and artificially amplifies them.
It preserves the null calibration of the direct test, while restoring power against fixed over-specified alternatives; see Theorem~\ref{thm:ce.direct}.

For application, we consider a two-sample extension.
Given two samples, we estimate a graph from the first group and use it as the null support for testing the second group.
Corollary \ref{cor:two.sample.direct.informal} shows that the Gaussian approximation and block bootstrap validity remains as long as the graph estimator is support consistent.
Our proposed method is illustrated on the Sherlock fMRI dataset that compares the brain connectivity pattern of subjects when watching a movie and when verbally recalling the contents.

We would list the following three major contributions.
\begin{enumerate}
    \item We propose a novel approach for testing edges in Gaussian graphical models for stationary time series that allows {\bf long-memory}.
    We apply a local-inverse-based test statistic related to sample covariance, and provide a Berry-Esseen type bound of Gaussian approximation and a block-bootstrap validity result.
    These {\bf finite-sample bounds} demonstrate that our test applies to the {\bf ultra-high-dimensional} regime under explicit sparsity and dependence conditions, and guarantees the test size.
    \item We also propose a {\bf consistency-empowered} corrected test statistic for time series, which is additionally capable of detecting edges in the hypothesized graph that are not significant in the true graph.
    We develop theoretical results that the test size is not affected, while the test power is guaranteed.
    \item We impose a {\bf two-sample extension} of the theory that tests whether the graphical structure is identical across the two samples by combining two directional support-containment tests.
    This setup is more convenient in application, and we demonstrate this two-sample methodology on the Sherlock fMRI dataset.
\end{enumerate}

The structure of the paper is as follows.
Section \ref{sec:methodology} introduces the problem setup, propose the statistic, describe the block bootstrap algorithm, and propose the consistency-empowered correction on the statistic.
Section \ref{sec:theory} provides theoretical guarantee of our proposed methods.
Quantitative study that illustrates the theory can be found in Section \ref{sec:simulation}.
The movie study on Sherlock fMRI dataset is given in Section \ref{sec:real.data}.

\section{Methodology}
\label{sec:methodology}

We observe a $p$-dimensional stationary Gaussian linear process $\{X_t\}_{t=1}^n$, where $X_t = \sum_{l=0}^\infty A_l \epsilon_{t-l}$ and $\epsilon_t \sim N(0, I_d)$.
The dimension $p$ is allowed to be much larger than $n$.
The temporal dependence is captured by the decay rate of the coefficient matrices, as specified in the following condition.
\begin{condition}\label{cond.A}
    There exists a constant $\beta>3/4$ and $C_0 > 0$ such that 
    \begin{equation}
    \label{eq:cond.A}
    \max_{1 \le j \le p} |(A_l)_{j \cdot}| \le C_0 (1 \vee l)^{-\beta},
    \qquad
    \beta > \frac{3}{4}.
    \end{equation}
\end{condition}
A larger $\beta$ corresponds to weaker temporal dependence of the process.
$\beta>1$ corresponds to short-range dependence (i.e., short memory).
The more challenging regime is long-range dependence (i.e., long memory) for $\beta\leq 1$, where the autocovariance matrices $\{\Gamma_k\}_{k=0}^\infty$ are no longer summable, thus the usual Central Limit Theorem for time series is not applicable.

Denote the covariance and precision matrices of the process by $\Sigma = \E[X_t X_t^\top]$ and $\Omega = \Sigma^{-1}$.
For Gaussian processes, the support of the precision matrix $E = \supp(\Omega)$ induces a graph.
We want to test whether it is identical to a hypothesized graph $E_0$.
That is,
\[
H_0 : E = E_0,
\qquad
H_1 : E \ne E_0.
\]

\subsection{Direct Statistic}\label{sec:direct.stat}
We build our test statistic based on the sample covariance matrix, $\hat\Sigma_n = n^{-1} \sum_{t=1}^n X_t X_t^\top$, following the construction as in \cite{le2023universally}.
For each $1 \le j \le p$, let $S_{0,j} = \{ i : i = j \text{ or } (i,j) \in E_0 \}$ be the column-wise support set, with degree $s_{0,j} = |S_{0,j}|$.
The sparsity level of the whole graph is measured by $s_0 = \max_{1 \le j \le p} s_{0,j}$.
We first assume that the graph is reasonably sparse, in that $s_0 < n$.
Let $B_{0,j} \in \R^{p \times s_{0,j}}$ be the corresponding selection matrix, such that $B_{0,j}^\top x$ extracts the coordinates of a vector $x \in \R^p$ indexed by $S_{0,j}$.

We write the $j$-th column of the precision matrix $\Omega_0$ as $\bm \omega_{0,j}$.
Under $H_0$, $\Sigma \bm \omega_{0,j} = e_j$ and $\bm \omega_{0,j}$ is zero outside $S_{0,j}$, we thus have
$\bm \omega_{0,j} = B_{0,j} (B_{0,j}^\top \Sigma B_{0,j})^{-1} B_{0,j}^\top e_j$.
Replacing $\Sigma$ by $\hat\Sigma_n$ gives a local-inverse estimator
\begin{equation}
\label{eq:omega.hat}
\hat{\bm \omega}_{0,j}
=
B_{0,j} (B_{0,j}^\top \hat\Sigma_n B_{0,j})^{-1} B_{0,j}^\top e_j.
\end{equation}
Throughout this paper, whenever $B_{0,j}^\top \hat\Sigma_n B_{0,j}$ is singular, its inverse is understood as the zero matrix.
Under the conditions stated later, this exceptional event has vanishing probability.
Our proposed direct statistic is based on the fact that $\Sigma \Omega_0 = I$ on $H_0$.
In particular, define the residual $\hat r_{ij} = e_i^\top \hat\Sigma_n \hat{\bm \omega}_{0,j} - \one{i=j}$.
Since the graph test only concerns entries outside the null support, we focus on the off-support index set $\mathcal J_0 = \{ (i,j) : i \notin S_{0,j} \}$.
The direct statistic is defined as the root-scale maximum of the residuals over $\mathcal J_0$.
Since $i \ne j$ on this set, it can be written as
\begin{equation}
\label{eq:T.hat.def}
\hat T_n
=
\sqrt{n} \max_{(i,j) \in \mathcal J_0} |e_i^\top \hat\Sigma_n \hat{\bm \omega}_{0,j}|.
\end{equation}
When the true underlying graph $E$ contains edges that are missing in the hypothesized graph $E_0$, we will obtain a larger $\hat T_n$, where the rejection region is situated.
There are two caveats though.
First, we need a data-adaptive method to determine the rejection region for $\hat T_n$ --- the block bootstrap method in Section \ref{sec:block.bootstrap} addresses this concern.
Second, while it is easier for $\hat T_n$ to reject $E_0$ when $E_0 \subsetneq E$, it is less capable of rejecting a hypothesized graph larger than the true graph --- we propose a consistency-empowered corrected statistic in Section \ref{sec:ce.direct} that solves this issue.

\subsection{Block Bootstrap}\label{sec:block.bootstrap}

Despite the Gaussian approximation results on $\hat T_n$, its covariance structure is usually unknown in practice.
We use the following block-bootstrap method is proposed to imitate the behavior of the Gaussian approximation in a data-adaptive manner with a carefully chosen block length.
For $(i,j) \in \mathcal J_0$, define the plug-in vector $\hat{\bm u}_{ij} = e_i - B_{0,j} (B_{0,j}^\top \hat\Sigma_n B_{0,j})^{-1} B_{0,j}^\top \hat\Sigma_n e_i$.
For each $1 \le t \le n$ and $(i,j) \in \mathcal J_0$, define the plug-in score
\begin{equation}
\label{eq:X.hat.def}
\hat{\mathcal X}_{t,ij}
=
\hat{\bm u}_{ij}^\top (X_t X_t^\top - \hat\Sigma_n) \hat{\bm \omega}_{0,j},
\end{equation}
and collect these coordinates into the vector $\hat{\mathcal X}_t = (\hat{\mathcal X}_{t,ij})_{(i,j) \in \mathcal J_0}$.
Because the map in \eqref{eq:X.hat.def} is linear in $X_t X_t^\top - \hat\Sigma_n$, we have $n^{-1} \sum_{t=1}^n \hat{\mathcal X}_t = 0$, hence no additional centering is needed.
Let $l = l_n$ be the block length.
For $r = l,\ldots,n$, define the overlapping block sums $\hat B_{r,l} = \sum_{t=r-l+1}^r \hat{\mathcal X}_t$.
The empirical distribution function of the block statistics is
\begin{equation}
\label{eq:F.hat.def}
\hat F_{n,l}(u)
=
\frac{1}{n-l+1} \sum_{r=l}^n \one{ l^{-1/2} |\hat B_{r,l}|_\infty \le u },
\qquad
u \in \R.
\end{equation}
Define the bootstrap critical value by
\begin{equation}
\label{eq:q.hat.def}
\hat q_{1-\alpha}
=
\inf\{ u \in \R : \hat F_{n,l}(u) \ge 1-\alpha \}.
\end{equation}
The direct test rejects $H_0$ whenever $\hat T_n > \hat q_{1-\alpha}$; see Algorithm~\ref{alg:block.bootstrap}.

\begin{algorithm}[t]
\caption{Test procedure for the direct graph statistic}
\label{alg:block.bootstrap}
\begin{algorithmic}[1]
\Require $\{X_t\}_{t=1}^n$, null graph $E_0$, block length $l$, and level $\alpha$
\State Compute sample covariance matrix $\hat\Sigma_n$ and the local inverse estimators $\hat{\bm \omega}_{0,j}$ from \eqref{eq:omega.hat}
\State Compute the direct statistic $\hat T_n$ from \eqref{eq:T.hat.def}
\State Compute the plug-in score vectors $\hat{\mathcal X}_t$ from \eqref{eq:X.hat.def} for $t = 1,\ldots,n$
\State Form the overlapping block sums $\hat B_{r,l}$ for $r = l,\ldots,n$
\State Construct the empirical CDF $\hat F_{n,l}$ from \eqref{eq:F.hat.def} and the quantile $\hat q_{1-\alpha}$ from \eqref{eq:q.hat.def}
\State Reject $H_0$ if $\hat T_n > \hat q_{1-\alpha}$
\end{algorithmic}
\end{algorithm}

\subsection{Consistency-empowered Correction}\label{sec:ce.direct}

While $\hat T_n$ is sensitive to missing edges in $E$ that are not included in $E_0$, it is not sensitive to false edges in $E_0$ that are not in $E$.
To recover power against included alternatives, $E\subsetneq E_0$, we propose a corrected statistic $\tilde T_n$ similar to the consistency-empowered mechanism in \citet{le2023universally}, but modified to adapt to the time series setup.

The detection of false edges in $E_0$ relies on the estimator in \eqref{eq:omega.hat}.
For $i \in S_{0,j}$, if the estimated precision entry $\hat \omega_{0,ij}$ is too close to zero, we add an artificial pump $C_n$ to it such that the test statistic is inflated.
Define $\hat\Delta_{ij} = C_n \one{|\hat\omega_{0,ij}| / \hat\tau_{ij} \le \delta_n}$,
where $C_n$ and $\delta_n$ are tuning parameters,
and set $\hat\Delta_{ij}=0$ for $i \notin S_{0,j}$.
Here, $\hat\tau_{ij}$ is the estimated long-run standard error of $\hat\omega_{0,ij}$ in the sampling space; see Appendix~\ref{app:ce.direct.details} for detailed buildup.
We replace the closed-form i.i.d. standard error of \citet{le2023universally} with an overlapping-block estimator of the long-run standard error and establish its uniform consistency in Lemma~\ref{lem:tau.uniform.consistency}.
Let $\hat{\bm \Delta}_j = (\hat\Delta_{1j},\ldots,\hat\Delta_{pj})^\top$, and $\tilde{\bm \omega}_{0,j} = \hat{\bm \omega}_{0,j}+\hat{\bm \Delta}_j$.
The consistency-empowered statistic is defined as
\begin{equation}
\label{eq:ce.stat}
\tilde T_n
=
\sqrt{n}
\max_{1\le i,j\le p}
\left|
e_i^\top \hat\Sigma_n \tilde{\bm \omega}_{0,j}
-
\one{i=j}
\right|.
\end{equation}
We shall see that for carefully calibrated $C_n$ and $\delta_n$, $\tilde T_n$ does an equally good job in detecting missing edges as $\hat T_n$ while ensuring the power of the test.
The consistency-empowered test also relies on the same block-bootstrap critical value $\hat q_{1-\alpha}$ in \eqref{eq:q.hat.def}.
$H_0$ is rejected whenever $\tilde T_n > \hat q_{1-\alpha}$; see Algorithm \ref{alg:block.bootstrap.ce} in Appendix \ref{app:algo}.

\section{Theoretical Guarantee for the Proposed Statistics}\label{sec:theory}

\subsection{Gaussian Approximation}

For $(i,j) \in \mathcal J_0$, define $\bm u_{ij}
=
e_i - B_{0,j} (B_{0,j}^\top \Sigma B_{0,j})^{-1} B_{0,j}^\top \Sigma e_i$, and the linearized score process $\mathcal X_{t,ij} = \bm u_{ij}^\top (X_t X_t^\top - \Sigma) \bm \omega_{0,j}$.
Then $\mathcal X_t = (\mathcal X_{t,ij})_{(i,j) \in \mathcal J_0}$ is the oracle version of $\hat{\mathcal X}_t$.
The major theoretical finding of this work is that the distribution of the proposed statistic $\hat T_n$ is close to that of $|G|_\infty$, where $G$ is a centered Gaussian vector in $\R^{|\mathcal J_0|}$ with covariance matrix $\Xi = \sum_{k=-\infty}^{\infty} \cov(\mathcal X_0, \mathcal X_k)$.

Our proposed direct statistic $\hat T_n$ is based on local inverse of sample covariance matrix.
The following proposition bridges the gap between two inference tasks, one on the statistic $\hat T_n$, the other on sample covariance error $n^{1/2}|\hat\Sigma_n - \Sigma|_\infty$ by showing that $\hat T_n$ is dominated by a linear term with respect to the covariance error.
This result is a key prelude to the subsequent Gaussian approximation result on $\hat T_n$, essentially enabling high-dimensional Gaussian graphical model goodness-of-fit testing without the need of inference theory on the precision matrix.
To see this, define the associated linearized oracle statistic as $T_n^{\mathrm L} = \left| n^{-1/2} \sum_{t=1}^n \mathcal X_t \right|_\infty$.
The following technical conditions are needed.
\begin{condition}\label{cond:Sigma.Omega.1norm}
    There exist positive constants $C_\Sigma$ and $C_\Omega$ such that the true underlying covariance and precision matrices satisfy $|\Sigma|_1 \leq C_\Sigma$ and $|\Omega|_1 \leq C_\Omega$, respectively.
\end{condition}

\begin{condition}
\label{cond:local.inverse}
There exists a positive constant $C_A$ such that $\max_{1\le j\le p} |(B_{0,j}^\top\Sigma B_{0,j})^{-1}|_1 \le C_A$.
\end{condition}

Examples satisfying Conditions~\ref{cond:Sigma.Omega.1norm}--\ref{cond:local.inverse} are given at the beginning of Appendix~\ref{app:pf.theory}.
A first-order expansion of the local inverse gives the linearized statistic $T_n^{\mathrm L}$ plus a controllable residual.

\begin{proposition}
\label{prop:deterministic.reduction}
Assume Conditions \ref{cond:Sigma.Omega.1norm} and \ref{cond:local.inverse}.
On the event $2 C_A s_0 |\hat\Sigma_n - \Sigma|_\infty \le 1$, we have under $H_0$
\[
|\hat T_n - T_n^{\mathrm L}| \le C_1 \sqrt{n} s_0^2 |\hat\Sigma_n - \Sigma|_\infty^2,
\]
where $C_1 > 0$ depends only on $C_\Sigma$, $C_\Omega$, and $C_A$, and thus
\[
\hat T_n \le C_1 \sqrt{n} |\hat\Sigma_n - \Sigma|_\infty + C_1 \sqrt{n} s_0^2 |\hat\Sigma_n - \Sigma|_\infty^2.
\]

\end{proposition}

The proof is deferred to Appendix~\ref{app:proof.det.reduction}.
This reduction bridges the gap between our proposed statistic and covariance error.
The validity of the goodness-of-fit test is thus associated to the Gaussian approximation of $n^{1/2}|\hat\Sigma_n - \Sigma|_\infty$.
The latter can be solved by an application of Theorem 1 by \citet{zhai2026simultaneous} for high-dimensional, long-memory time series.
The following condition on the autocovariance matrices $\Gamma_k = \cov(X_t, X_{t-k})$ is imposed to prevent the trivial scenario where the distribution of $|\hat\Sigma_n - \Sigma|_\infty$ is degenerate.
\begin{condition}\label{cond.G}
    There exists $c_0>0$ such that for all $1\leq s,t\leq p$, $\sum_{k=-\infty}^{\infty} ( (\Gamma_k)_{ss} (\Gamma_k)_{tt} + (\Gamma_k)_{st} (\Gamma_k)_{ts} ) \ge c_0$.
\end{condition}
Note that the quantity on the left-hand side is the long-run variance of $n^{1/2}\hat\Sigma_n$.
In a similar spirit to Condition \ref{cond.G}, we impose the following condition to prevent the case where the distribution of the score $\hat T_n$ is degenerate.
This time, the long-run variance of the aforementioned score process $\mathcal X_t$, which is directly associated with the test statistic $\hat T_n$, shall be nondegenerate.
\begin{condition}
\label{cond:score.nondegenerate}
There exists $c_1>0$ such that for all $(i,j)\in\mathcal J_0$, $\sum_{k=-\infty}^{\infty} \cov(\mathcal X_{0,ij},\mathcal X_{k,ij}) \ge c_1$.
\end{condition}
Under Conditions \ref{cond.A} and \ref{cond.G}, Theorem 1 of
\cite{zhai2026simultaneous} established a bound on the Kolmogorov distance for Gaussian approximation of the sample covariance error,
\[
\sup_{u \in \R} \left| \P( n^{1/2} |\hat\Sigma_n - \Sigma|_\infty \ge u ) - \P( |Z|_\infty \ge u ) \right| \leq C\Psi(p,n) := C \frac{\log^{(5\tilde\beta+12)/(4\tilde\beta+8)}(pn)}{n^{\tilde\beta/(4\tilde\beta+8)}},
\]
where $\tilde\beta = (4\beta-3) \wedge (2\beta-1)$, and $Z$ is a centered Gaussian vector with the same covariance structure as $n^{1/2} \operatorname{vec}(\hat\Sigma_n - \Sigma)$.
Combining this result with Proposition \ref{prop:deterministic.reduction}, we obtain the following theorem for the finite-sample Gaussian approximation bound for our proposed direct statistic.

\begin{theorem}
\label{thm:ga.direct}
Assume Conditions \ref{cond.A}-\ref{cond:score.nondegenerate}.
Then there exist positive constant $C$ and $c$, depending only on $C_0$, $c_0$, $c_1$, $C_\Sigma$, $C_\Omega$, and $C_A$, such that under $H_0$,
\[
\sup_{u \in \R}
|
\P(\hat T_n \le u) - \P(|G|_\infty \le u)
|
\le
C \Psi(p,n)
+
C p^2 \exp\left( - \frac{c n}{s_0^2} \right)
+
C \frac{s_0^2 \log^{3/2}(pn)}{\sqrt{n}}.
\]
\end{theorem}
In particular, when the right-hand side converges to zero, then the Gaussian approximation is asymptotically valid.
Specifically, the bound vanishes if $\log(pn)=o(n^{\tilde\beta/(5\tilde\beta+12)})$ and $s_0^2\log^{3/2}(pn)/\sqrt n\to0$.
Here, the second term is to control the probability of the event that Proposition \ref{prop:deterministic.reduction} is not covered, and the third term accounts for the linerarization residual between $T_n^{\mathrm{L}}$ and $\hat T_n$.
The proof of Theorem \ref{thm:ga.direct} is deferred to Appendix~\ref{app:proof.ga.direct}.

\subsection{Block Bootstrap}
\label{sec:block.bootstrap.theory}

The Gaussian approximation theorem identifies the target null law, but that law is not directly usable because the covariance of $G$ is unknown.
Now we provide a fully data-driven calibration rule based on overlapping block bootstrap.
The theorem below is what justifies the actual rejection rule used in practice.

The next proposition controls the gap between the plug-in score process and its oracle analogue.
This is the extra approximation step that does not appear in the covariance bootstrap problem of \cite{zhai2026simultaneous}.
Let $\bar{\mathcal X}_n = \frac{1}{n} \sum_{t=1}^n \mathcal X_t$ and $M_n = \max_{1 \le t \le n} |X_t|_\infty$.

\begin{proposition}
\label{prop:score.perturb}
Assume Conditions \ref{cond:Sigma.Omega.1norm} and \ref{cond:local.inverse}.
On the event $2 C_A s_0 |\hat\Sigma_n - \Sigma|_\infty \le 1$, we have $\max_{(i,j) \in \mathcal J_0} |\hat{\bm u}_{ij} - \bm u_{ij}|_1 \le C'_1 s_0 |\hat\Sigma_n - \Sigma|_\infty$.
Moreover,
\begin{equation}
\label{eq:score.perturb.pointwise}
\max_{1 \le t \le n} \max_{(i,j) \in \mathcal J_0}
|\hat{\mathcal X}_{t,ij} - (\mathcal X_{t,ij} - \bar{\mathcal X}_{n,ij})|
\le
C'_1 s_0 |\hat\Sigma_n - \Sigma|_\infty (M_n^2+1).
\end{equation}
Here $C'_1>0$ depends only on $C_A$ and $C_\Sigma$.
\end{proposition}

Proposition \ref{prop:score.perturb} identifies the second approximation error in the bootstrap theorem.
Once this plug-in error is controlled, the bootstrap problem reduces to the oracle block-bootstrap problem for the linearized score process.
The proof is deferred to Appendix~\ref{app:proof.score.perturb}.

\begin{theorem}
\label{thm:block.bootstrap.direct}
Assume Conditions \ref{cond.A}-\ref{cond:score.nondegenerate}.
Consider any fixed $0 < \epsilon < 1$.
Set $\phi = \frac{2\tilde\beta+4}{(3-\epsilon)\tilde\beta+4}$, $\psi = \frac{(5\tilde\beta+12)(1-\epsilon)}{(3-\epsilon)\tilde\beta+4}$, and $l \asymp n^\phi \log^\psi(p)$.
Further define the aggregate bootstrap error rate
\[
\mathfrak R_{n,l}(\epsilon)
=
\frac{\log^{\frac{5\tilde\beta+12}{(6-2\epsilon)\tilde\beta+8}}(pn)}{n^{\frac{\tilde\beta}{(6-2\epsilon)\tilde\beta+8}}}
+
p^2 \exp\left( - \frac{c n}{s_0^2} \right)
+
\frac{s_0^2 \log^{3/2}(pn)}{\sqrt{n}}
+
s_0 \sqrt{\frac{l}{n}} \log^2(pn).
\]
Then under the null hypothesis $H_0$, there exist constants $C>0$ depending only on $C_0$, $c_0$, $c_1$, $C_\Sigma$, $C_\Omega$, and $C_A$, as well as $r_\epsilon\in(0,1]$ depending only on $C$ and $\epsilon$, such that whenever $\mathfrak R_{n,l}(\epsilon)\le r_\epsilon$, with probability no less than $1-3\mathfrak R_{n,l}(\epsilon)^\epsilon$,
\[
\sup_{u\in\R}|\hat F_{n,l}(u)-\P(\hat T_n\le u)|
\le
C\mathfrak R_{n,l}(\epsilon)^{1-\epsilon}.
\]
\end{theorem}
The proof of Theorem \ref{thm:block.bootstrap.direct} is deferred to Appendix~\ref{app:proof.block.bootstrap}.
Theorem \ref{thm:block.bootstrap.direct} provides a finite-sample bound.
When $\mathfrak R_{n,l}(\epsilon) \rightarrow 0$, the block bootstrap method is then asymptotically valid.
Under the null hypothesis, as shown in Appendix~\ref{app:proof.block.bootstrap},
\[
\P_{H_0}(\hat T_n>\hat q_{1-\alpha})\to \alpha.
\]
In addition to the conditions stated after Theorem~\ref{thm:ga.direct}, the bound vanishes if $s_0\sqrt{l/n}\log^2(pn)\to0$.

\subsection{Theoretical Extensions}\label{sec:thm.ext}

A natural extension of Theorems \ref{thm:ga.direct} and \ref{thm:block.bootstrap.direct} is for two-sample testing.
Consider that we observe two stationary Gaussian linear processes $\{X_t^{(1)}\}_{t=1}^{n_1}$ and $\{X_t^{(2)}\}_{t=1}^{n_2}$, and want to test whether the two underlying precision graphs, $E^{(1)}$ and $E^{(2)}$, are identical.
We may, for instance, estimate $E^{(1)}$ by $\hat E^{(1)}$, and use it as the hypothesized graph for the test of $E^{(2)}$.
\begin{corollary}[directional two-sample test, informal] \label{cor:two.sample.direct.informal}
    The Gaussian approximation result in Theorem \ref{thm:ga.direct} and the block bootstrap validity guarantee in Theorem \ref{thm:block.bootstrap.direct} still apply to this two-sample test, but with an extra term related to $\pi_{n_1} = \P(\hat E^{(1)} \ne E^{(1)})$ in the upper bound, respectively.
\end{corollary}
The detailed testing setup and the formal statement of the corollary is deferred to Appendix \ref{sec:two.sample.direct}.
Thresholded CLIME provides a support-consistent graph estimator under the conditions given in Appendix~\ref{sec:clime.support}.
A symmetric test of graph equality is obtained by defining $p_{\mathrm{sym}}=\min\{1,2\min(p_{1\to2},p_{2\to1})\}$, where $p_{1\to2}$ and $p_{2\to1}$ are the two directional $p$-values.
With support-consistent graph selectors for both samples, the Bonferroni argument gives asymptotic level at most $\alpha$ without requiring independence between the directional $p$-values; see Appendix~\ref{sec:two.sample.direct}.

On the other hand, we also provide a brief, informal theorem for the consistency-empowered estimator $\tilde T_n$ in \eqref{eq:ce.stat}.
In theory, if we fix the two tuning parameters as $\delta_n=\sqrt{\log n}$ and $C_n=C_\Delta\sqrt{\log p}$ for some fixed constant $C_\Delta>0$, then the artificial pump is strong enough for $\tilde T_n$ to be rejected when $E_0$ contains false edges in the true underlying graph $E$, thus ensuring the power of the test, without affecting the size of the test.
Under the same assumptions, $\delta_n=\sqrt{\log n}$ may be replaced by any sequence satisfying $\delta_n\to\infty$ and $\delta_n+\sqrt{\log p}=o(\sqrt n)$.

\begin{theorem}[consistency-empowered estimator, informal]
\label{thm:ce.direct}
Assume that the conditions in Theorems \ref{thm:ga.direct} and \ref{thm:block.bootstrap.direct} hold.
Let $0<\epsilon<1$, choose the block length $l$ as in Theorem~\ref{thm:block.bootstrap.direct}, and assume $\mathfrak R_{n,l}(\epsilon)\to0$.
Then
\begin{enumerate}
    \item Under $H_0:E=E_0$, $\P_{H_0}(\tilde T_n=\hat T_n)\to 1$.
    Consequently, Theorems~\ref{thm:ga.direct} and~\ref{thm:block.bootstrap.direct} remain valid with $\tilde T_n$ in place of $\hat T_n$, and $\P_{H_0} ( \tilde T_n>\hat q_{1-\alpha} ) \to \alpha$.
    \item Under fixed alternative $H_1 : E\ne E_0$ subject to some conditions, $\P_{H_1} ( \tilde T_n>\hat q_{1-\alpha} ) \to 1$.
\end{enumerate}
\end{theorem}
The detailed technical discussion and the formal statement of the theorem (i.e. Theorem~\ref{thm:ce.direct.formal}) is deferred to Appendix \ref{app:ce.direct.details}.
The proof combines long-run standard-error estimation with coordinate screening; see Appendix~\ref{app:proof.ce.direct.formal}.

\section{Quantitative Study}\label{sec:simulation}

\subsection{Null Test Validity}\label{sec:null.text.validity}

The first quantitative study aims at demonstrating Gaussian approximation result in Theorem \ref{thm:ga.direct} and the bootstrap validity guarantee in Theorem \ref{thm:block.bootstrap.direct}.
We use the a similar high-dimensional sample-size and dimension regimes as in \cite{zhai2026simultaneous}, i.e., $n \in \{200,400,800\}$ and $p \in \{2000,4000,8000\}$.
For the temporal dependence, we consider three levels, $\beta \in \{2,0.9,0.55\}$.
The case $\beta=2$ represents short memory, $\beta=0.9$ represents long memory covered by the theory, and $\beta=0.55$ is included as an out-of-theory ultra-long-memory case.

The null graph $E_0$ is taken to be a banded graph with fixed sparsity level, $E_0 = \{(i,j): 1 \le |i-j| \le b \}$.
In the original null-validity experiment, we set $b=4$.
We first construct a banded baseline precision matrix $\Omega_0$ supported on $E_0$. 
For $i \ne j$, set $(\Omega_0)_{ij} = \theta_0 \rho^{|i-j|-1} \mathbf{1}\{1 \le |i-j| \le b\}$, where $\theta_0=0.25$ and $\rho=0.5$. 
The diagonal entries are chosen by strict diagonal dominance: $(\Omega_0)_{ii} = 1 + \sum_{j \ne i} |(\Omega_0)_{ij}| + \delta_0$, where we set $\delta_0=0.10$.
This guarantees that $\Omega_0$ is positive definite and that its off-diagonal support is exactly $E_0$.

For each precision matrix, we generate 200 copies of time series (see Appendix \ref{app:fft} for details).
For each copy, we can compute a realization of $\hat\Sigma_n$ and the statistic $\hat T_n$.
Meanwhile, we may compute the covariance matrix $\Xi$ and generate copies of the Gaussian vector $G$.
Based on each time series, we may also compute the bootstrap block sum $\hat B_{r,l}$.
Here we choose $l=n^{2/3}$, same as in \citet{zhai2026simultaneous}.
This choice is a convenient practical substitute for the $\beta$-dependent block-length order in Theorem~\ref{thm:block.bootstrap.direct}.
To demonstrate the validity of the Gaussian approximation result in Theorem~\ref{thm:ga.direct} or \ref{thm:block.bootstrap.direct}, we may compare the distribution of $\hat T_n$ against that of $|G|_\infty$ or $l^{-1/2}|\hat B_{r,l}|_\infty$, respectively.

The QQ plots for the three temporal dependence levels are given in Figures~\ref{fig:qq.beta2}, \ref{fig:qq.beta0.9}, and \ref{fig:qq.beta0.55} (Fig.~\ref{fig:qq.beta0.55} in Appendix \ref{app:null.exp.details}).
As a supplement, the corresponding Kolmogorov and Wasserstein-1 distances are given in Tables~\ref{tab:ks-distance-high-dimensional} and \ref{tab:w1-distance-high-dimensional} (both in Appendix \ref{app:null.exp.details}).
Generally, the shorter-memory, the larger $n$ and smaller $p$, the more accurate Gaussian approximation and block bootstrap are.
This agrees with the trend of finite-sample bounds in Theorems~\ref{thm:ga.direct} and \ref{thm:block.bootstrap.direct}.
This simulation setup confirms that they have reasonable empirical validity even when $p$ is a lot larger than $n$, and when the temporal dependence is long-range.
The case where $\beta=0.55$, which exceeds the range in Condition \ref{cond.A}, fails expectedly.

We further examine the effect of sparsity by fixing $n=800$ and $p=2000$ and varying $b\in\{2,4,6,8,12,16\}$ under $\beta=2$ and $0.9$.
The corresponding results in Appendix~\ref{app:null.exp.details} show that approximation accuracy and null calibration generally deteriorate at larger $s_0=2b+1$, especially under long memory.

To examine the behavior in the long memory regime, we also add $\beta=0.85$ and $0.8$.
We fix $b=4$ and $p=2000$, using $n=800$ and $1600$.
The results in Appendix~\ref{app:null.exp.details} show deteriorating bootstrap calibration near $\beta=3/4$ and improvement in these near-boundary settings when $n$ is doubled.

\begin{figure}[t]
    \centering
    \includegraphics[width=0.9\linewidth]{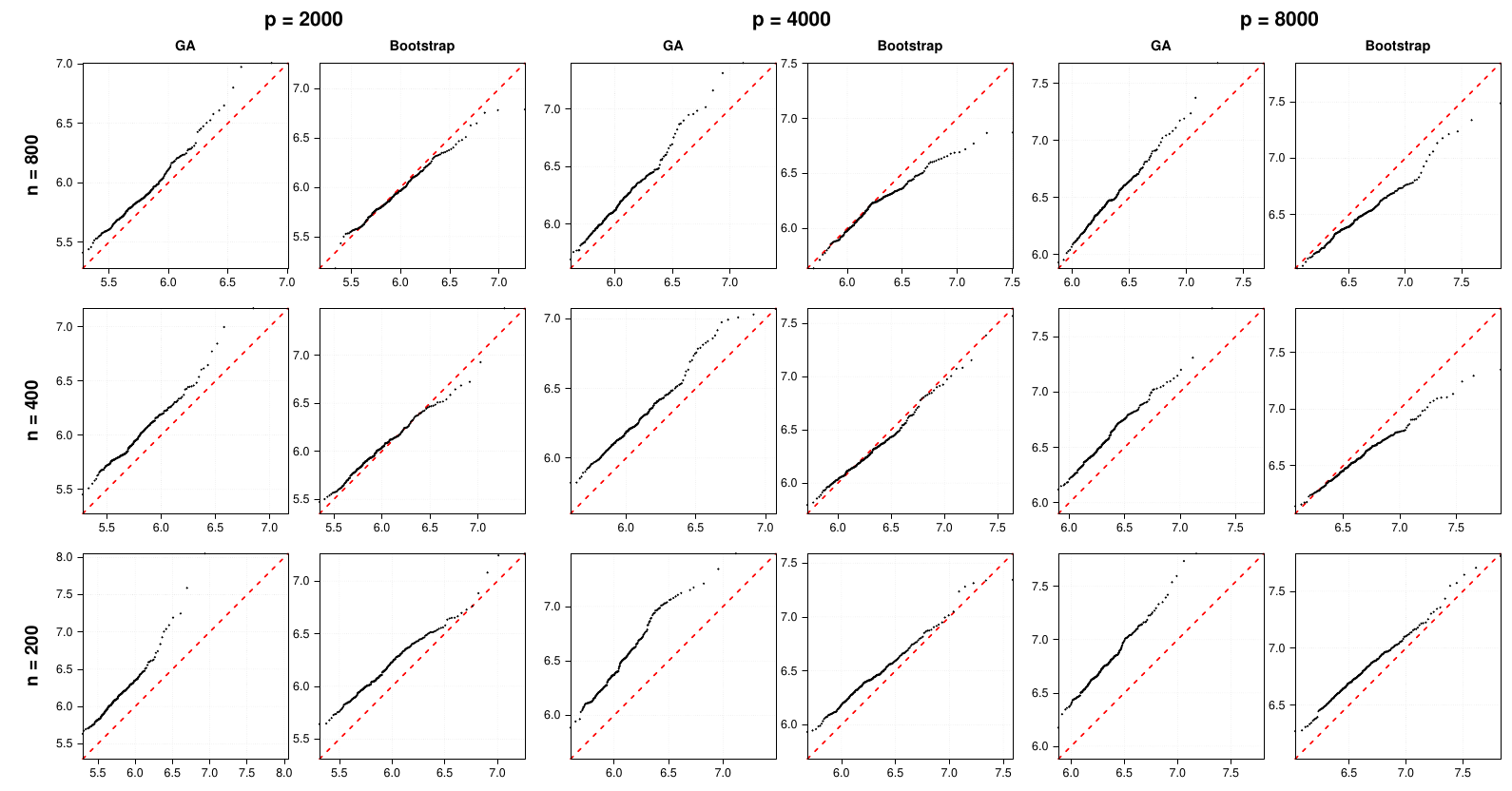}
    \caption{QQ-plots for null tests under short-memory, $\beta=2$.
    The $x$-axis is $|G|_\infty$ for GA and $l^{-1/2}|\hat B_{r,l}|_\infty$ for bootstrap.
    The $y$-axis is $\hat T_n$ for both.
    The red dashed line is $y=x$.
    Closer to this line indicates better proximity.}
    \label{fig:qq.beta2}
\end{figure}

\begin{figure}[t]
    \centering
    \includegraphics[width=0.9\linewidth]{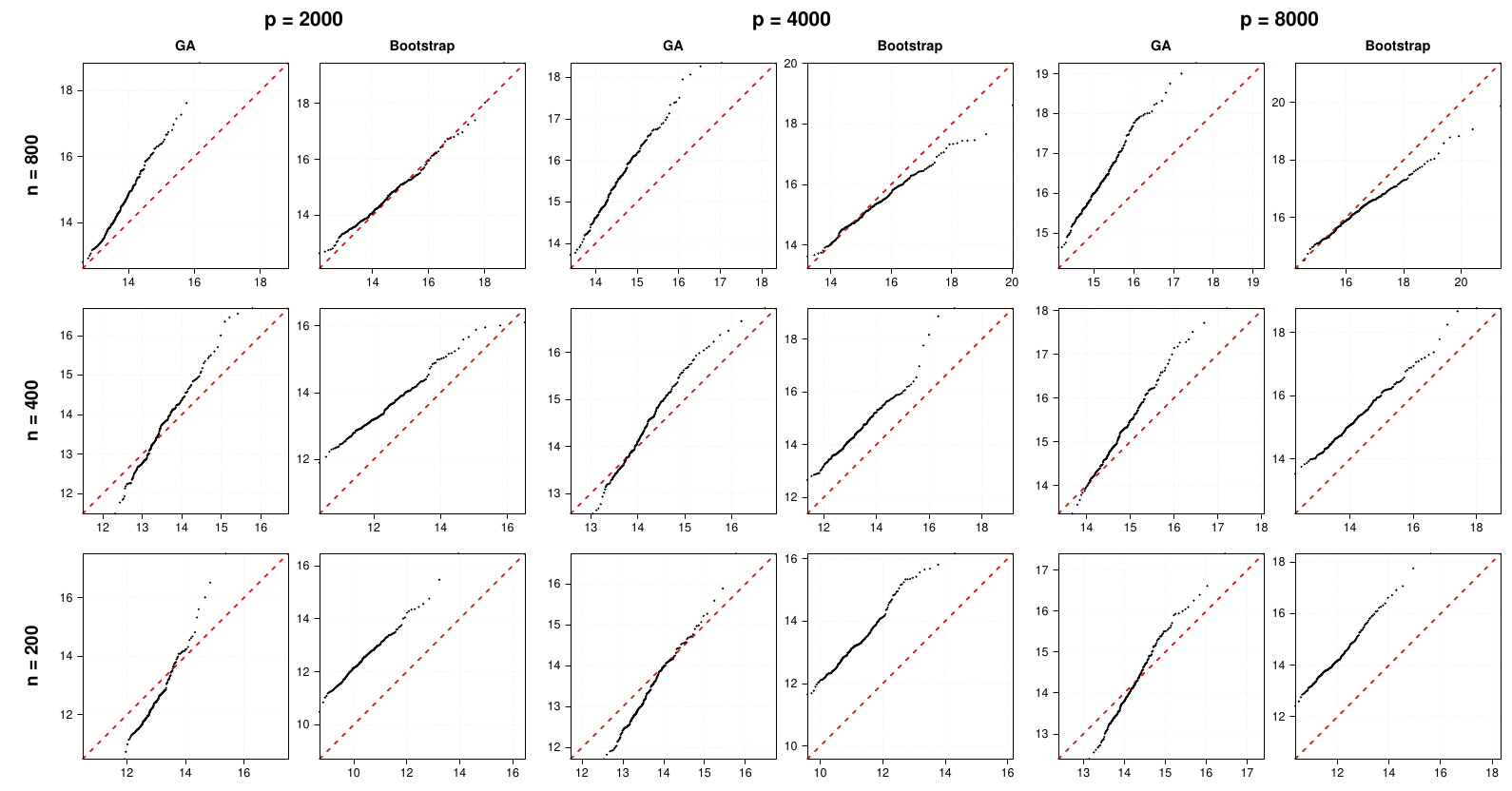}
    \caption{QQ-plots for null tests under long-memory, $\beta=0.9$.
    The $x$-axis is $|G|_\infty$ for GA and $l^{-1/2}|\hat B_{r,l}|_\infty$ for bootstrap.
    The $y$-axis is $\hat T_n$ for both.
    The red dashed line is $y=x$.
    Closer to this line indicates better proximity.}
    \label{fig:qq.beta0.9}
\end{figure}

\subsection{Alternative Test Power}\label{sec:alt.test.power}

We next study the power of the proposed graph test against random off-band alternatives.
For the null graph $E_0$, we apply the same banded structure and precision matrix setup as the null test experiment in Section \ref{sec:null.text.validity}.
For the true precision matrix $\Omega^\ast$ used to generate time series, however, we add $m$ off-support entries sampled uniformly without replacement from $\{(i,j):i<j,\ |i-j|>b\}$.
In this experiment, we consider the low-dimensional setup $n=2000, p=800$ and the high-dimensional setup $n=800, p=2000$.
For each setup, we consider the short-memory regime $\beta=2$ and the long-memory regime $\beta=0.9$.
The strength of the planted edges is calibrated by the effective strength, denoted by $\kappa = |\Omega^\ast_{ij}| / \sqrt{\Omega^\ast_{ii}\Omega^\ast_{jj}}$ for the added edges.
In this experiment, we let $m\in\{80,200,400,800,1200\}$ for the low-dimensional setup and $m\in\{200,500,1000,2000,3000\}$ for the high-dimensional setup, and $\kappa\in\{0.03,0.06,0.09,0.12,0.15,0.18\}$.
We test the null graph $E_0$ on each of the 200 instances of FFT-generated time series at the level $\alpha=0.1$, and report the rejection rate.

\begin{table}[t]
\centering
\caption{Empirical rejection rates under random off-band alternatives.}
\label{tab:kappa-power-four-panel}
\scriptsize
\setlength{\tabcolsep}{2.5pt}
\renewcommand{\arraystretch}{1.12}
\resizebox{\textwidth}{!}{%
\begin{tabular}{@{}cc@{}}
\multicolumn{2}{c}{\textbf{Low-dimensional regime: }$(n,p)=(2000,800)$} \\[0.6ex]

\begin{tabular}{c|cccccc}
\toprule
\multicolumn{7}{c}{$\beta=2$} \\
\midrule
$m\backslash \kappa$ & 0.03 & 0.06 & 0.09 & 0.12 & 0.15 & 0.18 \\
\midrule
80   & 0.135 & 0.375 & 0.995 & 1.000 & 1.000 & 1.000 \\
200  & 0.140 & 0.650 & 1.000 & 1.000 & 1.000 & 1.000 \\
400  & 0.145 & 0.840 & 1.000 & 1.000 & 1.000 & 1.000 \\
800  & 0.165 & 0.935 & 1.000 & 1.000 & 1.000 & 1.000 \\
1200 & 0.175 & 0.955 & 1.000 & 1.000 & 1.000 & 1.000 \\
\bottomrule
\end{tabular}
&
\begin{tabular}{c|cccccc}
\toprule
\multicolumn{7}{c}{$\beta=0.9$} \\
\midrule
$m\backslash \kappa$ & 0.03 & 0.06 & 0.09 & 0.12 & 0.15 & 0.18 \\
\midrule
80   & 0.090 & 0.100 & 0.085 & 0.100 & 0.195 & 0.335 \\
200  & 0.090 & 0.100 & 0.115 & 0.155 & 0.230 & 0.520 \\
400  & 0.095 & 0.120 & 0.145 & 0.210 & 0.390 & 0.770 \\
800  & 0.110 & 0.125 & 0.155 & 0.270 & 0.485 & 0.680 \\
1200 & 0.105 & 0.135 & 0.210 & 0.320 & 0.465 & 0.550 \\
\bottomrule
\end{tabular}
\\[3ex]

\multicolumn{2}{c}{\textbf{High-dimensional regime: }$(n,p)=(800,2000)$} \\[0.6ex]

\begin{tabular}{c|cccccc}
\toprule
\multicolumn{7}{c}{$\beta=2$} \\
\midrule
$m\backslash \kappa$ & 0.03 & 0.06 & 0.09 & 0.12 & 0.15 & 0.18 \\
\midrule
200  & 0.085 & 0.115 & 0.175 & 0.770 & 1.000 & 1.000 \\
500  & 0.100 & 0.140 & 0.355 & 0.945 & 1.000 & 1.000 \\
1000 & 0.085 & 0.105 & 0.475 & 0.970 & 1.000 & 1.000 \\
2000 & 0.050 & 0.120 & 0.585 & 0.990 & 1.000 & 1.000 \\
3000 & 0.070 & 0.175 & 0.665 & 0.965 & 0.990 & 1.000 \\
\bottomrule
\end{tabular}
&
\begin{tabular}{c|cccccc}
\toprule
\multicolumn{7}{c}{$\beta=0.9$} \\
\midrule
$m\backslash \kappa$ & 0.03 & 0.06 & 0.09 & 0.12 & 0.15 & 0.18 \\
\midrule
200  & 0.075 & 0.090 & 0.090 & 0.110 & 0.110 & 0.125 \\
500  & 0.090 & 0.085 & 0.105 & 0.145 & 0.160 & 0.165 \\
1000 & 0.095 & 0.095 & 0.130 & 0.160 & 0.180 & 0.195 \\
2000 & 0.070 & 0.100 & 0.130 & 0.160 & 0.190 & 0.205 \\
3000 & 0.085 & 0.105 & 0.160 & 0.215 & 0.260 & 0.295 \\
\bottomrule
\end{tabular}
\end{tabular}%
}
\end{table}

Table~\ref{tab:kappa-power-four-panel} summarizes the resulting rejection probabilities.
The weak-signal columns stay close to the nominal level, while rejection rates generally increase as the effective edge strength $\kappa$ grows, or as the number of extra edges $m$ grows.
A larger number of stronger extra edges are usually more easily detected.
In the short-memory case, the power converges to one, as described in theory.
This convergence is slower for longer memory or higher dimension, as explained by the slower rates in Theorem~\ref{thm:block.bootstrap.direct}.

We show another simulation that demonstrates the power of consistency-empowered statistic $\tilde T_n$.
We apply the same banded precision matrix structure $\Omega_0$ with $b=4$, but with $\theta_0 = 1$, and $\rho=1$.
This setup is to prevent the precision entries from being too small such that the assumption in Theorem \ref{thm:ce.direct.formal} is violated.
For the true precision matrix $\Omega^\ast$ used to generate time series, we randomly select $\eta$ of all entries in $\Omega_0$, and these entries are weakened by the ratio of $\chi$.
The experiment is on the same $(n,p,\beta)$ combinations.
The tuning parameters $\delta_n$ and $C_\Delta$ are altered in a grid on a heldout copy of time series, until the null test rejection rate approach the nominal size $\alpha=0.10$.
For our setup, the choice is $(\delta_n, C_\Delta) = (0.05, 0.5)$, with $C_n=C_\Delta\sqrt{\log p}$.
As a side-by-side comparison, the tests are also run under the direct statistic $\hat T_n$ under each setup.
The results are in Table~\ref{tab:ce.power}.

The consistency-empowered tests are indeed recovering the edges in $\Omega_0$ that are missing in $\Omega^\ast$, compared to the direct tests.
This phenomenon is more prominent when the temporal dependence is weaker and the dimension is lower.
For the short-memory, low-dimensional case, in particular, we can see the rejection rate rapidly converging to one as $\eta$ and $\chi$ increases.
The rejection rate for the short-memory, high-dimensional case is also growing, but not as rapidly.
In both cases, the consistency-empowered rejection rates are significantly higher than their direct test counterparts.
This demonstrates the recovered power of $\tilde T_n$ as $E \subsetneq E_0$.
In the long-memory case, however, the growth is significantly slower, and the advantage upon direct test is far less significant.
This is unsurprising, since the rate in Theorem~\ref{thm:ce.direct.formal}, which directly depends on $\mathfrak R_{n,l}(\epsilon)$, converges significantly slower as the temporal dependence becomes stronger.

\begin{table}[t]
\centering
\caption{Empirical rejection rates under within-band weakening alternatives.}
\label{tab:ce.power}
\setlength{\tabcolsep}{6pt}
\renewcommand{\arraystretch}{1.25}
\begin{tabular}{@{}c*{4}{rr}@{}}
\toprule
\multirow{3}{*}{$(\eta,\chi)$}
& \multicolumn{4}{c}{$n=2000,\ p=800$}
& \multicolumn{4}{c}{$n=800,\ p=2000$}
\\
\cmidrule(lr){2-5}
\cmidrule(lr){6-9}
&
\multicolumn{2}{c}{$\beta=2$}
&
\multicolumn{2}{c}{$\beta=0.9$}
&
\multicolumn{2}{c}{$\beta=2$}
&
\multicolumn{2}{c}{$\beta=0.9$}
\\
\cmidrule(lr){2-3}
\cmidrule(lr){4-5}
\cmidrule(lr){6-7}
\cmidrule(lr){8-9}
&
$\tilde T_n$ & $\hat T_n$
&
$\tilde T_n$ & $\hat T_n$
&
$\tilde T_n$ & $\hat T_n$
&
$\tilde T_n$ & $\hat T_n$
\\
\midrule
$(0.0,0.0)$
& {\bf 0.140} & 0.130
& {\bf 0.080} & 0.070
& {\bf 0.070} & 0.070
& {\bf 0.110} & 0.090
\\
$(0.1,0.5)$
& {\bf 0.835} & 0.135
& {\bf 0.140} & 0.125
& {\bf 0.130} & 0.125
& {\bf 0.120} & 0.090
\\
$(0.3,1.0)$
& {\bf 1.000} & 0.165
& {\bf 0.100} & 0.085
& {\bf 0.270} & 0.115
& {\bf 0.110} & 0.085
\\
$(0.5,1.0)$
& {\bf 1.000} & 0.150
& {\bf 0.100} & 0.100
& {\bf 0.350} & 0.110
& {\bf 0.120} & 0.100
\\
\bottomrule
\end{tabular}
\end{table}

\section{Sherlock fMRI Data Analysis}\label{sec:real.data}

We apply our proposed test on the Sherlock fMRI dataset \citep{chen2017shared}.
Each of the 17 subjects were observed in two phases.
In the movie phase, the subject watched the same movie as a stimulus.
In the recall phase, the subject verbally recalled the movie content.
The fMRI scan in each phase was transformed to multivariate time series over 268 cortical regions of interest (ROI, \citet{shen2013groupwise}).
The scientific question is whether the conditional independence graph structure during the movie phase is sufficient to explain the recall phase.
This comparison is natural because movie watching and recall involve related cognitive content, but may entail different patterns of brain connectivity \citep{baldassano2017discovering}.
Our methodology is useful for this type of data, since fMRI time series typically exhibit long memory (see Appendix \ref{app:long.memory}).
For each subject, the movie phase contains 1976 time points, while the recall phase has subject-specific length ranging from 437 to 1824 time points.

The goal of the experiment is a directional two-sample test on the conditional-independence graph between the movie phase and the recall phase.
We first estimate the movie-phase graph as follows.
For each subject, compute $\hat\Sigma^{\mathrm{mov}}$, the movie-phase sample covariance matrix, and let $\hat\Omega^{\mathrm{mov}} = (\hat\Sigma^{\mathrm{mov}})^{-1}$.
For threshold $\nu > 0$, which is a tuning parameter, let the movie graph be $E_\nu^{\mathrm{mov}} = \{(i,j): i < j,\ |\hat\Omega^{\mathrm{mov}}_{ij}| > \nu\}$.
We select the largest value of $\nu$ such that the graph passes the movie self-diagnostic on the individual's movie-phase data at level $\alpha=0.05$.
Equivalently, this selects the sparsest graph in the threshold path whose movie self-test p-value is at least 0.05.

We then test the recall phase against this movie-phase graph $E_{\hat\nu}^{\mathrm{mov}}$ using $\hat T_n$.
A rejection indicates that the recall graph contains off-support conditional dependencies not contained in the selected movie graph, signaling extra brain connectivity during the recall phase.
The recall p-values are computed by the block-bootstrap method, and the 17 subject-level recall p-values are adjusted using the Holm procedure to control family-wise error rate.
Under the conditions of Corollary~\ref{cor:two.sample.direct}, support-consistent graph selection preserves asymptotic validity of the recall p-values; see Appendix~\ref{sec:two.sample.direct}.
For each selected graph, we report the average column-wise support size $\bar s$ and the quartiles of node degrees.
The result is available in Table \ref{tab:movie_recall_holm_subjects} in Appendix \ref{app:sherlock}.

All 17 selected movie graphs pass the self-diagnostic at level $\alpha=0.05$.
After Holm correction across subjects, all 17 recall tests reject at level $\alpha=0.05$.
Thus, the recall phase exhibits conditional-dependence structure that is not contained in the corresponding selected movie-phase graph for every subject in the cohort.
This result provides evidence of a structural difference of brain connectivity between watching a movie (receiving information) and recalling it (utilizing memory and describing), relative to the selected movie graphs.

\section{Concluding Remark}\label{sec:conclusion}

This paper develops a framework for testing Gaussian graphical model structure in high-dimensional long-memory time series.
The key contribution is to reduce the graph specification testing problem to simultaneous inference of covariance through local-inverse based on the sample covariance matrix.
This avoids full precision-matrix estimation and enables finite-sample Gaussian approximation and block-bootstrap calibration.
The theory relies on temporal dependence $\beta>3/4$.
For ultra-long-memory time series with $1/2<\beta<3/4$, the asymptotic distribution can be Rosenblatt-type, and is beyond the scope of this paper.
Another limitation is that the one-sample test requires a predetermined graphical structure to be known.
Entrywise testing requires inference theory on the precision matrix.
Doing so under high dimension is a very interesting yet challenging aspect of study, and we leave it for future work.

\newpage

\bibliographystyle{apalike-forthcoming}
\bibliography{TGG_papers}       

\newpage
\appendix

\section{Miscellaneous Contents}
\subsection{Notation}
Throughout the paper, $|\cdot|_\infty$ denotes the maximum absolute entry of a vector or a matrix, and $|\cdot|_1$ denotes the matrix $\ell_1$ norm.
For a vector $v$, $|v|$ denotes its Euclidean norm.
For a matrix $A$, $A_{j\cdot}$ denotes its $j$-th row, $A_{\cdot j}$ denotes its $j$-th column, and $A^\top$ denotes its transpose.
We write $[p]=\{1,\ldots,p\}$, and let $e_j$ be the $j$-th canonical basis vector in $\mathbb R^p$.
For a matrix $A$, $\mathrm{supp}(A)=\{(i,j):A_{ij}\ne 0\}$ denotes its entrywise support.
For a random variable $X$, $\|X\|_q$ is defined as $(\E|X|^q)^{1/q}$.
For a set $S$, $|S|$ denotes its cardinality.
We use $\mathbf 1\{\cdot\}$ for the indicator function.
For two positive sequences $a_n$ and $b_n$, $a_n\lesssim b_n$ means $a_n\le Cb_n$ for some constant $C>0$, and $a_n\asymp b_n$ means $a_n\lesssim b_n$ and $b_n\lesssim a_n$.
We use $a\vee b=\max\{a,b\}$ and $a\wedge b=\min\{a,b\}$.
The notation $O_{\mathbb P}(\cdot)$ and $o_{\mathbb P}(\cdot)$ has its usual meaning.
Unless otherwise stated, constants may change from line to line and do not depend on $n$ or $p$.

\subsection{Algorithms}\label{app:algo}
This paper introduces two types of goodness-of-fit test statistics for Gaussian graphical models of time series.
The first one is the direct statistic $\hat T_n$ as discussed in Section \ref{sec:direct.stat}, while the second one is the consistency-empowered corrected statistic $\tilde T_n$ as discussed in Section \ref{sec:ce.direct}.

Algorithm \ref{alg:block.bootstrap} describes the detailed algorithm that computes $\hat T_n$, and how to compute the rejection threshold $\hat q_{1-\alpha}$ using block bootstrap method.
Meanwhile, Algorithm \ref{alg:block.bootstrap.ce} describes the detailed algorithm that computes $\tilde T_n$, and uses the same rejection threshold $\hat q_{1-\alpha}$ computed by the same block bootstrap method.

\begin{algorithm}[t]
\caption{Test procedure for the consistency-empowered statistic}
\label{alg:block.bootstrap.ce}
\begin{algorithmic}[1]
\Require $\{X_t\}_{t=1}^n$, null graph $E_0$, block length $l$, level $\alpha$, and tuning parameters $C_n,\delta_n$
\State Compute sample covariance matrix $\hat\Sigma_n$ and the local inverse estimators $\hat{\bm \omega}_{0,j}$ from \eqref{eq:omega.hat}

\State For each $j$ and $i\in S_{0,j}$, compute $\hat{\bm v}_{0,ij} = B_{0,j} (B_{0,j}^{\top}\hat\Sigma_n B_{0,j})^{-1} B_{0,j}^{\top}e_i$ and the plug-in estimated local inverse score process $\hat{\mathcal X}^{\omega}_{t,ij}$ from \eqref{eq:local.inverse.score}

\State Compute its block-bootstrap estimated long-run variance, $\hat L^\omega_{ij}$, from \eqref{eq:tau.hat.block}, and set $\hat\tau_{ij}=\sqrt{\hat L^\omega_{ij}/n}$

\State Define $\hat\Delta_{ij} = C_n \one{ |\hat\omega_{0,ij}|/\hat\tau_{ij}\le \delta_n }$ for $i\in S_{0,j}$,
and set $\hat\Delta_{ij}=0$ for $i\notin S_{0,j}$.

\State Let $\hat{\bm \Delta}_j=(\hat\Delta_{1j},\ldots,\hat\Delta_{pj})^\top$ and compute $\tilde{\bm \omega}_{0,j} = \hat{\bm \omega}_{0,j} + \hat{\bm \Delta}_j$

\State Compute $\tilde T_n$ from \eqref{eq:ce.stat}

\State Compute the plug-in score vectors $\hat{\mathcal X}_t$ from \eqref{eq:X.hat.def} for $t = 1,\ldots,n$
\State Form the overlapping block sums $\hat B_{r,l}$ for $r = l,\ldots,n$
\State Construct the empirical CDF $\hat F_{n,l}$ from \eqref{eq:F.hat.def} and the quantile $\hat q_{1-\alpha}$ from \eqref{eq:q.hat.def}
\State Reject $H_0$ if $\tilde T_n > \hat q_{1-\alpha}$
\end{algorithmic}
\end{algorithm}

\section{Supplemental Theory}\label{sec:supp.theory}

\subsection{Detailed Discussion on Two-sample Extension}
\label{sec:two.sample.direct}

We first record a direct two-sample consequence of Theorems~\ref{thm:ga.direct} and~\ref{thm:block.bootstrap.direct}.
Let $\{ X_t^{(1)} \}_{t=1}^{n_1}$ and $\{ X_t^{(2)} \}_{t=1}^{n_2}$ be two possibly dependent processes, each marginally satisfying the stationary Gaussian linear-process setup of Section~\ref{sec:methodology}.
For $k = 1,2$, the true underlying covariance and precision matrices are $\Sigma^{(k)} = \E[ X_t^{(k)} X_t^{(k)\top} ]$ and $\Omega^{(k)} = (\Sigma^{(k)})^{-1}$, with the graph $E^{(k)} = \mathrm{supp}(\Omega^{(k)})$.
We also have each sample covariance $\hat\Sigma^{(k)}$.
Let $\hat E^{(1)}$ be a graph estimator based only on the first sample.
We use $\hat E^{(1)}$ as the null graph for the second sample and test
\[
H_0^{1\rightarrow 2} : E^{(2)} \subseteq E^{(1)},
\qquad
H_1^{1\rightarrow 2} : E^{(2)} \nsubseteq E^{(1)}.
\]
This construction is a directional support-containment test because the second sample is tested against the first-sample graph.
Let $S_{0,j}^{(1)}$ and $B_{0,j}^{(1)}$ denote the support set and selection matrix constructed from the fixed population graph $E^{(1)}$.
Under $H_0^{1\rightarrow 2}$, the $j$th column $\bm\omega_j^{(2)}$ of $\Omega^{(2)}$ is supported on $S_{0,j}^{(1)}$, so
\[
B_{0,j}^{(1)}
\bigl(B_{0,j}^{(1)\top}\Sigma^{(2)}B_{0,j}^{(1)}\bigr)^{-1}
B_{0,j}^{(1)\top}e_j
=\bm\omega_j^{(2)}.
\]
Thus the fixed-reference population local inverses and graph scores use $\Sigma^{(2)}$ and the reference graph $E^{(1)}$.
The direct-statistic proofs of Theorems~\ref{thm:ga.direct} and~\ref{thm:block.bootstrap.direct} use only this support containment and the resulting off-support moment identities, so their bounds extend from graph equality to this containment null under the same matrix-norm, local-inverse stability, and selected-score nondegeneracy conditions.

For each $1 \le j \le p$, define $\hat S_{0,j}^{(1)} = \{ i : i = j \text{ or } (i,j) \in \hat E^{(1)} \}$, $\hat s_{0,j}^{(1)} = |\hat S_{0,j}^{(1)}|$ and $\hat s_0^{(1)} = \max_{1 \le j \le p} \hat s_{0,j}^{(1)}$.
Let $\hat B_{0,j}^{(1)} \in \R^{p \times \hat s_{0,j}^{(1)}}$ be the corresponding selection matrix.
Define
\[
\hat{\bm \omega}_{0,j}^{(1)}
=
\hat B_{0,j}^{(1)}
(\hat B_{0,j}^{(1)\top} \hat\Sigma^{(2)} \hat B_{0,j}^{(1)})^{-1}
\hat B_{0,j}^{(1)\top} e_j.
\]
Let $\hat{\mathcal J}^{(1)} = \{ (i,j) : i \notin \hat S_{0,j}^{(1)} \}$.
The directional two-sample statistic is
\[
\hat T^{1\rightarrow 2}
=
\sqrt{n_2}
\max_{(i,j) \in \hat{\mathcal J}^{(1)}}
| e_i^\top \hat\Sigma^{(2)} \hat{\bm \omega}_{0,j}^{(1)} |.
\]

The bootstrap distribution is formed as in Section~\ref{sec:block.bootstrap}, but using the second sample and the random null graph $\hat E^{(1)}$.
For $(i,j) \in \hat{\mathcal J}^{(1)}$, define $\hat{\bm u}_{ij}^{(1)} = e_i - \hat B_{0,j}^{(1)} (\hat B_{0,j}^{(1)\top} \hat\Sigma^{(2)} \hat B_{0,j}^{(1)})^{-1} \hat B_{0,j}^{(1)\top} \hat\Sigma^{(2)} e_i$.
For $1 \le t \le n_2$, define $\hat{\mathcal X}_{t,ij}^{(1)} = \hat{\bm u}_{ij}^{(1)\top} ( X_t^{(2)} X_t^{(2)\top} - \hat\Sigma^{(2)} ) \hat{\bm \omega}_{0,j}^{(1)}$ and $\hat{\mathcal X}_t^{(1)} = (\hat{\mathcal X}_{t,ij}^{(1)})_{(i,j) \in \hat{\mathcal J}^{(1)}}$.
For a block length $l = l_{n_2}$, let $\hat B_{r,l}^{(1)} = \sum_{t=r-l+1}^{r} \hat{\mathcal X}_t^{(1)}$ for $r = l,\ldots,n_2$.
The directional bootstrap CDF and critical value are
\[
\hat F_l^{1\rightarrow 2}(u)
=
\frac{1}{n_2-l+1}
\sum_{r=l}^{n_2}
\one{ l^{-1/2} |\hat B_{r,l}^{(1)}|_\infty \le u },
\quad
\hat q_{1-\alpha}^{1\rightarrow 2}
=
\inf
\{ u \in \R : \hat F_l^{1\rightarrow 2}(u) \ge 1-\alpha \}.
\]
The directional two-sample test rejects $H_0^{1\rightarrow 2}$ whenever $\hat T^{1\rightarrow 2} > \hat q_{1-\alpha}^{1\rightarrow 2}$.
Its bootstrap $p$-value is $p_{1\rightarrow 2}=1-\hat F_l^{1\rightarrow 2}(\hat T^{1\rightarrow 2}-)$, where $F(x-)$ denotes the left limit of $F$ at $x$.
If the selected off-support set is empty, its maximum is defined as zero and the $p$-value is one.
Let $s_0^{(1)} = \max_{1 \le j \le p} | \{ i : i = j \text{ or } (i,j) \in E^{(1)} \} |$ be the sparsity level of $E^{(1)}$.
Let $G^{1\rightarrow 2}$ be the centered Gaussian vector obtained by applying Theorem~\ref{thm:ga.direct} to the second sample with null graph $E^{(1)}$.
The following corollary of Theorems \ref{thm:ga.direct} and \ref{thm:block.bootstrap.direct} is the rigorous version of Corollary \ref{cor:two.sample.direct.informal}.
\begin{corollary}[Directional two-sample extension]
\label{cor:two.sample.direct}
Let $\pi_{n_1} = \P(\hat E^{(1)} \ne E^{(1)})$ and assume $\pi_{n_1} \to 0$.
Under $H_0^{1\rightarrow 2}$, assume that the second sample satisfies the assumptions of Theorems~\ref{thm:ga.direct} and~\ref{thm:block.bootstrap.direct} with the fixed reference graph $E^{(1)}$, using the containment interpretation above and a nonempty fixed off-support set.
Then
\begin{enumerate}
    \item The following Gaussian approximation bound holds:
    \begin{multline*}
    \sup_{u \in \R}
    |
    \P_{H_0^{1\rightarrow 2}}( \hat T^{1\rightarrow 2} \le u )
    -
    \P( |G^{1\rightarrow 2}|_\infty \le u )
    |
    \\
    \le
    C \pi_{n_1}
    +
    C \Psi(p,n_2)
    +
    C p^2 \exp\{ - c n_2 / (s_0^{(1)})^2 \}
    +
    C \frac{(s_0^{(1)})^2 \log^{3/2}(p n_2)}{\sqrt{n_2}}.
    \end{multline*}

    \item For any fixed $0<\epsilon<1$, take $\phi = \frac{2\tilde\beta+4}{(3-\epsilon)\tilde\beta+4}$, $\psi = \frac{(5\tilde\beta+12)(1-\epsilon)}{(3-\epsilon)\tilde\beta+4}$, and $l \asymp n_2^\phi \log^\psi(p)$.
    Whenever $\mathfrak R_{n,l}^{1\rightarrow 2}(\epsilon) \le 1$, we have
    \[
    \P
    (
    \sup_{u \in \R}
    |
    \hat F_l^{1\rightarrow 2}(u)
    -
    \P_{H_0^{1\rightarrow 2}}( \hat T^{1\rightarrow 2} \le u )
    |
    >
    C \mathfrak R_{n,l}^{1\rightarrow 2}(\epsilon)^{1-\epsilon}
    )
    \le
    C \mathfrak R_{n,l}^{1\rightarrow 2}(\epsilon)^\epsilon,
    \]
    where $\Psi_B(p,n,\epsilon)=\log^{(5\tilde\beta+12)/[(6-2\epsilon)\tilde\beta+8]}(pn)/n^{\tilde\beta/[(6-2\epsilon)\tilde\beta+8]}$ is the first term of the rate in Theorem~\ref{thm:block.bootstrap.direct}, and
    \begin{multline*}
    \mathfrak R_{n,l}^{1\rightarrow 2}(\epsilon)
    ={}
    \pi_{n_1}
    +
    \Psi_B(p,n_2,\epsilon)
    +
    \Psi(p,n_2)
    +
    p^2 \exp\{ - c n_2 / (s_0^{(1)})^2 \}
    +
    (p n_2)^{-1}
    \\
    +
    \frac{(s_0^{(1)})^2 \log^{3/2}(p n_2)}{\sqrt{n_2}}
    +
    s_0^{(1)}
    \sqrt{\frac{l}{n_2}}
    \log^2(p n_2).
    \end{multline*}
\end{enumerate}
In particular, if $\mathfrak R_{n,l}^{1\rightarrow 2}(\epsilon) \to 0$, and the distribution of $|G^{1\rightarrow 2}|_\infty$ is continuous at its $(1-\alpha)$-quantile, then we have
\[
\P_{H_0^{1\rightarrow 2}}
(
\hat T^{1\rightarrow 2} > \hat q_{1-\alpha}^{1\rightarrow 2}
)
\to
\alpha.
\]
The proof of Corollary \ref{cor:two.sample.direct} is given in Appendix \ref{app:proof.two.sample.direct}.
\end{corollary}

Construct $p_{2\rightarrow 1}$ by interchanging the two samples, and define
\[
p_{\mathrm{sym}}=\min\{1,2\min(p_{1\rightarrow 2},p_{2\rightarrow 1})\}.
\]
Since $E^{(1)}=E^{(2)}$ is equivalent to containment in both directions, this tests graph equality by rejecting whenever either directional test rejects at level $\alpha/2$.
Suppose that $\pi_{n_k}=\P(\hat E^{(k)}\ne E^{(k)})\to0$ for $k=1,2$ and, under the equality null, each tested sample marginally satisfies the assumptions of Corollary~\ref{cor:two.sample.direct} with the other population graph as reference, with the corresponding oracle error rates tending to zero.
Then, for each fixed $0<\alpha<1$,
\[
\P_{E^{(1)}=E^{(2)}}(p_{\mathrm{sym}}\le\alpha)
\le\alpha+\pi_{n_1}+\pi_{n_2}+o(1)
=\alpha+o(1).
\]
The reference choice may affect finite-sample power and selection error; the symmetric test treats both directions equally.
\begin{proof}
For each direction $k\rightarrow h$, let $p_{k\rightarrow h}^{\circ}$ use the fixed reference graph $E^{(k)}$ and the same sample $h$.
On $\{\hat E^{(k)}=E^{(k)}\}$ the selected-graph and oracle procedures coincide pathwise, so
\[
\P(p_{k\rightarrow h}\le\alpha/2)
\le\P(p_{k\rightarrow h}^{\circ}\le\alpha/2)+\pi_{n_k}
\le\alpha/2+\pi_{n_k}+o(1).
\]
The union bound over the two directions proves the claim without requiring independence between the samples or between their directional $p$-values.
\end{proof}

\subsection{Detailed Discussion on Consistency-empowered Estimator}
\label{app:ce.direct.details}

This subsection gives the technical details behind the consistency-empowered correction in Sections~\ref{sec:ce.direct} and \ref{sec:thm.ext}.
The new object is the standard error used to decide whether an inside-support coordinate $\hat\omega_{0,ij}$ is stochastically close to zero.
Because the observations are temporally dependent unlike in \citet{le2023universally}, this standard error must be defined through a long-run variance.
Note that the local inverse precision estimation error can be written as
\[
\hat\omega_{0,ij}-\omega_{0,ij} = -\frac{1}{n} \sum_{t=1}^n \mathcal X^\omega_{t,ij} + R^\omega_{n,ij}.
\]
Here, the local-inverse score process is defined as
\begin{equation*}
\mathcal X^\omega_{t,ij} = \bm v_{0,ij}^{\top} ( X_tX_t^{\top}-\Sigma ) \bm \omega_{0,j}, \qquad i\in S_{0,j},
\end{equation*}
where $\bm v_{0,ij}=B_{0,j}(B_{0,j}^{\top}\Sigma B_{0,j})^{-1}B_{0,j}^{\top}e_i$.
Note that the residual $R^\omega_{n,ij}$ is of smaller order than the leading linear term on the same local stability event used in Proposition \ref{prop:deterministic.reduction}, using the same proof technique.
Therefore, the stochastic scale of $\hat\omega_{0,ij}$ is governed by the long-run variance of $\mathcal X^\omega_{t,ij}$, which is denoted by
\[
L^\omega_{ij}
=
\sum_{h=-\infty}^{\infty}
\cov
(
\mathcal X^\omega_{0,ij},
\mathcal X^\omega_{h,ij}
).
\]
We refer to $\tau_{ij} = \sqrt{L_{ij}^\omega / n}$ as the long-run standard error of $\hat\omega_{0,ij}$.
The estimation of this quantity is directly useful in the implementation of the method; see Section \ref{sec:ce.direct}.
The long-run variance $L^\omega_{ij}$ is useful in the subsequent proof.
\begin{condition}\label{cond:long-run.nondegeneracy}
    There exist constants $0<c_\omega<C_\omega<\infty$ such that $c_\omega \le L^\omega_{ij}\le C_\omega$ for all $1\le j\le p$ and all $i\in S_{0,j}$.
\end{condition}
This condition plays the same role as the score nondegeneracy condition for the direct statistic.
It prevents the standard error used in the consistency-empowered screening rule from degenerating.

We estimate $L^\omega_{ij}$ by the same overlapping-block principle used in the bootstrap construction.
For $i\in S_{0,j}$, define $\hat{\bm v}_{0,ij}=B_{0,j}(B_{0,j}^{\top}\hat\Sigma_n B_{0,j})^{-1}B_{0,j}^{\top}e_i$.
The plug-in local-inverse estimator is
\begin{equation}\label{eq:local.inverse.score}
    \hat{\mathcal X}^\omega_{t,ij} = \hat{\bm v}_{0,ij}^{\top} ( X_tX_t^{\top}-\hat\Sigma_n ) \hat{\bm \omega}_{0,j}.
\end{equation}
Using the same block length $l$ as in the buildup of Theorem~\ref{thm:block.bootstrap.direct}, define the block-bootstrap estimator of the long-run variance by
\begin{equation}
\label{eq:tau.hat.block}
\hat L^\omega_{ij}
=
\frac{1}{n-l+1}
\sum_{r=l}^n
\left(
l^{-1/2}
\sum_{t=r-l+1}^r
\hat{\mathcal X}^\omega_{t,ij}
\right)^2,
\end{equation}
and $\hat\tau_{ij} = \sqrt{\hat L^\omega_{ij} / n}$.
This estimator is scalar, yet it is constructed from the same block-sampling idea similar to that of $\hat F_{n,l}$ in \eqref{eq:F.hat.def}.
The following lemma verifies that this estimator is uniformly consistent for our Gaussian linear process setup.
The proof is deferred to Appendix~\ref{app:proof.tau.uniform.consistency}.

\begin{lemma}
\label{lem:tau.uniform.consistency}
Assume Conditions~\ref{cond.A}-\ref{cond:long-run.nondegeneracy}.
Let $0<\epsilon<1$, choose the block length $l$ as in Theorem~\ref{thm:block.bootstrap.direct}, and assume $\mathfrak R_{n,l}(\epsilon)\to0$.
Then
\[
\max_{1\le j\le p}
\max_{i\in S_{0,j}}
\left|
\frac{\hat\tau_{ij}}{\tau_{ij}}-1
\right|
=
o_{\P}(1).
\]
\end{lemma}

The following lemma records the screening behavior of the consistency-empowered correction.
It is the time-series analogue of the zero-versus-nonzero separation used in \citet{le2023universally}.
The proof is deferred to Appendix~\ref{app:proof.local.inverse.screening}.

\begin{lemma}
\label{lem:local.inverse.screening}
Assume Conditions~\ref{cond.A}-\ref{cond:long-run.nondegeneracy}.
Assume also the block-length and rate conditions of Lemma~\ref{lem:tau.uniform.consistency}.
Let $\delta_n=\sqrt{\log n}$.
Under $H_0:E=E_0$, assume that the nonzero coordinates $\omega_{0,ij}$ with $i\in S_{0,j}$ are uniformly bounded away from zero.
Then
\[
\P_{H_0}
\left(
\hat{\bm \Delta}_j=0
\text{ for all }
1\le j\le p
\right)
\to1.
\]
Under every fixed included alternative $E\subsetneq E_0$,
\[
\P
\left(
\hat{\bm \Delta}_j\ne0
\text{ for some }
1\le j\le p
\right)
\to1.
\]
\end{lemma}

The screening conclusions also hold for any sequence satisfying $\delta_n\to\infty$ and $\delta_n+\sqrt{\log p}=o(\sqrt n)$, with the block length and rate conditions of Lemma~\ref{lem:tau.uniform.consistency}; see Appendix~\ref{app:proof.local.inverse.screening}.

We can now state the formal version of Theorem~\ref{thm:ce.direct}.
The theorem uses the fixed consistency-empowerment levels $\delta_n=\sqrt{\log n}$ and $C_n=C_\Delta\sqrt{\log p}$ with any fixed $C_\Delta>0$.
It ensures that under certain conditions, the consistency-empowered statistic is equal to the direct test statistic in probability, yet unlike the latter, it also ensures that the test power converge to one under fixed alternatives.
We have redirected the proof of this theorem to Appendix \ref{app:proof.ce.direct.formal}.

\begin{theorem}
\label{thm:ce.direct.formal}
Assume Conditions~\ref{cond.A}-\ref{cond:long-run.nondegeneracy} under $H_0$.
Let $0<\epsilon<1$, choose the block length $l$ as in Theorem~\ref{thm:block.bootstrap.direct}, and assume $\mathfrak R_{n,l}(\epsilon)\to0$.
Set $\delta_n=\sqrt{\log n}$ and $C_n=C_\Delta\sqrt{\log p}$ for any fixed constant $C_\Delta>0$.
\begin{enumerate}
    \item Under the null hypothesis $H_0:E=E_0$, assume additionally that the nonzero coordinates $\omega_{0,ij}$ with $i\in S_{0,j}$ are uniformly bounded away from zero by a fixed constant.
    Then
    \[
    \P_{H_0} (\tilde T_n=\hat T_n) \to 1.
    \]
    Consequently, Theorems~\ref{thm:ga.direct} and~\ref{thm:block.bootstrap.direct} remain valid with $\tilde T_n$ in place of $\hat T_n$, and we have $\P_{H_0} ( \tilde T_n>\hat q_{1-\alpha} ) \to \alpha$.
    \item Furthermore, under any fixed alternative $H_1:E\ne E_0$ for which Conditions \ref{cond.A}-\ref{cond:local.inverse} and Condition \ref{cond:long-run.nondegeneracy} hold with the alternative population covariance and precision matrices, and for which $\mathfrak R_{n,l}(\epsilon)\to0$, we have the asymptotic guarantee on test power,
    \[
    \P_{H_1}(\tilde T_n>\hat q_{1-\alpha})\to1.
    \]
\end{enumerate}
\end{theorem}

\begin{remark}
The signal condition is only needed to prevent true nonzero precision coordinates from being mistakenly pumped under the null.
If, alternatively, the nonzero entries are allowed to shrink with $n$, then an explicit beta-min condition will be necessary.
Here ``fixed alternative'' means that the population graph and covariance structure do not approach the null with $n$.
The choice $\delta_n=\sqrt{\log n}$ is a convenient sufficient choice: the same conclusions hold when $\delta_n\to\infty$ and $\delta_n+\sqrt{\log p}=o(\sqrt n)$, with $C_n=C_\Delta\sqrt{\log p}$ unchanged.
\end{remark}

\subsection{Support Recovery by Thresholded CLIME}
\label{sec:clime.support}

We give a thresholded CLIME selector whose error bound below verifies the consistency requirement in Corollary~\ref{cor:two.sample.direct} whenever it vanishes.
For each $1\le j\le p$, define
\[
\check{\bm\omega}^{(1)}_j
\in
\argmin_{\bm\omega\in\R^p}
|\bm\omega|_1
\quad\text{subject to}\quad
|\hat\Sigma^{(1)}\bm\omega-e_j|_\infty
\le
\lambda_{n_1},
\]
and let
\[
\check\Omega^{(1)}
=
\big(
\check{\bm\omega}^{(1)}_1,\ldots,
\check{\bm\omega}^{(1)}_p
\big).
\]
If a feasible set is empty, we set the corresponding column equal to zero.
The proof below shows that all feasible sets are nonempty on the event of interest.

Following \citet{cai2011constrained}, symmetrize $\check\Omega^{(1)}$ by retaining, for each $i<j$, whichever of $\check\omega_{ij}^{(1)}$ and $\check\omega_{ji}^{(1)}$ has smaller absolute value (breaking ties arbitrarily), and use that value for both symmetric entries.
Keep the diagonal unchanged and denote the resulting matrix by $\hat\Omega^{\mathrm{CL},(1)}$.
This is the CLIME estimator of $\Omega^{(1)}$.
For a threshold $\tau_{n_1,p}>0$, select
\[
\hat E^{(1)}
=
\left\{
(i,j):1\le i,j\le p,\quad
i=j
\ \text{or}\
|\hat\omega_{ij}^{\mathrm{CL},(1)}|>\tau_{n_1,p}
\right\}.
\]
Thus the diagonal is included by convention, as in the entrywise-support definition used in the main paper, while only the off-diagonal entries are thresholded.

For the theoretical guarantee, choose
\[
\lambda_{n_1}
=
C_\Omega K_A
\sqrt{\frac{\log(pn_1)}{n_1}},
\qquad
\tau_{n_1,p}
=
4C_\Omega\lambda_{n_1},
\qquad
K_A^2
\ge
2v_\Sigma(C_0,\beta)(A+2),
\]
where $A>0$ is a fixed constant, and recall that $C_\Omega$ is the constant in Condition~\ref{cond:Sigma.Omega.1norm}.
Under Condition~\ref{cond.A}, the proof of Lemma~18 in Appendix~A of \citet{zhai2026simultaneous} also gives the absolute-covariance bound, uniformly over $1\le s,t\le p$,
\[
\sum_{k=-\infty}^{\infty}
\left|
(\Gamma_k)_{ss}(\Gamma_k)_{tt}
+
(\Gamma_k)_{st}(\Gamma_k)_{ts}
\right|
\le
v_\Sigma(C_0,\beta),
\]
where $\Gamma_k=\cov(X_0^{(1)},X_{-k}^{(1)})$ and $v_\Sigma(C_0,\beta)$ is a finite constant depending only on $C_0$ and the fixed $\beta>3/4$.
The finite-sample coordinate variances in the covariance Gaussian approximation are Bartlett-weighted sums of these covariances and are therefore also bounded by $v_\Sigma(C_0,\beta)$.
Hence both $\lambda_{n_1}$ and $\tau_{n_1,p}$ have order $\{\log(pn_1)/n_1\}^{1/2}$.
The quantity $A$ is only a tail-probability exponent.
Taking $A=1$ is a simple default; increasing $A$ reduces the tail bound but raises the threshold and the required signal strength.

\begin{proposition}[Support recovery by thresholded CLIME]
\label{prop:clime.support.recovery}
Suppose that the first sample satisfies Conditions~\ref{cond.A}, \ref{cond:Sigma.Omega.1norm}, and~\ref{cond.G}.
For any fixed $A>0$, choose $K_A$ as above.
If every nonzero off-diagonal entry satisfies
\[
\min_{1\le i\ne j \le p: \omega_{ij}^{(1)}\ne0}|\omega_{ij}^{(1)}|
>
2\tau_{n_1,p},
\]
then
\[
\pi_{n_1}
=
\P\left\{
\hat E^{(1)}\ne E^{(1)}
\right\}
\le
C\Psi(p,n_1)
+
2(pn_1)^{-A},
\]
where $C>0$ depends only on $C_0$, $c_0$, and the fixed $\beta>3/4$ in Conditions~\ref{cond.A} and~\ref{cond.G}.
\end{proposition}

\begin{proof}
Define the event
\[
\mathcal E_{n_1}
=
\left\{
|\hat\Sigma^{(1)}-\Sigma^{(1)}|_\infty
\le
\frac{\lambda_{n_1}}{C_\Omega}
\right\}.
\]
Let $\bm\omega_j^{(1)}$ be the $j$th column of $\Omega^{(1)}$.
On $\mathcal E_{n_1}$,
\[
|\hat\Sigma^{(1)}\bm\omega_j^{(1)}-e_j|_\infty
\le
\frac{\lambda_{n_1}}{C_\Omega}
|\bm\omega_j^{(1)}|_1
\le
\lambda_{n_1}.
\]
Thus the true columns are feasible, and optimality gives $|\check\Omega^{(1)}|_1\le|\Omega^{(1)}|_1\le C_\Omega$.
Both matrices satisfy the CLIME constraint, so $|\hat\Sigma^{(1)}(\check\Omega^{(1)}-\Omega^{(1)})|_\infty \le2\lambda_{n_1}$.
Consequently,
\[
\begin{aligned}
|\Sigma^{(1)}(\check\Omega^{(1)}-\Omega^{(1)})|_\infty
&\le
2\lambda_{n_1}
+
\frac{\lambda_{n_1}}{C_\Omega}
\left(
|\check\Omega^{(1)}|_1
+
|\Omega^{(1)}|_1
\right)
\\
&\le
4\lambda_{n_1}.
\end{aligned}
\]
Since $\Omega^{(1)}$ is symmetric, its maximum absolute row sum is
$|\Omega^{(1)}|_1\le C_\Omega$.
Using the identity $\Omega^{(1)}\Sigma^{(1)}=I$ gives
\[
|\check\Omega^{(1)}-\Omega^{(1)}|_\infty
\le
4C_\Omega\lambda_{n_1}
=
\tau_{n_1,p}.
\]
Every entry of $\hat\Omega^{\mathrm{CL},(1)}$ is one of $\check\omega_{ij}^{(1)}$ and $\check\omega_{ji}^{(1)}$.
Since $\omega_{ij}^{(1)}=\omega_{ji}^{(1)}$, the symmetrization preserves this bound:
\[
|
\hat\Omega^{\mathrm{CL},(1)}
-
\Omega^{(1)}
|_\infty
\le
\tau_{n_1,p}.
\]

Hence every zero off-diagonal entry is removed by thresholding, while every nonzero entry is retained because $|\hat\omega_{ij}^{\mathrm{CL},(1)}| \ge|\omega_{ij}^{(1)}|-\tau_{n_1,p}>\tau_{n_1,p}$.
Because $\Omega^{(1)}$ is positive definite, its diagonal belongs to $E^{(1)}$; these pairs are included in $\hat E^{(1)}$ by definition.
Thus $\hat E^{(1)}=E^{(1)}$ on $\mathcal E_{n_1}$.

It remains to bound $\P(\mathcal E_{n_1}^c)$.
Let $Z_\Sigma^{(1)}$ be the Gaussian vector in the covariance approximation for the first sample.
Under Conditions~\ref{cond.A} and~\ref{cond.G}, Theorem 1 of \citet{zhai2026simultaneous} bounds the approximation error by $C\Psi(p,n_1)$, and under Condition~\ref{cond.A}, each coordinate variance is bounded by
$v_\Sigma(C_0,\beta)$.
The Gaussian union bound, the choice of $K_A$, and $p^2\le(pn_1)^2$ therefore give
\[
\begin{aligned}
\P(\mathcal E_{n_1}^c)
&\le
C\Psi(p,n_1)
+
2p^2
\exp\left\{
-\frac{K_A^2\log(pn_1)}{2v_\Sigma(C_0,\beta)}
\right\}
\\
&\le
C\Psi(p,n_1)
+
2(pn_1)^{-A}.
\end{aligned}
\]
Since $\hat E^{(1)}=E^{(1)}$ on $\mathcal E_{n_1}$, $\pi_{n_1}\le\P(\mathcal E_{n_1}^c)$, which proves the result.
\end{proof}

\section{Additional Experiment Details}\label{app:exp.details}

\subsection{Generating Copies of Time Series using Fast Fourier Transformation}\label{app:fft}

The simulation studies require many repeated realizations of the Gaussian linear process $X_i=\sum_{t=0}^\infty A_t\epsilon_{i-t}$.
A direct implementation is computationally expensive because each observation involves an infinite summation.
Following the FFT-based simulation scheme summarized in \citet{zhai2026simultaneous}, we first replace the process by the truncated version
\[
X_i^\ast=\sum_{t=0}^{N-1} A_t\epsilon_{i-t}.
\]
The truncation length $N$ needs to be large when the process has long memory, since the coefficient matrices $A_t$ decay slowly.
In all simulations, we take $N=n^2$ and write $K_{N,n}=\lfloor N/n\rfloor$.
The key computational idea is to generate one long circular convolution and then divide it into $K_{N,n}$ consecutive length-$n$ segments, which gives many approximate copies of $\bm{X}_n^\ast=[X_1^\ast,X_2^\ast,\ldots,X_n^\ast]^\top$.

We first describe the univariate case.
Suppose $p=d=1$ and $X_i^\ast=\sum_{t=0}^{N-1}a_t\epsilon_{i-t}$.
Let
\[
\bm{a}_N=[a_1,a_2,\ldots,a_{N-1},a_0]^\top,
\qquad
\bm{e}_N=[\epsilon_n,\epsilon_{n-1},\ldots,\epsilon_{n-N+1}]^\top.
\]
The circulant embedding approximation is
\[
\bm{W}_N=[W_1,W_2,\ldots,W_N]^\top
=\operatorname{Circ}(\bm{a}_N)\bm{e}_N,
\]
where $\operatorname{Circ}(\bm{a}_N)$ denotes the circulant matrix generated from $\bm{a}_N$.
Since circulant matrices are diagonalized by the discrete Fourier transform, this multiplication can be computed by
\[
\bm{b}_N=\operatorname{FFT}(\bm{a}_N),
\qquad
\bm{f}_N=\operatorname{FFT}(\bm{e}_N),
\qquad
\bm{W}_N=\operatorname{FFT}^{-1}(\bm{p}_N),
\]
where $\bm{p}_N=\bm{b}_N\circ \bm{f}_N$ and $\circ$ denotes entrywise multiplication.
The vector $\bm{W}_N$ is then partitioned into consecutive blocks
\[
\bm{W}^{(k)}=[W_{(k-1)n+1},W_{(k-1)n+2},\ldots,W_{kn}]^\top,
\qquad
1\leq k\leq K_{N,n}.
\]
Each $\bm{W}^{(k)}$ is treated as one simulated realization of $\bm{X}_n^\ast$.
These blocks are not exactly independent because they are extracted from the same long circular convolution, but the resulting dependence is weak enough for the Monte Carlo approximation used in the simulation study.

The multidimensional case is handled coordinate-by-coordinate while using the same innovation sequence, which preserves the cross-sectional dependence induced by the coefficient matrices.
For $X_i^\ast\in\R^p$, $A_t\in\R^{p\times d}$, and $\epsilon_t\in\R^d$, first generate
\[
\bm{e}_N=[\epsilon_n^\top,\epsilon_{n-1}^\top,\ldots,\epsilon_{n-N+1}^\top]^\top\in\R^{Nd},
\qquad
\bm{f}_N=\operatorname{FFT}(\bm{e}_N).
\]
For the $j$-th coordinate, define
\[
\bm{a}_{N,j}=\left[(A_1)_{j\cdot},(A_2)_{j\cdot},\ldots,(A_{N-1})_{j\cdot},(A_0)_{j\cdot}\right]^\top\in\R^{Nd},
\qquad
\bm{A}^j=\operatorname{Circ}(\bm{a}_{N,j}).
\]
We compute
\[
\bm{b}_{N,j}=\operatorname{FFT}(\bm{a}_{N,j}),
\qquad
\bm{p}_{N,j}=\bm{b}_{N,j}\circ\bm{f}_N,
\qquad
\bm{Y}^j=\operatorname{FFT}^{-1}(\bm{p}_{N,j})\in\R^{Nd}.
\]
The simulated value of the $j$-th coordinate at time $i$ is then read from the first coordinate of the corresponding $d$-dimensional block, namely $W_{ij}=(\bm{Y}^j)_{(i-1)d+1}$.
Repeating this operation for $j=1,\ldots,p$ fills the full simulated array.
Finally, for $1\leq k\leq K_{N,n}$, we form
\[
\bm{W}^{(k)}=
\left(W_{ij}\right)_{(k-1)n+1\leq i\leq kn,\ 1\leq j\leq p}\in\R^{n\times p}.
\]
We treat $\bm{W}^{(1)},\bm{W}^{(2)},\ldots,\bm{W}^{(K_{N,n})}$ as $K_{N,n}$ approximate copies of $\bm{X}_n^\ast$ and use the first 200 available blocks in the experiments.
This FFT construction is the reason that the high-dimensional simulation regimes in Section~\ref{sec:simulation} remain computationally feasible.

\subsection{Additional Null Test Results}\label{app:null.exp.details}

Figure~\ref{fig:qq.beta0.55} gives the QQ plots for the out-of-theory case $\beta=0.55$, while Tables~\ref{tab:ks-distance-high-dimensional} and \ref{tab:w1-distance-high-dimensional} report the Kolmogorov and Wasserstein-$1$ distances for all three dependence levels in the original $b=4$ experiment of Section~\ref{sec:null.text.validity}.

\begin{figure}[H]
    \centering
    \includegraphics[width=0.8\linewidth]{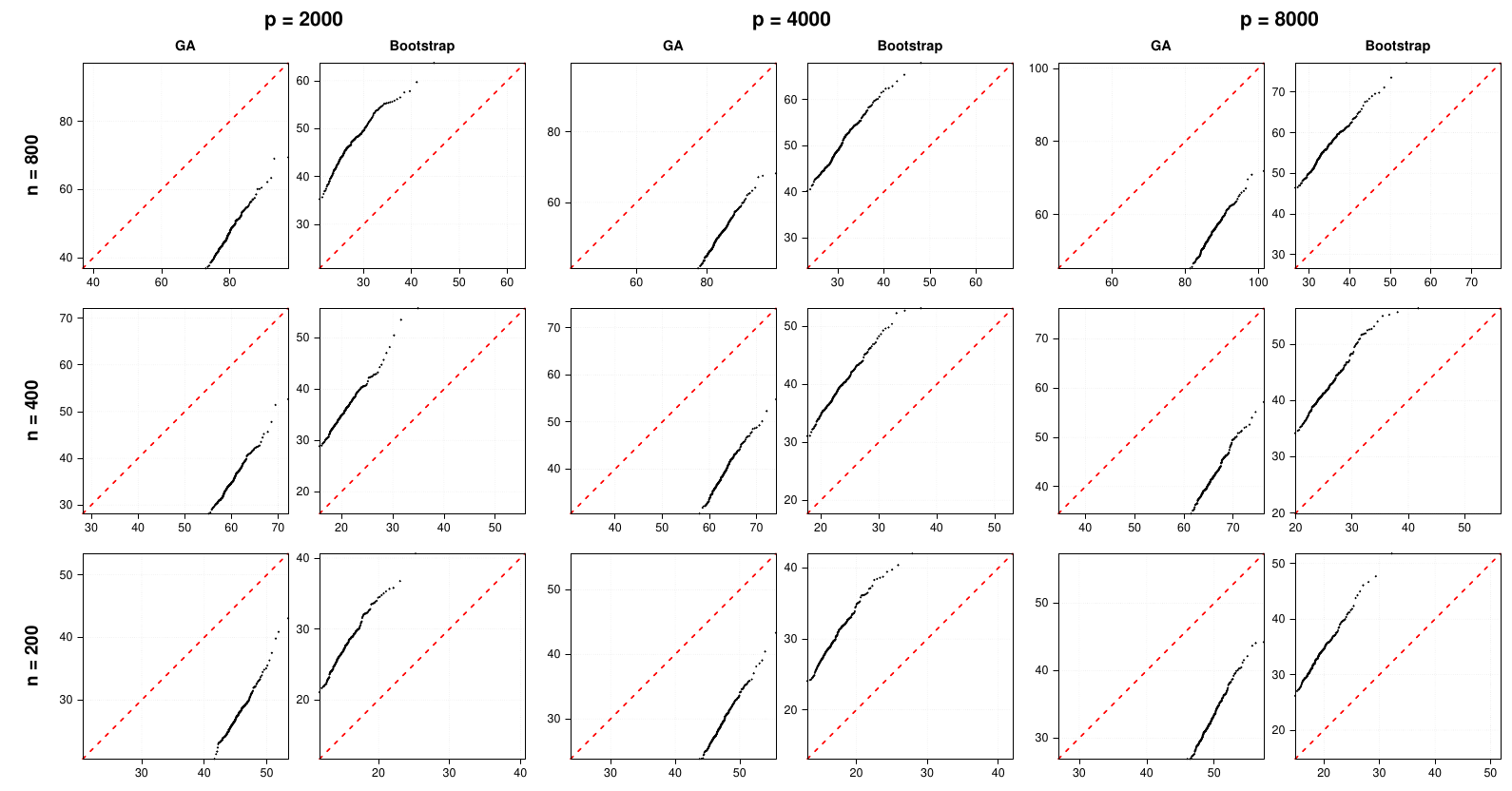}
    \caption{QQ-plots for null tests under ultra-long-memory, $\beta=0.55$, which exceeds the range in Condition \ref{cond.A}.
    The $x$-axis is $|G|_\infty$ for GA and $l^{-1/2}|\hat B_{r,l}|_\infty$ for bootstrap.
    The $y$-axis is $\hat T_n$ for both.
    The red dashed line is $y=x$.
    Closer to this line indicates better proximity.}
    \label{fig:qq.beta0.55}
\end{figure}

\begin{table}[H]
\centering
\caption{
Kolmogorov distance for the statistic $\hat T_n$ under $H_0$.
The GA section compares $\hat T_n$ against its finite-sample Gaussian approximation $|G|_\infty$, while the Bootstrap section compares $\hat T_n$ against its overlapping block bootstrap approximation $l^{-1/2}|\hat B_{r,l}|_\infty$.
Range of the Kolmogorov distance is $[0,1]$.
A smaller value indicates a closer approximation.
}
\label{tab:ks-distance-high-dimensional}
\resizebox{\textwidth}{!}{%
\begin{tabular}{llccccccccc}
\toprule
& & \multicolumn{3}{c}{$\beta = 2$}
& \multicolumn{3}{c}{$\beta = 0.9$}
& \multicolumn{3}{c}{$\beta = 0.55$} \\
\cmidrule(lr){3-5}
\cmidrule(lr){6-8}
\cmidrule(lr){9-11}
& $n$
& $p = 2000$ & $p = 4000$ & $p = 8000$
& $p = 2000$ & $p = 4000$ & $p = 8000$
& $p = 2000$ & $p = 4000$ & $p = 8000$ \\
\midrule
\multirow{3}{*}{GA}
& 800 & 0.20 & 0.24 & 0.26 & 0.36 & 0.49 & 0.58 & 1.00 & 1.00 & 1.00 \\
& 400 & 0.36 & 0.36 & 0.43 & 0.15 & 0.18 & 0.25 & 1.00 & 1.00 & 1.00 \\
& 200 & 0.52 & 0.56 & 0.67 & 0.43 & 0.32 & 0.20 & 1.00 & 1.00 & 1.00 \\
\midrule
\multirow{3}{*}{Bootstrap}
& 800 & 0.06 & 0.14 & 0.22 & 0.10 & 0.10 & 0.14 & 0.96 & 0.97 & 0.98 \\
& 400 & 0.10 & 0.08 & 0.16 & 0.57 & 0.52 & 0.47 & 0.98 & 0.97 & 0.98 \\
& 200 & 0.35 & 0.26 & 0.30 & 0.84 & 0.83 & 0.85 & 0.98 & 0.98 & 0.98 \\
\bottomrule
\end{tabular}%
}
\end{table}

\begin{table}[H]
\centering
\caption{
Wasserstein-$1$ distance for the statistic $\hat T_n$ under $H_0$.
The GA section compares $\hat T_n$ against its finite-sample Gaussian approximation $|G|_\infty$, while the Bootstrap section compares $\hat T_n$ against its overlapping block bootstrap approximation $l^{-1/2}|\hat B_{r,l}|_\infty$.
Range of the W1-distance is $[0,\infty)$.
A smaller value indicates a closer approximation.
}
\label{tab:w1-distance-high-dimensional}
\resizebox{\textwidth}{!}{%
\begin{tabular}{llccccccccc}
\toprule
& & \multicolumn{3}{c}{$\beta = 2$}
& \multicolumn{3}{c}{$\beta = 0.9$}
& \multicolumn{3}{c}{$\beta = 0.55$} \\
\cmidrule(lr){3-5}
\cmidrule(lr){6-8}
\cmidrule(lr){9-11}
& $n$
& $p = 2000$ & $p = 4000$ & $p = 8000$
& $p = 2000$ & $p = 4000$ & $p = 8000$
& $p = 2000$ & $p = 4000$ & $p = 8000$ \\
\midrule
\multirow{3}{*}{GA}
& 800 & 0.11 & 0.14 & 0.13 & 0.76 & 0.92 & 1.10 & 32.68 & 32.96 & 32.45 \\
& 400 & 0.20 & 0.19 & 0.23 & 0.30 & 0.26 & 0.38 & 24.96 & 23.97 & 23.79 \\
& 200 & 0.36 & 0.41 & 0.44 & 0.62 & 0.34 & 0.29 & 17.76 & 17.84 & 16.78 \\
\midrule
\multirow{3}{*}{Bootstrap}
& 800 & 0.04 & 0.07 & 0.13 & 0.15 & 0.24 & 0.30 & 18.90 & 19.28 & 21.42 \\
& 400 & 0.05 & 0.04 & 0.09 & 1.16 & 1.19 & 1.04 & 15.08 & 15.86 & 16.54 \\
& 200 & 0.22 & 0.14 & 0.16 & 2.09 & 2.10 & 2.26 & 12.02 & 13.06 & 14.25 \\
\bottomrule
\end{tabular}%
}
\end{table}

We extend the same null-validity setup and simulation scheme in Appendix~\ref{app:fft} in two directions.
Each configuration uses 200 simulated time series, and the bootstrap rejection rates are reported at the nominal level $\alpha=0.10$.
The distances below compare $\hat T_n$ with the same Gaussian and block-bootstrap approximations as in Tables~\ref{tab:ks-distance-high-dimensional} and \ref{tab:w1-distance-high-dimensional}.

In Extension 1, we fix $n=800$ and $p=2000$ and vary the bandwidth $b\in\{2,4,6,8,12,16\}$ under both $\beta=2$ and $0.9$.
For this banded graph, $s_0=2b+1$, giving $s_0\in\{5,9,13,17,25,33\}$.
The precision-matrix construction is unchanged, and the $b=4$ results are reused from the original experiment.
Table~\ref{tab:extension1} reports the results.
Although the distances are not exactly monotone at small bandwidths, approximation accuracy and null calibration generally deteriorate once $b$, and hence $s_0$, becomes large, especially under long memory.
This is consistent with the $s_0$-dependent remainder terms in Theorems~\ref{thm:ga.direct} and \ref{thm:block.bootstrap.direct}.

\begin{table}[H]
\centering
\small
\setlength{\tabcolsep}{4pt}
\caption{Extension 1: varying sparsity under short memory ($\beta=2$, left) and long memory ($\beta=0.9$, right), with $n=800$ and $p=2000$. Each configuration uses 200 repetitions. Rejection rates use the nominal level $\alpha=0.10$.}
\label{tab:extension1}
\begin{tabular}{@{}rrrrrrr@{\hspace{12pt}}rrrrr@{}}
\toprule
& & \multicolumn{5}{c}{$\beta=2$} & \multicolumn{5}{c}{$\beta=0.9$} \\
\cmidrule(lr){3-7}\cmidrule(lr){8-12}
& & \multicolumn{2}{c}{Gaussian} & \multicolumn{2}{c}{Bootstrap} & \multirow{2}{*}{\shortstack{Rejection\\rate}} & \multicolumn{2}{c}{Gaussian} & \multicolumn{2}{c}{Bootstrap} & \multirow{2}{*}{\shortstack{Rejection\\rate}} \\
\cmidrule(lr){3-4}\cmidrule(lr){5-6}\cmidrule(lr){8-9}\cmidrule(lr){10-11}
$b$ & $s_0$ & $d_K$ & $W_1$ & $d_K$ & $W_1$ & & $d_K$ & $W_1$ & $d_K$ & $W_1$ & \\
\midrule
2  & 5  & 0.11 & 0.05 & 0.18 & 0.11 & 0.05 & 0.67 & 1.58 & 0.06 & 0.11 & 0.11 \\
4  & 9  & 0.20 & 0.11 & 0.06 & 0.04 & 0.10 & 0.36 & 0.76 & 0.10 & 0.15 & 0.08 \\
6  & 13 & 0.22 & 0.14 & 0.04 & 0.02 & 0.11 & 0.22 & 0.31 & 0.17 & 0.31 & 0.12 \\
8  & 17 & 0.24 & 0.13 & 0.08 & 0.04 & 0.11 & 0.48 & 0.66 & 0.35 & 0.62 & 0.21 \\
12 & 25 & 0.28 & 0.17 & 0.15 & 0.08 & 0.15 & 0.83 & 1.61 & 0.59 & 1.05 & 0.41 \\
16 & 33 & 0.40 & 0.23 & 0.31 & 0.16 & 0.20 & 0.94 & 2.32 & 0.73 & 1.38 & 0.66 \\
\bottomrule
\end{tabular}
\end{table}

In Extension 2, we retain $p=2000$ and $b=4$ and add $\beta=0.85$ and $0.8$ to examine behavior closer to the boundary $\beta=3/4$.
The left panel of Table~\ref{tab:extension2} gives the results at $n=800$, including the original $\beta=2$ and $0.9$ results for comparison.
At this fixed sample size and dimension, the bootstrap distances become larger and the null rejection rates become increasingly inflated in the near-boundary settings.
We then double $n$ to $1600$, keeping $p=2000$ and $b=4$ fixed; see the right panel of Table~\ref{tab:extension2}.
The bootstrap approximation and null calibration generally improve.
These results illustrate the finite-sample limitation near $\beta=3/4$ and are qualitatively consistent with the joint $(n,p,\beta)$ scaling in our bounds.

\begin{table}[H]
\centering
\small
\setlength{\tabcolsep}{4pt}
\caption{Extension 2: varying temporal dependence with $n=800$ (left) and $n=1600$ (right), with $p=2000$ and $b=4$. Each configuration uses 200 repetitions. Rejection rates use the nominal level $\alpha=0.10$.}
\label{tab:extension2}
\begin{tabular}{@{}rrrrrr@{\hspace{12pt}}rrrrr@{}}
\toprule
& \multicolumn{5}{c}{$n=800$} & \multicolumn{5}{c}{$n=1600$} \\
\cmidrule(lr){2-6}\cmidrule(lr){7-11}
& \multicolumn{2}{c}{Gaussian} & \multicolumn{2}{c}{Bootstrap} & \multirow{2}{*}{\shortstack{Rejection\\rate}} & \multicolumn{2}{c}{Gaussian} & \multicolumn{2}{c}{Bootstrap} & \multirow{2}{*}{\shortstack{Rejection\\rate}} \\
\cmidrule(lr){2-3}\cmidrule(lr){4-5}\cmidrule(lr){7-8}\cmidrule(lr){9-10}
$\beta$ & $d_K$ & $W_1$ & $d_K$ & $W_1$ & & $d_K$ & $W_1$ & $d_K$ & $W_1$ & \\
\midrule
2    & 0.20 & 0.11 & 0.06 & 0.04 & 0.10 & 0.14 & 0.07 & 0.07 & 0.05 & 0.07 \\
0.9  & 0.36 & 0.76 & 0.10 & 0.15 & 0.08 & 0.56 & 1.19 & 0.36 & 1.10 & 0.02 \\
0.85 & 0.27 & 0.56 & 0.37 & 1.14 & 0.30 & 0.48 & 1.20 & 0.07 & 0.26 & 0.05 \\
0.8  & 0.29 & 0.63 & 0.64 & 2.12 & 0.56 & 0.27 & 0.70 & 0.36 & 1.21 & 0.25 \\
\bottomrule
\end{tabular}
\end{table}

\subsection{Sherlock fMRI Experiment Result}\label{app:sherlock}

\begin{table}[H]
\centering
\caption{Movie self-selected graph and recall containment test results at $\alpha=0.05$.
Here $\hat\nu$ is the movie-phase threshold used to construct the selected graph, $d_i$ is the degree of node $i$ in the selected movie graph, $q_{0.25}(d)$ and $q_{0.75}(d)$ are the empirical 25\% and 75\% quartiles of the node degrees, and $\bar s$ is the average column-wise support size, with $s_i=d_i+1$.
The column $p_{\mathrm{mov}}$ is the movie self-diagnostic p-value used for graph selection.
The column $p_{\mathrm{rec}}$ is the raw recall-phase containment-test p-value.
The column $p_{\mathrm{rec}}^{\mathrm{Holm}}$ is the recall-phase p-value after Holm correction across the 17 subjects.}
\label{tab:movie_recall_holm_subjects}
\begin{tabular}{rrrrrrrr}
\toprule
Subject & $\hat\nu$ & $q_{0.25}(d)$ & $\bar s$ & $q_{0.75}(d)$ & $p_{\mathrm{mov}}$ & $p_{\mathrm{rec}}$ & $p_{\mathrm{rec}}^{\mathrm{Holm}}$ \\
\midrule
 1 & 0.1956 & 22.5 & 35.9 & 42.5 & 0.0692 & 0.00217 & 0.0145 \\
 2 & 0.1567 & 26.0 & 37.8 & 44.5 & 0.0648 & 0.00121 & 0.0145 \\
 3 & 0.1089 & 53.0 & 63.5 & 71.0 & 0.0516 & 0.00157 & 0.0145 \\
 4 & 0.1219 & 36.0 & 45.5 & 52.0 & 0.0835 & 0.00262 & 0.0145 \\
 5 & 0.1083 & 47.0 & 56.3 & 62.0 & 0.0549 & 0.00196 & 0.0145 \\
 6 & 0.1511 & 23.0 & 32.8 & 39.0 & 0.0643 & 0.00126 & 0.0145 \\
 7 & 0.1340 & 47.5 & 59.5 & 68.0 & 0.0511 & 0.00166 & 0.0145 \\
 8 & 0.1242 & 39.0 & 48.1 & 53.0 & 0.0582 & 0.00110 & 0.0143 \\
 9 & 0.1590 & 19.0 & 29.0 & 34.0 & 0.0725 & 0.00162 & 0.0145 \\
10 & 0.1129 & 45.0 & 55.0 & 61.0 & 0.1153 & 0.00132 & 0.0145 \\
11 & 0.1448 & 33.0 & 44.9 & 53.8 & 0.0637 & 0.00217 & 0.0145 \\
12 & 0.1302 & 39.0 & 50.9 & 59.0 & 0.0571 & 0.000899 & 0.0129 \\
13 & 0.1307 & 40.0 & 52.2 & 58.0 & 0.2202 & 0.000596 & 0.0101 \\
14 & 0.1257 & 33.0 & 41.9 & 46.0 & 0.0670 & 0.000861 & 0.0129 \\
15 & 0.1456 & 24.0 & 33.7 & 39.0 & 0.0544 & 0.00135 & 0.0145 \\
16 & 0.1729 & 31.0 & 43.8 & 52.0 & 0.0511 & 0.00128 & 0.0145 \\
17 & 0.1322 & 33.0 & 44.6 & 50.0 & 0.0577 & 0.000760 & 0.0122 \\
\bottomrule
\end{tabular}
\end{table}

\subsection{Long Memory of Sherlock fMRI Time Series}\label{app:long.memory}

The time series in the Sherlock fMRI dataset exhibit significant long memory.
To demonstrate this, we plot the histogram of Hurst exponent of all $17 \times 268$ movie-phase time series in Figure \ref{fig:hurst}.
The Hurst exponent $H$, ranging from zero to one, measures the scaling of temporal dependence in a time series, with $H=1/2$ corresponding to the short-memory benchmark associated with uncorrelated or rapidly decorrelating fluctuations.
Values $H>1/2$ indicate persistence: positive deviations tend to be followed by positive deviations, and negative deviations tend to be followed by negative deviations over long time scales.
For fractional long-memory processes, the autocorrelation typically satisfies
\[
\rho(k) \asymp k^{2H-2},
\]
so when $H>1/2$, we have $2H-2 \in (-1,0)$.
Consequently, the autocorrelation decays slowly and is not summable, which means that dependence persists across distant lags, i.e., long-memory.
We can see from Figure \ref{fig:hurst} that $99.8\%$ of all time series have $H>1/2$, i.e., exhibiting long-memory.
The median Hurst exponent is at $0.719$.
For each subject, we have also provided a boxplot.
All subjects have a confident interval of Hurst exponent that completely exceeds 0.5, a strong sign of these time series being long-memory.

The ACF plots in Figure \ref{fig:acf} provide a complementary diagnostic.
In a short-memory process, the autocorrelation decays rapidly toward zero and usually becomes negligible after a small number of lags.
By contrast, a slowly decaying or effectively non-decaying ACF over many lags indicates that distant observations remain correlated.
This persistent autocorrelation is the empirical signature of long-range dependence.
Thus, the combination of $H>1/2$ and significant autocorrelations across long lags provides evidence that the Sherlock movie time series exhibit long-memory behavior.
We selected the time series at the 50th, 75th, 90th, and 95th percentile of Hurst exponents.
It appears that all these four time series have significant ACF lags near 100.
For the latter two, particularly, the ACF is decaying very slowly.
These time series are typically long-range dependent.

\begin{figure}
    \centering
    \includegraphics[width=0.8\linewidth]{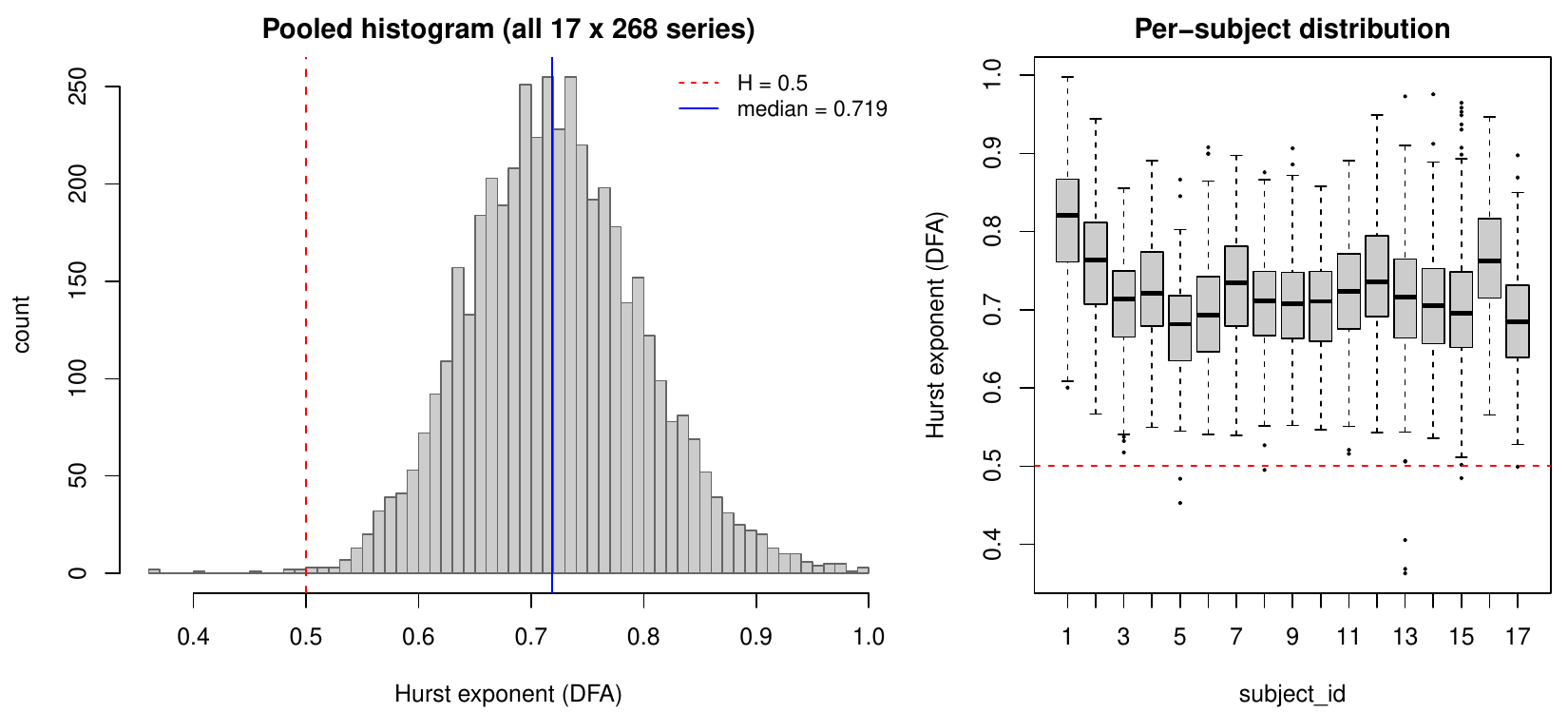}
    \caption{Distribution of Hurst exponents across all subject-ROI movie time series. Vertical dashed line: $H = 0.5$.
    Values $H > 0.5$ are considered consistent with long-range dependence.}
    \label{fig:hurst}
\end{figure}

\begin{figure}
    \centering
    \includegraphics[width=0.8\linewidth]{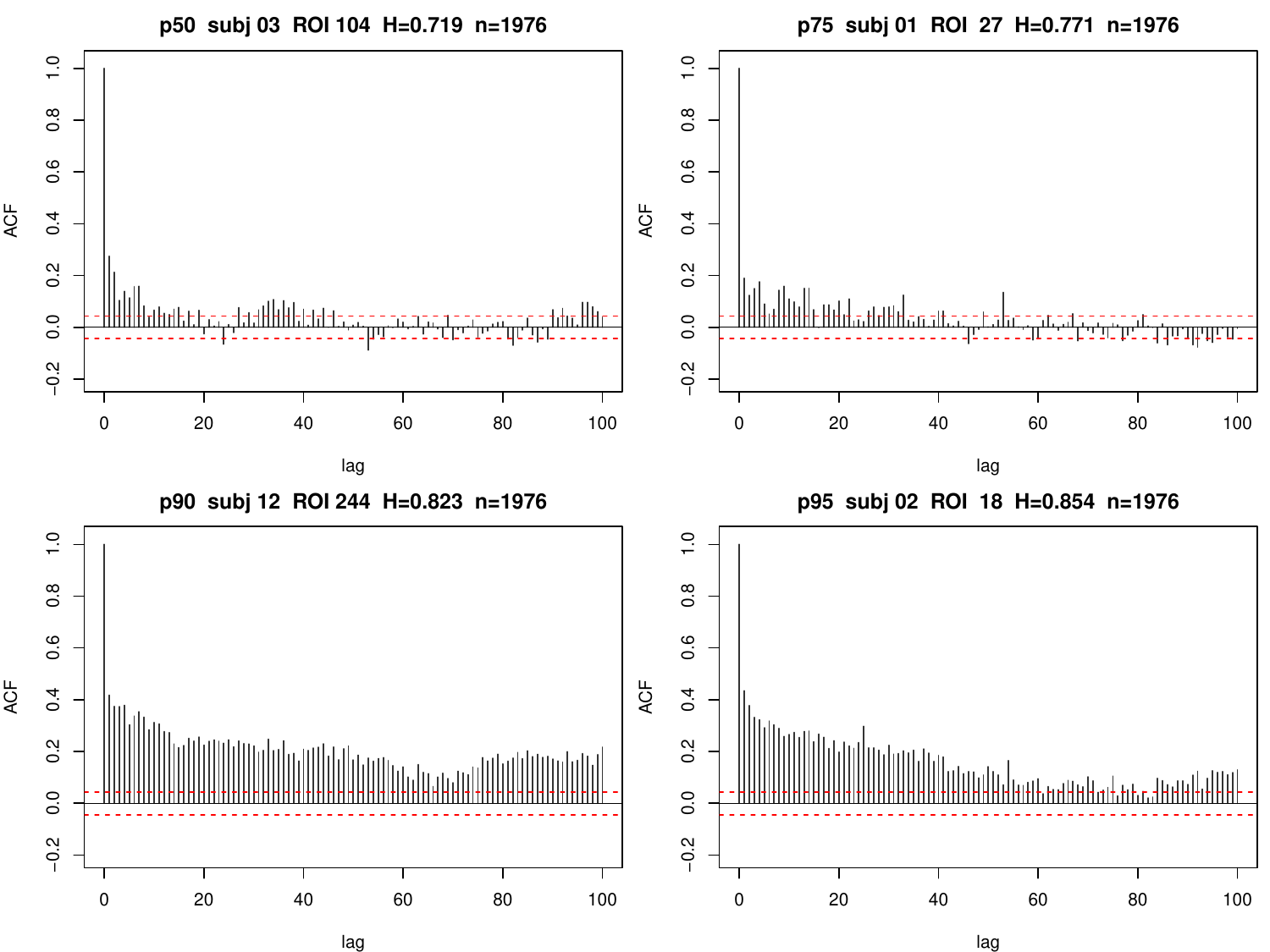}
    \caption{Sample autocorrelation functions (ACF) for movie time series at Hurst exponent distribution percentiles
    (50th, 75th, 90th, 95th).
    Each panel shows ACF up to lag 100; dashed bands at $\pm 1.96/\sqrt{n}$.
    Slowly-decaying significant ACF coefficients are evidence consistent with long-range dependence.}
    \label{fig:acf}
\end{figure}

\section{Proof for Section~\ref{sec:theory}}
\label{app:pf.theory}

Consider the following examples satisfying Conditions~\ref{cond:Sigma.Omega.1norm}--\ref{cond:local.inverse}.
Let $\Omega=I_p+W$, where $W$ is symmetric with zero diagonal and $\max_i\sum_j|W_{ij}|\le\rho$ for a fixed $0<\rho<1$.
Then $\Omega$ is positive definite, and the Neumann series gives
\[
|\Omega|_1\le1+\rho,
\qquad
|\Sigma|_1=|(I_p+W)^{-1}|_1\le\frac1{1-\rho}.
\]
For any nonempty proper subset $S\subset\{1,\ldots,p\}$, the Schur complement identity gives
\[
(\Sigma_{SS})^{-1}
=\Omega_{SS}-\Omega_{SS^c}(\Omega_{S^cS^c})^{-1}\Omega_{S^cS}.
\]
The two cross-block matrix norms are at most $\rho$, while $|(\Omega_{S^cS^c})^{-1}|_1\le(1-\rho)^{-1}$.
Consequently,
\[
|(\Sigma_{SS})^{-1}|_1
\le1+\rho+\frac{\rho^2}{1-\rho}
=\frac1{1-\rho}.
\]
The same bound holds for $S=\{1,\ldots,p\}$.
Thus Conditions~\ref{cond:Sigma.Omega.1norm}--\ref{cond:local.inverse} hold with $C_\Sigma=C_A=(1-\rho)^{-1}$ and $C_\Omega=1+\rho$, uniformly in the dimension and the support sizes.

For a banded example, take $W_{ij}=\rho/(2b)$ when $0<|i-j|\le b$ and zero otherwise, where $b\ge1$ is the bandwidth.
For a block-diagonal example, partition the coordinates into blocks and, within each block of size $m\ge2$, take $W_{ij}=\rho/(m-1)$ for $i\ne j$; all other entries are zero.
Singleton blocks have precision entry one.
Both constructions satisfy the row-sum bound above.
The bandwidth or block size may grow with $n$, subject to the sparsity--dimension rates in Theorems~\ref{thm:ga.direct} and~\ref{thm:block.bootstrap.direct}, while the constants in Conditions~\ref{cond:Sigma.Omega.1norm}--\ref{cond:local.inverse} remain unchanged.

\subsection{Proof of Proposition~\ref{prop:deterministic.reduction}}
\label{app:proof.det.reduction}

For brevity, write
\[
    \Delta = \hat\Sigma_n - \Sigma,\quad
    A_j = B_{0,j}^\top \Sigma B_{0,j},\quad
    \hat A_j = B_{0,j}^\top \hat\Sigma_n B_{0,j},\quad
    b_j = B_{0,j}^\top e_j .
\]
Throughout the proof, the constant $C>0$ may vary from line to line and depends only on $C_A$, $C_\Sigma$, and $C_\Omega$.
On the event $2C_A s_0|\Delta|_\infty \le 1$, Condition \ref{cond:local.inverse} gives
\[
    |A_j^{-1}(\hat A_j-A_j)|_1
    \le
    |A_j^{-1}|_1|\hat A_j-A_j|_1
    \le
    C_A s_{0,j}|\Delta|_\infty
    \le
    \frac12 .
\]
Hence $\hat A_j$ is invertible for every $1\le j\le p$, and the exceptional definition of $\hat{\bm \omega}_{0,j}$ being zero when $\hat A_j$ being degenerate is then inactive on this event.
Moreover,
\[
    |\hat A_j^{-1}|_1
    \le
    \frac{|A_j^{-1}|_1}{1-|A_j^{-1}(\hat A_j-A_j)|_1}
    \le
    2C_A .
\]
The inverse identity gives
\[
    \hat A_j^{-1}-A_j^{-1}
    =
    -A_j^{-1}(\hat A_j-A_j)A_j^{-1}
    +
    A_j^{-1}(\hat A_j-A_j)A_j^{-1}(\hat A_j-A_j)\hat A_j^{-1}.
\]
Since $\bm \omega_{0,j}=B_{0,j}A_j^{-1}b_j$ and $\hat{\bm \omega}_{0,j}=B_{0,j}\hat A_j^{-1}b_j$, we obtain
\[
    \hat{\bm \omega}_{0,j}-\bm \omega_{0,j}
    =
    -B_{0,j}A_j^{-1}B_{0,j}^\top\Delta \bm \omega_{0,j}
    +
    r_{j,n}^{(\omega)},
\]
where
\[
    r_{j,n}^{(\omega)}
    =
    B_{0,j}A_j^{-1}(\hat A_j-A_j)A_j^{-1}(\hat A_j-A_j)\hat A_j^{-1}b_j .
\]
Because $|b_j|_1=1$, $|\hat A_j^{-1}|_1\le 2C_A$, and $|\hat A_j-A_j|_1\le s_{0,j}|\Delta|_\infty$, we have
\[
    \max_{1\le j\le p}|r_{j,n}^{(\omega)}|_1
    \le
    2C_A^3s_0^2|\Delta|_\infty^2 .
\]
Under $H_0$, $\bm \omega_{0,j}$ is the $j$th column of $\Omega$.
Thus Condition \ref{cond:Sigma.Omega.1norm} implies $\max_{1\le j\le p}|\bm \omega_{0,j}|_1 \le |\Omega|_1 \le C_\Omega$.
Consequently, under the event $2C_A s_0|\Delta|_\infty \le 1$,
\[
\begin{aligned}
    |\hat{\bm \omega}_{0,j}-\bm \omega_{0,j}|_1
    &\le
    |A_j^{-1}|_1|B_{0,j}^\top\Delta \bm \omega_{0,j}|_1
    +
    |r_{j,n}^{(\omega)}|_1  \\
    &\le
    C_A s_{0,j}|\Delta|_\infty |\bm \omega_{0,j}|_1
    +
    2C_A^3s_0^2|\Delta|_\infty^2 \\
    &\le
    C_1 s_0|\Delta|_\infty .
\end{aligned}
\]
Now fix $(i,j)\in\mathcal J_0$.
Since $i\notin S_{0,j}$ and we work under $H_0$, we have $e_i^\top \Sigma \bm \omega_{0,j}=0$.
Therefore
\[
\begin{aligned}
    \hat r_{ij}
    &=
    e_i^\top\hat\Sigma_n\hat{\bm \omega}_{0,j} \\
    &=
    e_i^\top(\Sigma+\Delta)(\bm \omega_{0,j}+\hat{\bm \omega}_{0,j}-\bm \omega_{0,j}) \\
    &=
    e_i^\top\Delta \bm \omega_{0,j}
    +
    e_i^\top\Sigma(\hat{\bm \omega}_{0,j}-\bm \omega_{0,j})
    +
    e_i^\top\Delta(\hat{\bm \omega}_{0,j}-\bm \omega_{0,j}) .
\end{aligned}
\]
Substituting the expansion of $\hat{\bm \omega}_{0,j}-\bm \omega_{0,j}$ yields
\[
    \hat r_{ij}
    =
    \bm u_{ij}^\top\Delta \bm \omega_{0,j}
    +
    \rho_{ij,n},
\]
where
\[
    \rho_{ij,n}
    =
    e_i^\top\Sigma r_{j,n}^{(\omega)}
    +
    e_i^\top\Delta(\hat{\bm \omega}_{0,j}-\bm \omega_{0,j}) .
\]
By Condition \ref{cond:Sigma.Omega.1norm} and the previous bounds,
\[
    \max_{(i,j)\in\mathcal J_0}|\rho_{ij,n}|
    \le
    |\Sigma|_1\max_{1\le j\le p}|r_{j,n}^{(\omega)}|_1
    +
    |\Delta|_\infty\max_{1\le j\le p}|\hat{\bm \omega}_{0,j}-\bm \omega_{0,j}|_1 
    \le
    C_1 s_0^2|\Delta|_\infty^2 .
\]
Since $\bm u_{ij}^\top\Delta \bm \omega_{0,j} = \frac1n\sum_{t=1}^n \mathcal X_{t,ij}$,
we obtain
\[
    |\hat T_n-T_n^{\mathrm L}|
    \le
    \sqrt n \max_{(i,j)\in\mathcal J_0}|\rho_{ij,n}|
    \le
    C_1\sqrt n s_0^2|\Delta|_\infty^2 .
\]

It remains to bound the linear term.
Condition \ref{cond:Sigma.Omega.1norm} gives $|B_{0,j}^\top\Sigma e_i|_1 \le |\Sigma e_i|_1 \le |\Sigma|_1 \le C_\Sigma$.
Together with Condition \ref{cond:local.inverse}, this implies
\[
    |\bm u_{ij}|_1
    \le
    1+|A_j^{-1}|_1|B_{0,j}^\top\Sigma e_i|_1
    \le
    1+C_A C_\Sigma .
\]
Since $|\bm \omega_{0,j}|_1\le C_\Omega$, we have
\[
    T_n^{\mathrm L}
    =
    \sqrt n\max_{(i,j)\in\mathcal J_0}|\bm u_{ij}^\top\Delta \bm \omega_{0,j}|
    \le
    C_1 \sqrt n|\Delta|_\infty .
\]
Combining this bound with the bound for $|\hat T_n-T_n^{\mathrm L}|$ gives
\[
    \hat T_n
    \le
    C_1\sqrt n|\Delta|_\infty
    +
    C_1 \sqrt n s_0^2|\Delta|_\infty^2 .
\]
This proves Proposition~\ref{prop:deterministic.reduction}.

\subsection{Proof of Theorem~\ref{thm:ga.direct}}
\label{app:proof.ga.direct}

The following lemma connects Proposition~\ref{prop:deterministic.reduction} to Theorems~\ref{thm:ga.direct} and~\ref{thm:block.bootstrap.direct} by representing the oracle scores as selected covariance coordinates of $Y_t=HX_t$ and verifying the assumptions of Theorems 1 and 2 of \citet{zhai2026simultaneous}.

\begin{lemma}[Selected covariance coordinates]
\label{lem:selected.covariance.reduction}
Assume Conditions~\ref{cond.A}, \ref{cond:Sigma.Omega.1norm}, \ref{cond:local.inverse}, and~\ref{cond:score.nondegenerate} under $H_0$.
Define
\[
    Y_t
    =
    \left(
        \{\bm u_{ij}^\top X_t\}_{(i,j)\in\mathcal J_0},
        \{\bm \omega_{0,j}^\top X_t\}_{1\le j\le p}
    \right)^\top .
\]
Then $Y_t=HX_t=\sum_{\ell=0}^\infty HA_\ell\epsilon_{t-\ell}$ is a Gaussian linear process of dimension $q=|\mathcal J_0|+p\le p^2+p$, and every $\mathcal X_{t,ij}$ is a selected coordinate of $Y_tY_t^\top-\E[Y_tY_t^\top]$.
The number of selected coordinates is $|\mathcal J_0|\le p^2$.
With $C_H=\max\{1+C_AC_\Sigma,C_\Omega\}$, the loading and coefficient bounds are
\[
\max_r|H_{r\cdot}|_1\le C_H,
\qquad
\max_r|(HA_\ell)_{r\cdot}|\le C_0C_H(1\vee\ell)^{-\beta}.
\]
For $\Xi_m=\var(m^{-1/2}\sum_{t=1}^m\mathcal X_t)$, there is a constant $C>0$ such that, for $m\ge2$,
\begin{equation}
\label{eq:selected.covariance.comparison}
|\Xi_m-\Xi|_\infty
\le C m^{-(4\beta-3)\wedge1}\log^3(m).
\end{equation}
In particular, $\min_{a\in\mathcal J_0}\Xi_{aa}\ge c_1$ and $\min_{a\in\mathcal J_0}(\Xi_m)_{aa}\ge c_1/2$ for all sufficiently large $m$.
These bounds hold uniformly over the selected coordinates, with constants depending only on the fixed $\beta$ and $C_0,c_1,C_\Sigma,C_\Omega,C_A$, not on $|\mathcal J_0|$ or $s_0$.
They verify the coefficient-decay and nondegeneracy assumptions needed by the Gaussian-approximation and block-bootstrap arguments of \citet{zhai2026simultaneous} when the maxima are restricted to these selected covariance coordinates; all logarithmic dimension factors remain $O(\log p)$.
No nondegeneracy is required of the unselected covariance coordinates of $Y_t$.
\end{lemma}

\begin{proof}
The rows of $H$ are the vectors $\bm u_{ij}^\top$ and $\bm \omega_{0,j}^\top$.
For each $(i,j)\in\mathcal J_0$,
\[
    \mathcal X_{t,ij}
    =
    (\bm u_{ij}^\top X_t)(\bm \omega_{0,j}^\top X_t)
    -
    \E[(\bm u_{ij}^\top X_t)(\bm \omega_{0,j}^\top X_t)],
\]
which proves the selected-coordinate representation.
Under $H_0$, $\bm \omega_{0,j}$ is the $j$th column of $\Omega$, and therefore Condition \ref{cond:Sigma.Omega.1norm} gives $\max_j |\bm \omega_{0,j}|_1\le C_\Omega$.
Moreover, by Conditions \ref{cond:Sigma.Omega.1norm} and \ref{cond:local.inverse},
\[
    |\bm u_{ij}|_1
    \le
    1+
    |(B_{0,j}^\top\Sigma B_{0,j})^{-1}|_1
    |B_{0,j}^\top\Sigma e_i|_1
    \le
    1+C_A C_\Sigma .
\]
Consequently, the coefficient matrices satisfy
\[
\max_r |(H A_\ell)_{r\cdot}|
\le
\max_r |H_{r\cdot}|_1
\max_s |(A_\ell)_{s\cdot}|
\le
C_0C_H(1\vee \ell)^{-\beta}.
\]
For $a,b\in\mathcal J_0$, write $\gamma_{ab}(k)=\cov(\mathcal X_{0,a},\mathcal X_{k,b})$.
By Isserlis' formula and the loading bounds,
\[
|\gamma_{ab}(k)| \le C \left(\sum_{\ell=0}^{\infty}
(1\vee \ell)^{-\beta}(1\vee(\ell+|k|))^{-\beta}\right)^2.
\]
Since $\beta>3/4$, the right-hand side is summable in $k$, uniformly in $a,b$.
Using
\[
\Xi_m=\sum_{|k|<m}\left(1-\frac{|k|}{m}\right)
\cov(\mathcal X_0,\mathcal X_k),
\]
the tail sum and the Bartlett-weight error give
\[
\begin{aligned}
|\Xi_m-\Xi|_\infty
&\le
\max_{a,b}\sum_{|k|\ge m}|\gamma_{ab}(k)|
+\frac1m\max_{a,b}\sum_{1\le|k|<m}|k||\gamma_{ab}(k)|\\
&\le C m^{-(4\beta-3)\wedge1}\log^3(m),
\end{aligned}
\]
where $C$ may depend on the fixed $\beta$.
Condition~\ref{cond:score.nondegenerate} now yields
\[
\min_{a\in\mathcal J_0}(\Xi_m)_{aa}
\ge c_1-|\Xi_m-\Xi|_\infty\ge c_1/2
\]
for all sufficiently large $m$, uniformly in the selected set and in $s_0$.
The dimensions $q\le p^2+p$ and $|\mathcal J_0|\le p^2$ imply $\log(qm)=O(\log(pm))$.
In the arguments of Theorems 1 and 2 of \citet{zhai2026simultaneous}, the coefficient bounds are uniform over the participating coordinates and the variance lower bounds are needed for the coordinates entering the maximum.
Restricting their coordinate indices to the selected entries therefore preserves these bounds, without imposing nondegeneracy on any unselected entry.
\end{proof}

By Proposition \ref{prop:deterministic.reduction}, for every $\eta>0$,
\[
\P(|\hat T_n-T_n^{\mathrm{L}}|>\eta)
\le
\P(2C_A s_0|\hat\Sigma_n-\Sigma|_\infty>1)
+
\P(C\sqrt n s_0^2|\hat\Sigma_n-\Sigma|_\infty^2>\eta).
\]
We shall now approximate $T_n^{\mathrm{L}}$ by a Gaussian maximum.
By Lemma~\ref{lem:selected.covariance.reduction}, Theorem 1 of \cite{zhai2026simultaneous} applies to our selected score vector $\mathcal X_t$, in that there exists a constant $C'$ related to the stated class constants and the fixed $\beta$ such that
\[
    \sup_{u\in\R}
    \left|
        \P(T_n^{\mathrm{L}}\le u)-\P(|G_n|_\infty\le u)
    \right|
    \le
    C' \Psi(p,n).
\]
Here, $G_n$ is a centered Gaussian vector with the covariance matrix 
\[
\Xi_n
=
\var\left(
\frac{1}{\sqrt n}\sum_{t=1}^n \mathcal X_t
\right)
=
\sum_{|k|<n}
\left(1-\frac{|k|}{n}\right)
\cov(\mathcal X_0,\mathcal X_k).
\]
It remains to replace the finite-sample covariance of the Gaussian approximation by its long-run limit.
By \eqref{eq:selected.covariance.comparison},
\[
|\Xi_n-\Xi|_\infty
\le
C n^{-(4\beta-3)\wedge 1} \log^3(n)
:=\Delta_{\Xi,n}.
\]
We now apply the Gaussian comparison inequality of \citet{chernozhukov2015comparison}, which leads to
\[
\sup_{u\in\mathbb R}
\left|
\P(|G_n|_\infty\le u)-\P(|G|_\infty\le u)
\right|
\le
C
\Delta_{\Xi,n}^{1/3}
\left\{1\vee\log\left(\frac{2|\mathcal J_0|}{\Delta_{\Xi,n}}\right)\right\}^{2/3}.
\]
Note that the right-hand side is of smaller polynomial order than $\Psi(p,n)$, and can be absorbed into the constant multiplying $\Psi(p,n)$.
Finally, Lemma~\ref{lem:selected.covariance.reduction} gives $\min_{a\in\mathcal J_0}\Xi_{aa}\ge c_1$ and $\min_{a\in\mathcal J_0}(\Xi_n)_{aa}\ge c_1/2$ for all sufficiently large $n$.
Thus the finite-sample Gaussian target $G_n$ is nondegenerate, and the anti-concentration and Gaussian comparison constants depend only on the stated class constants and the fixed $\beta$.
Therefore, we arrive at the fact that there exists a constant $C'$ related to $C_0,c_1,C_\Sigma,C_\Omega,C_A$ such that
\begin{equation}\label{eq:app.linearized.gaussian}
    \sup_{u\in\R}
    \left|
        \P(T_n^{\mathrm{L}}\le u)-\P(|G|_\infty\le u)
    \right|
    \le
    C' \Psi(p,n).
\end{equation}

We now transfer the Gaussian approximation from $T_n^{\mathrm{L}}$ to $\hat T_n$.
For every $u\in\R$ and every $\eta>0$,
\[
    \P(\hat T_n\le u)
    \le
    \P(T_n^{\mathrm{L}}\le u+\eta)+\P(|\hat T_n-T_n^{\mathrm{L}}|>\eta),
\]
and similarly,
\[
    \P(\hat T_n\le u)
    \ge
    \P(T_n^{\mathrm{L}}\le u-\eta)-\P(|\hat T_n-T_n^{\mathrm{L}}|>\eta).
\]
Combining these inequalities give
\begin{multline}
\label{eq:ga.direct.transfer}
\sup_{u\in\R}
\left|
    \P(\hat T_n\le u)-\P(|G|_\infty\le u)
\right|
\le
C' \Psi(p,n)
+
\P(2C_A s_0|\hat\Sigma_n-\Sigma|_\infty>1)
\\
+
\P(C\sqrt n s_0^2|\hat\Sigma_n-\Sigma|_\infty^2>\eta)
+
\sup_{u\in\R}\P(u<|G|_\infty\le u+\eta).
\end{multline}

By Condition \ref{cond:score.nondegenerate}, every coordinate of $G$ has variance at least $c_1$.
The upper variance bound follows from Condition \ref{cond.A} applied to the derived process $Y_t$, with coefficient constant $C_0C_H$ from Lemma~\ref{lem:selected.covariance.reduction}.
Therefore Nazarov's anti-concentration inequality implies that there exists a constant $C'$ related to $C_0,c_1,C_\Sigma,C_\Omega,C_A$ such that
\[
    \sup_{u\in\R}\P(u<|G|_\infty\le u+\eta)
    \le
    C' \eta\sqrt{\log p}.
\]

It remains to control the two probability terms in \eqref{eq:ga.direct.transfer}.
Let $Z_\Sigma$ be the centered Gaussian vector in the covariance Gaussian approximation theorem of \cite{zhai2026simultaneous}, so that under Conditions \ref{cond.A} and \ref{cond.G},
\[
    \sup_{u\in\R}
    \left|
        \P(\sqrt n|\hat\Sigma_n-\Sigma|_\infty\ge u)
        -
        \P(|Z_\Sigma|_\infty\ge u)
    \right|
    \le
    C_{\rm cov}\Psi(p,n),
\]
where $C_{\rm cov}$ is related to $C_0$ and $c_0$.
By the upper variance bound implied by Condition \ref{cond.A}, there exists $v_\Sigma=v_\Sigma(C_0)$ such that every coordinate of $Z_\Sigma$ has variance at most $v_\Sigma$.
Hence, for all $x>0$,
\[
    \P(|Z_\Sigma|_\infty>x)
    \le
    2p^2\exp\left(-\frac{x^2}{2v_\Sigma}\right).
\]
Consequently,
\[
\begin{aligned}
\P(2C_A s_0|\hat\Sigma_n-\Sigma|_\infty>1)
&=
\P\left(\sqrt n|\hat\Sigma_n-\Sigma|_\infty>\frac{\sqrt n}{2C_A s_0}\right)  \\
&\le
2p^2\exp\left(-\frac{n}{8v_\Sigma C_A^2s_0^2}\right)
+
C_{\rm cov}\Psi(p,n),
\end{aligned}
\]
and
\[
\begin{aligned}
\P(C\sqrt n s_0^2|\hat\Sigma_n-\Sigma|_\infty^2>\eta)
&=
\P\left(
    \sqrt n|\hat\Sigma_n-\Sigma|_\infty
    >
    \left(\frac{\eta\sqrt n}{Cs_0^2}\right)^{1/2}
\right)  \\
&\le
2p^2\exp\left(
    -\frac{\eta\sqrt n}{2v_\Sigma Cs_0^2}
\right)
+
C_{\rm cov}\Psi(p,n).
\end{aligned}
\]
Substituting the last three bounds into \eqref{eq:ga.direct.transfer}, we obtain
\begin{multline}
\label{eq:ga.direct.before.eta}
\sup_{u\in\R}
\left|
    \P(\hat T_n\le u)-\P(|G|_\infty\le u)
\right|
\le
(C' +2C_{\rm cov})\Psi(p,n)
+
2p^2\exp\left(-\frac{n}{8v_\Sigma C_A^2s_0^2}\right)
\\
+
2p^2\exp\left(
    -\frac{\eta\sqrt n}{2v_\Sigma Cs_0^2}
\right)
+
C' \eta\sqrt{\log p}.
\end{multline}

We now choose $\eta$ to balance the last two terms.
Write $a=(2v_\Sigma C)^{-1}\sqrt n/s_0^2$, $b=C' \sqrt{\log p}$, and $M=2p^2$.
Take $\eta_\ast = a^{-1}\log(e+Ma/b)$.
Then we have
\[
    M\exp(-a\eta_\ast)+b\eta_\ast
    \le
    \frac{b}{a}
    \left[
        1+\log\left(e+\frac{Ma}{b}\right)
    \right].
\]
Since $b/a\le C_\ast s_0^2\sqrt{\log p}/\sqrt n$ and $Ma/b\le C_\ast p^2\sqrt n/(s_0^2\sqrt{\log p})$, for a constant $C_\ast$ depending only on $C_0$, $c_1$, $C_\Sigma$, $C_\Omega$, and $C_A$, we have
\[
    M\exp(-a\eta_\ast)+b\eta_\ast
    \le
    C_\ast
    \frac{s_0^2\sqrt{\log p}}{\sqrt n}
    \left[
        1+\log\left(
            e+\frac{p^2\sqrt n}{s_0^2\sqrt{\log p}}
        \right)
    \right].
\]
Using $\log p\le \log(pn)$ and
\[
    \log\left(
        e+\frac{p^2\sqrt n}{s_0^2\sqrt{\log p}}
    \right)
    \le
    C\log(pn),
\]
the last display is bounded by
\[
    C_\ast
    \frac{s_0^2\log^{3/2}(pn)}{\sqrt n}.
\]
Substituting $\eta=\eta_\ast$ into \eqref{eq:ga.direct.before.eta} and enlarging constants, we obtain
\[
\sup_{u\in\R}
\left|
    \P(\hat T_n\le u)-\P(|G|_\infty\le u)
\right|
\le
C\Psi(p,n)
+
Cp^2\exp\left(-\frac{cn}{s_0^2}\right)
+
C\frac{s_0^2\log^{3/2}(pn)}{\sqrt n},
\]
where $C$ depends only on $C_0$, $c_0$, $c_1$, $C_\Sigma$, $C_\Omega$, and $C_A$.
This finishes the proof of Theorem~\ref{thm:ga.direct}.

\subsection{Proof of Proposition~\ref{prop:score.perturb}}
\label{app:proof.score.perturb}

Recall the proof in Appendix \ref{app:proof.det.reduction} for the notation
\[
    \Delta = \hat\Sigma_n-\Sigma,
    \qquad
    A_j = B_{0,j}^\top\Sigma B_{0,j},
    \qquad
    \hat A_j = B_{0,j}^\top\hat\Sigma_n B_{0,j},
    \qquad
    b_j = B_{0,j}^\top e_j .
\]
In this proof, the constant $C'_1 >0$ may vary from line to line, but they depend only on $C_A$ and $C_\Sigma$.
Recall the proof in in Appendix~\ref{app:proof.det.reduction} that, on the event $2C_A s_0|\Delta|_\infty\le 1$, the matrices $\hat A_j$ are invertible and $\max_{1\le j\le p}|\hat A_j^{-1}|_1\le 2C_A$.
Moreover, by the inverse identity, we get
\[
    \hat A_j^{-1}-A_j^{-1}
    =
    -A_j^{-1}(\hat A_j-A_j)\hat A_j^{-1}.
\]
And since $|\hat A_j-A_j|_1\le s_{0,j}|\Delta|_\infty$, we have
\[
    \max_{1\le j\le p}|\hat A_j^{-1}-A_j^{-1}|_1
    \le
    C'_1  s_0|\Delta|_\infty .
\]
Therefore, since $\bm \omega_{0,j}=B_{0,j}A_j^{-1}b_j$, $\hat{\bm \omega}_{0,j}=B_{0,j}\hat A_j^{-1}b_j$, and $|b_j|_1=1$, we have
\[
    \max_{1\le j\le p}|\hat{\bm \omega}_{0,j}-\bm \omega_{0,j}|_1
    \le
    C'_1  s_0|\Delta|_\infty,
    \qquad
    \max_{1\le j\le p}\{|\bm \omega_{0,j}|_1+|\hat{\bm \omega}_{0,j}|_1\}
    \le
    C'_1  .
\]
We now control the perturbation of $\bm u_{ij}$.
For $(i,j)\in\mathcal J_0$,
\[
    \hat{\bm u}_{ij}-\bm u_{ij}
    =
    -B_{0,j}(\hat A_j^{-1}-A_j^{-1})B_{0,j}^\top\Sigma e_i
    -
    B_{0,j}\hat A_j^{-1}B_{0,j}^\top\Delta e_i .
\]
Condition~\ref{cond:Sigma.Omega.1norm} gives
\[
    |B_{0,j}^\top\Sigma e_i|_1
    \le
    |\Sigma e_i|_1
    \le
    |\Sigma|_1
    \le
    C_\Sigma,
\]
meanwhile $|B_{0,j}^\top\Delta e_i|_1 \le s_{0,j}|\Delta|_\infty$.
Combining these bounds with the preceding inverse bounds yields
\[
    \max_{(i,j)\in\mathcal J_0}|\hat{\bm u}_{ij}-\bm u_{ij}|_1
    \le
    C'_1  s_0|\Delta|_\infty .
\]
This proves the first assertion.
The same argument also gives
\[
    \max_{(i,j)\in\mathcal J_0}\{|\bm u_{ij}|_1+|\hat{\bm u}_{ij}|_1\}
    \le
    C'_1 ,
\]
because
\[
    |\bm u_{ij}|_1
    \le
    1+|A_j^{-1}|_1|B_{0,j}^\top\Sigma e_i|_1
    \le
    1+C_A C_\Sigma .
\]

Next compare the plug-in and oracle scores.
Note that $\mathcal X_{t,ij}-\bar{\mathcal X}_{n,ij} = \bm u_{ij}^\top(X_tX_t^\top-\hat\Sigma_n)\bm \omega_{0,j}$.
Therefore,
\begin{multline*}
    \hat{\mathcal X}_{t,ij}-(\mathcal X_{t,ij}-\bar{\mathcal X}_{n,ij})
    =
    \hat{\bm u}_{ij}^\top(X_tX_t^\top-\hat\Sigma_n)\hat{\bm \omega}_{0,j}
    -
    \bm u_{ij}^\top(X_tX_t^\top-\hat\Sigma_n)\bm \omega_{0,j}
    \\
    =
    (\hat{\bm u}_{ij}-\bm u_{ij})^\top(X_tX_t^\top-\hat\Sigma_n)\hat{\bm \omega}_{0,j}
    +
    \bm u_{ij}^\top(X_tX_t^\top-\hat\Sigma_n)(\hat{\bm \omega}_{0,j}-\bm \omega_{0,j}).
\end{multline*}
On the same aforementioned event, we have
\[
    |X_tX_t^\top-\hat\Sigma_n|_\infty
    \le
    |X_t|_\infty^2+|\Sigma|_\infty+|\Delta|_\infty
    \le
    C'_1 (M_n^2+1),
\]
where we used $|\Sigma|_\infty\le |\Sigma|_1\le C_\Sigma$ and $|\Delta|_\infty\le (2C_A s_0)^{-1}$.
Combining the preceding bounds gives
\[
    \max_{1\le t\le n}\max_{(i,j)\in\mathcal J_0}
    |\hat{\mathcal X}_{t,ij}-(\mathcal X_{t,ij}-\bar{\mathcal X}_{n,ij})|
    \le
    C'_1  s_0|\Delta|_\infty(M_n^2+1).
\]
This finishes the proof of Proposition \ref{prop:score.perturb}.
We would like to point out that, although this consequence is not part of the proposition statement, the same pointwise bound implies that for every $l\le r\le n$,
\[
    l^{-1/2}
    \left|
        \sum_{t=r-l+1}^r
        [\hat{\mathcal X}_t-(\mathcal X_t-\bar{\mathcal X}_n)]
    \right|_\infty
    \le
    C'_1\sqrt l s_0|\hat\Sigma_n-\Sigma|_\infty(M_n^2+1).
\]

\subsection{Proof of Theorem~\ref{thm:block.bootstrap.direct}}
\label{app:proof.block.bootstrap}

We first apply the block-bootstrap theorem of \citet{zhai2026simultaneous} to the oracle score process.
Lemma~\ref{lem:selected.covariance.reduction} verifies the selected-coordinate representation, coefficient decay, and nondegeneracy conditions, including the finite-block variance lower bound for $m=n$ and $m=l$ when these lengths are sufficiently large.
The selected-coordinate block-bootstrap argument of Zhai et al. therefore applies to $\mathcal X_t$, with logarithmic dimension factors unchanged up to constants.
Define the oracle overlapping block sums $B_{r,l}^{\mathrm L} = \sum_{t=r-l+1}^r(\mathcal X_t-\bar{\mathcal X}_n)$ for $r=l,\ldots,n$, and define the corresponding oracle empirical CDF by
\[
\hat F_{n,l}^{\mathrm L}(u)
=
\frac{1}{n-l+1}
\sum_{r=l}^n
\one{l^{-1/2}|B_{r,l}^{\mathrm L}|_\infty \le u}.
\]
By Theorem 2 of \citet{zhai2026simultaneous}, there exists an event $\mathcal A_n$, at a probability no less than $1-3\Psi_B^\epsilon(p,n,\epsilon)$, such that
\[
\sup_{u\in\R}
|
\hat F_{n,l}^{\mathrm L}(u)-\P(T_n^{\mathrm L}\le u)
|
\le
C\Psi_B^{1-\epsilon}(p,n,\epsilon),
\]
where
\[
\Psi_B(p,n,\epsilon)
=
\frac{\log^{(5\tilde\beta+12)/[(6-2\epsilon)\tilde\beta+8]}(pn)}{n^{\tilde\beta/[(6-2\epsilon)\tilde\beta+8]}}.
\]
Next, we shall compare the plug-in block CDF with the oracle one.
Let $\Lambda_{n,l}=\sup_{l\le r\le n}l^{-1/2}|\hat B_{r,l}-B_{r,l}^{\mathrm L}|_\infty$.
On the event $\Lambda_{n,l}\le \eta_2$, we have $\hat F_{n,l}^{\mathrm L}(u-\eta_2) \le \hat F_{n,l}(u) \le \hat F_{n,l}^{\mathrm L}(u+\eta_2)$ for every $u\in\R$.
Hence on the event $\mathcal A_n\cap\{\Lambda_{n,l}\le\eta_2\}$,
\[
\sup_{u\in\R}
|
\hat F_{n,l}(u)-\P(T_n^{\mathrm L}\le u)
|
\le
C\Psi_B^{1-\epsilon}(p,n,\epsilon)
+
\sup_{u\in\R}
\P(u-\eta_2<T_n^{\mathrm L}\le u+\eta_2).
\]
The increment probability of $T_n^{\mathrm L}$ is controlled by the linearized Gaussian approximation established in the proof of Theorem~\ref{thm:ga.direct}.
Together with Nazarov's inequality, we get
\[
\sup_{u\in\R}
\P(u-\eta_2 <T_n^{\mathrm L}\le u+\eta_2 )
\le
C\Psi(p,n)+C\eta_2 \sqrt{\log p}.
\]
On the event $\mathcal A_n\cap\{\Lambda_{n,l}\le\eta_2\}$,
\[
\sup_{u\in\R}
|
\hat F_{n,l}(u)-\P(T_n^{\mathrm L}\le u)
|
\le
C\Psi_B^{1-\epsilon}(p,n,\epsilon)
+
C\Psi(p,n)
+
C\eta_2\sqrt{\log p}.
\]

We now transfer from $T_n^{\mathrm L}$ to $\hat T_n$.
The same smoothing argument used in the proof of Theorem~\ref{thm:ga.direct} gives, for every $\eta_1>0$,
\[
\sup_{u\in\R}
|
\P(T_n^{\mathrm L}\le u)-\P(\hat T_n\le u)
|
\le
\P(|\hat T_n-T_n^{\mathrm L}|>\eta_1)
+
C\Psi(p,n)
+
C\eta_1\sqrt{\log p}.
\]
Combining the preceding displays, on the event $\mathcal A_n\cap\{\Lambda_{n,l}\le\eta_2\}$, which is of probability no less than $1 - C\Psi_B^\epsilon(p,n,\epsilon) - \P(\Lambda_{n,l}>\eta_2)$, we obtain
\begin{multline}
\label{eq:app.bootstrap.reduction}
\sup_{u\in\R}
|
\hat F_{n,l}(u)-\P(\hat T_n\le u)
|
\\
\le
C\Psi_B^{1-\epsilon}(p,n,\epsilon)
+
C\Psi(p,n)
+
\P(|\hat T_n-T_n^{\mathrm L}|>\eta_1)
+
C(\eta_1+\eta_2)\sqrt{\log p}.
\end{multline}

Take $\eta_1=K_1s_0^2\log(pn)/\sqrt n$ and $\eta_2=K_2s_0\sqrt{l/n}\log^{3/2}(pn)$, where $K_1$ and $K_2$ are large fixed constants.
By Proposition~\ref{prop:deterministic.reduction},
\[
\P(|\hat T_n-T_n^{\mathrm L}|>\eta_1)
\le
\P(2C_A s_0|\hat\Sigma_n-\Sigma|_\infty>1)
+
\P(C\sqrt n s_0^2|\hat\Sigma_n-\Sigma|_\infty^2>\eta_1).
\]
The covariance Gaussian approximation argument used in the proof of Theorem~\ref{thm:ga.direct} gives
\[
\P(2C_A s_0|\hat\Sigma_n-\Sigma|_\infty>1)
\le
C\Psi(p,n)+Cp^2\exp\left(-\frac{cn}{s_0^2}\right),
\]
and
\[
\P(C\sqrt n s_0^2|\hat\Sigma_n-\Sigma|_\infty^2>\eta_1)
\le
C\Psi(p,n)+Cp^2\exp(-cK_1\log(pn)).
\]
With the above choice of $K_1$, the last exponential term is bounded by $C(pn)^{-1}$.
Thus
\begin{equation}
\label{eq:app.bootstrap.target.bound}
\P(|\hat T_n-T_n^{\mathrm L}|>\eta_1)
\le
C\Psi(p,n)
+
Cp^2\exp\left(-\frac{cn}{s_0^2}\right)
+
C(pn)^{-1}.
\end{equation}

We next control the plug-in block-score error.
By Proposition~\ref{prop:score.perturb}, summing the pointwise bound over at most $l$ terms gives
\[
\Lambda_{n,l}
\le
C'_1\sqrt l s_0|\hat\Sigma_n-\Sigma|_\infty(M_n^2+1)
\]
on the event $2C_A s_0|\hat\Sigma_n-\Sigma|_\infty\le 1$.
Therefore
\[
\P(\Lambda_{n,l}>\eta_2)
\le
\P(2C_A s_0|\hat\Sigma_n-\Sigma|_\infty>1)
+
\P(C'_1\sqrt l s_0|\hat\Sigma_n-\Sigma|_\infty(M_n^2+1)>\eta_2).
\]
Define
\[
\mathcal H_n
=
\left\{\sqrt n|\hat\Sigma_n-\Sigma|_\infty\le C_1\sqrt{\log(pn)}\right\}
\cap
\left\{M_n^2+1\le C_2\log(pn)\right\},
\]
where $C_1$ and $C_2$ are large fixed constants.
By the covariance Gaussian approximation theorem of \cite{zhai2026simultaneous} and the Gaussian tail bound for the corresponding approximating vector,
\[
\P\left(\sqrt n|\hat\Sigma_n-\Sigma|_\infty>C_1\sqrt{\log(pn)}\right)
\le
C\Psi(p,n)+C(pn)^{-1}.
\]
Also, Condition~\ref{cond.A} implies that each coordinate of $X_t$ is a centered Gaussian variable with variance uniformly bounded by a constant depending only on $C_0$.
Hence a union bound gives $\P(M_n>x)\le 2pn\exp(-cx^2)$, and therefore $\P(M_n^2+1>C_2\log(pn))\le C(pn)^{-1}$.
On $\mathcal H_n$,
\[
C'_1\sqrt l s_0|\hat\Sigma_n-\Sigma|_\infty(M_n^2+1)
\le
Cs_0\sqrt{\frac{l}{n}}\log^{3/2}(pn)
\le
\eta_2.
\]
It follows that
\begin{equation}
\label{eq:app.bootstrap.plugin.bound}
\P(\Lambda_{n,l}>\eta_2)
\le
C\Psi(p,n)
+
Cp^2\exp\left(-\frac{cn}{s_0^2}\right)
+
C(pn)^{-1}.
\end{equation}

Finally, $\eta_1\sqrt{\log p}\le Cs_0^2\log^{3/2}(pn)/\sqrt n$ and $\eta_2\sqrt{\log p}\le Cs_0\sqrt{l/n}\log^2(pn)$.

Combining these bounds with \eqref{eq:app.bootstrap.reduction}, \eqref{eq:app.bootstrap.target.bound}, and \eqref{eq:app.bootstrap.plugin.bound}, we obtain
\begin{multline*}
\P\Biggl(
\sup_{u\in\R}
|\hat F_{n,l}(u)-\P(\hat T_n\le u)|
>
C\Psi_B^{1-\epsilon}(p,n,\epsilon)
+
C\Psi(p,n)
+
Cp^2\exp\left(-\frac{cn}{s_0^2}\right)
\\
+
C(pn)^{-1}
+
C\frac{s_0^2\log^{3/2}(pn)}{\sqrt n}
+
Cs_0\sqrt{\frac{l}{n}}\log^2(pn)
\Biggr)
\\
\le
3\Psi_B^\epsilon(p,n,\epsilon)
+
C\left\{
\Psi(p,n)
+
p^2\exp\left(-\frac{cn}{s_0^2}\right)
+
(pn)^{-1}
\right\}.
\end{multline*}
This is already a valid probability bound, but we shall do the following simplification.
For brevity, write
$R_f=\Psi(p,n)+p^2\exp(-cn/s_0^2)+(pn)^{-1}$.
Also note that the $(pn)^{-1}$ term can be absorbed by the $s_0^2\log^{3/2}(pn)/n^{1/2}$ term, and the $\Psi(p,n)$ term can be absorbed by the $\Psi_B(p,n,\epsilon)$ term for any $\epsilon\in(0,1)$.
The threshold on the left-hand side of the preceding display is thus bounded by $C\mathfrak R_{n,l}^{1-\epsilon}(\epsilon)$ whenever $\mathfrak R_{n,l}(\epsilon) \le 1$.
Meanwhile, since $\mathfrak R_{n,l}(\epsilon)\ge \Psi_B(p,n,\epsilon) +R_f$, the concavity of $t\mapsto t^\epsilon$ gives
\[
\mathfrak R_{n,l}^\epsilon(\epsilon)-\Psi_B^\epsilon(p,n,\epsilon)
\ge
\epsilon \mathfrak R_{n,l}^{\epsilon-1}(\epsilon)[\mathfrak R_{n,l}(\epsilon)-\Psi_B(p,n,\epsilon)]
\ge
\epsilon \mathfrak R_{n,l}^{\epsilon-1}(\epsilon)R_f .
\]
There exists a $r_\epsilon\in(0,1]$ small enough (concretely, we may choose $r_\epsilon = 1\wedge (3\epsilon/C)^{1/(1-\epsilon)}$) such that $CR_f\le 3[\mathfrak R_{n,l}^\epsilon(\epsilon)-\Psi_B^\epsilon(p,n,\epsilon)]$ whenever $\mathfrak R_{n,l}(\epsilon)\le r_\epsilon$.
Then we have
\[
\P\left(
\sup_{u\in\R}
|\hat F_{n,l}(u)-\P(\hat T_n\le u)|
>
C\mathfrak R_{n,l}(\epsilon)^{1-\epsilon}
\right)
\le
3\mathfrak R_{n,l}(\epsilon)^\epsilon .
\]
This finishes the proof of Theorem~\ref{thm:block.bootstrap.direct}.

\paragraph{Validity of the bootstrap critical value.}
We now establish the level statement following Theorem~\ref{thm:block.bootstrap.direct} for a nonempty tested coordinate set $\mathcal J_0$ and a fixed $0<\alpha<1$.
Write
\[
F_n(u)=\P_{H_0}(\hat T_n\le u),
\qquad
H_n(u)=\P(|G|_\infty\le u),
\qquad
a_n=\sup_{u\in\R}|F_n(u)-H_n(u)|.
\]
Condition~\ref{cond:score.nondegenerate} ensures that $H_n$ is continuous, even if the covariance matrix of $G$ is singular.
Let $q^G_\gamma=\inf\{u\in\R:H_n(u)\ge\gamma\}$ for $0<\gamma<1$, so $H_n(q^G_\gamma)=\gamma$.
Also define
\[
\Delta_{n,l}=\sup_{u\in\R}|\hat F_{n,l}(u)-F_n(u)|,
\qquad
\eta_n=C\mathfrak R_{n,l}(\epsilon)^{1-\epsilon},
\qquad
d_n=\eta_n+a_n.
\]
On the event $\Delta_{n,l}\le\eta_n$, we have $\sup_u|\hat F_{n,l}(u)-H_n(u)|\le d_n$.
Whenever $d_n<\min\{\alpha,1-\alpha\}$, the definitions of the quantiles therefore imply
\[
q^G_{1-\alpha-d_n}\le\hat q_{1-\alpha}\le q^G_{1-\alpha+d_n}.
\]
Comparing with these deterministic quantiles and using $|F_n-H_n|_\infty\le a_n$ gives
\[
\alpha-d_n-a_n-\P(\Delta_{n,l}>\eta_n)
\le\P_{H_0}(\hat T_n>\hat q_{1-\alpha})
\le\alpha+d_n+a_n+\P(\Delta_{n,l}>\eta_n).
\]
This argument requires continuity only of the Gaussian target, not of $F_n$, and does not require independence between the statistic and the bootstrap critical value.
Theorem~\ref{thm:ga.direct} gives $a_n\le C\mathfrak R_{n,l}(\epsilon)$ when $\mathfrak R_{n,l}(\epsilon)\le1$, and Theorem~\ref{thm:block.bootstrap.direct} gives $\P(\Delta_{n,l}>\eta_n)\le3\mathfrak R_{n,l}(\epsilon)^\epsilon$.
Hence, after enlarging $C$ if necessary, whenever $\mathfrak R_{n,l}(\epsilon)\to0$,
\[
\left|\P_{H_0}(\hat T_n>\hat q_{1-\alpha})-\alpha\right|
\le C\mathfrak R_{n,l}(\epsilon)^{1-\epsilon}
+3\mathfrak R_{n,l}(\epsilon)^\epsilon
\to0.
\]

\section{Proof of Supplemental Theory}
\label{app:proofs.extensions}

\subsection{Proof of Corollary~\ref{cor:two.sample.direct}}
\label{app:proof.two.sample.direct}

\begin{proof}
We reduce the random-null problem to the fixed-null problem covered by Theorems~\ref{thm:ga.direct} and~\ref{thm:block.bootstrap.direct}.
Let
\[
\mathcal A_{n_1}
=
\{ \hat E^{(1)} = E^{(1)} \}.
\]
By assumption, $\P(\mathcal A_{n_1}^c) = \pi_{n_1} \to 0$.
Under $H_0^{1\rightarrow 2}$, on $\mathcal A_{n_1}$ we have $E^{(2)}\subseteq E^{(1)}=\hat E^{(1)}$.
Thus $\hat T^{1\rightarrow 2}$ is exactly the direct statistic computed from the second sample with the fixed null graph $E^{(1)}$.
Denote this fixed-null statistic by $T_{n_2}^{1 \rightarrow 2}$.
The bootstrap CDF, its critical value, and the bootstrap $p$-value likewise coincide with their fixed-reference counterparts on $\mathcal A_{n_1}$.
These pathwise identities and the resulting unconditional probability bounds hold without cross-sample independence.
For every $u \in \R$,
\[
|
\P_{H_0^{1\rightarrow 2}}( \hat T^{1\rightarrow 2} \le u )
-
\P( T_{n_2}^{1 \rightarrow 2} \le u )
|
\le
\P(\mathcal A_{n_1}^c)
=
\pi_{n_1}.
\]
Therefore,
\begin{multline*}
\sup_{u \in \R}
|
\P_{H_0^{1\rightarrow 2}}( \hat T^{1\rightarrow 2} \le u )
-
\P( |G^{1\rightarrow 2}|_\infty \le u )
|
\\
\qquad
\le
\sup_{u \in \R}
|
\P( T_{n_2}^{1 \rightarrow 2} \le u )
-
\P( |G^{1\rightarrow 2}|_\infty \le u )
|
+
\pi_{n_1}.
\end{multline*}
Theorem~\ref{thm:ga.direct}, applied to the second sample with the fixed null graph $E^{(1)}$, gives
\begin{multline*}
\sup_{u \in \R}
|
\P( T_{n_2}^{1 \rightarrow 2} \le u )
-
\P( |G^{1\rightarrow 2}|_\infty \le u )
|
\\
\le
C \Psi(p,n_2)
+
C p^2 \exp\{ - c n_2 / (s_0^{(1)})^2 \}
+
C \frac{(s_0^{(1)})^2 \log^{3/2}(p n_2)}{\sqrt{n_2}}.
\end{multline*}
This proves the Gaussian approximation bound.

We next prove the bootstrap bound.
Let $F_{n_2,l}^{(1)}$ be the block bootstrap CDF computed from the second sample with fixed null graph $E^{(1)}$.
Write $p_{1\rightarrow 2}^{\circ}=1-F_{n_2,l}^{(1)}(T_{n_2}^{1\rightarrow 2}-)$.
The pathwise equality $p_{1\rightarrow 2}=p_{1\rightarrow 2}^{\circ}$ on $\mathcal A_{n_1}$ gives, for every $u\in[0,1]$,
\[
\left|\P(p_{1\rightarrow 2}\le u)-\P(p_{1\rightarrow 2}^{\circ}\le u)\right|\le\pi_{n_1}.
\]
Let $\mathcal D_{n,l}^{1\rightarrow 2}=\sup_{u\in\R}|\hat F_l^{1\rightarrow 2}(u)-\P_{H_0^{1\rightarrow 2}}(\hat T^{1\rightarrow 2}\le u)|$.
Define
\[
D_{n,l}^{\circ}
=
\sup_{u \in \R}
|
F_{n_2,l}^{(1)}(u)
-
\P( T_{n_2}^{1 \rightarrow 2} \le u )
|.
\]
On $\mathcal A_{n_1}$, we have $\hat F_l^{1\rightarrow 2} = F_{n_2,l}^{(1)}$.
Hence, on $\mathcal A_{n_1}$,
\[
\begin{aligned}
\mathcal D_{n,l}^{1\rightarrow 2}
&\le
D_{n,l}^{\circ}
+
\sup_{u \in \R}
|
\P_{H_0^{1\rightarrow 2}}( \hat T^{1\rightarrow 2} \le u )
-
\P( T_{n_2}^{1 \rightarrow 2} \le u )
|
\\
&\le
D_{n,l}^{\circ}
+
\pi_{n_1}.
\end{aligned}
\]
Let
\[
\begin{aligned}
R_{n_2,l}^{(1)}(\epsilon)
={}&
\Psi_B(p,n_2,\epsilon)
+
\Psi(p,n_2)
+
p^2 \exp\{ - c n_2 / (s_0^{(1)})^2 \}
+
(p n_2)^{-1}
\\
&+
\frac{(s_0^{(1)})^2 \log^{3/2}(p n_2)}{\sqrt{n_2}}
+
s_0^{(1)}
\sqrt{\frac{l}{n_2}}
\log^2(p n_2).
\end{aligned}
\]
Theorem~\ref{thm:block.bootstrap.direct}, applied to the second sample with fixed null graph $E^{(1)}$, implies
\[
\P
(
D_{n,l}^{\circ}
>
C R_{n_2,l}^{(1)}(\epsilon)^{1-\epsilon}
)
\le
C R_{n_2,l}^{(1)}(\epsilon)^\epsilon.
\]
Therefore,
\[
\begin{aligned}
&\P
(
\mathcal D_{n,l}^{1\rightarrow 2}
>
C \pi_{n_1}
+
C R_{n_2,l}^{(1)}(\epsilon)^{1-\epsilon}
)
\\
&\qquad
\le
\P(\mathcal A_{n_1}^c)
+
\P
(
D_{n,l}^{\circ}
>
C R_{n_2,l}^{(1)}(\epsilon)^{1-\epsilon}
)
\\
&\qquad
\le
\pi_{n_1}
+
C R_{n_2,l}^{(1)}(\epsilon)^\epsilon.
\end{aligned}
\]
By the definition of $\mathfrak R_{n,l}^{1\rightarrow 2}(\epsilon)$ and by enlarging $C$ if necessary, the last display implies
\[
\P
(
\mathcal D_{n,l}^{1\rightarrow 2}
>
C \mathfrak R_{n,l}^{1\rightarrow 2}(\epsilon)^{1-\epsilon}
)
\le
C \mathfrak R_{n,l}^{1\rightarrow 2}(\epsilon)^\epsilon
\]
whenever $\mathfrak R_{n,l}^{1\rightarrow 2}(\epsilon) \le 1$.
This proves the bootstrap bound.

It remains to prove the exact asymptotic level.
Let
\[
F_n^{1\rightarrow 2}(u)
=
\P_{H_0^{1\rightarrow 2}}( \hat T^{1\rightarrow 2} \le u ),
\qquad
q_n^{1\rightarrow 2}(\gamma)
=
\inf
\{ u \in \R : F_n^{1\rightarrow 2}(u) \ge \gamma \}.
\]
The bootstrap bound implies
\[
\sup_{u \in \R}
|
\hat F_l^{1\rightarrow 2}(u)
-
F_n^{1\rightarrow 2}(u)
|
\to 0
\]
in probability.
Fix $\eta > 0$.
With probability tending to one,
\[
q_n^{1\rightarrow 2}(1-\alpha-\eta)
\le
\hat q_{1-\alpha}^{1\rightarrow 2}
\le
q_n^{1\rightarrow 2}(1-\alpha+\eta).
\]
Hence
\[
\alpha - \eta + o(1)
\le
\P_{H_0^{1\rightarrow 2}}
(
\hat T^{1\rightarrow 2} > \hat q_{1-\alpha}^{1\rightarrow 2}
)
\le
\alpha + \eta + o(1).
\]
The Gaussian approximation already proved shows that $F_n^{1\rightarrow 2}$ converges uniformly to the distribution function of $|G^{1\rightarrow 2}|_\infty$.
Since this limiting distribution is continuous at its $(1-\alpha)$-quantile, we may let $\eta \downarrow 0$.
This gives
\[
\P_{H_0^{1\rightarrow 2}}
(
\hat T^{1\rightarrow 2} > \hat q_{1-\alpha}^{1\rightarrow 2}
)
\to
\alpha.
\]
This finishes the proof.
\end{proof}

\subsection{Proof of Lemma~\ref{lem:tau.uniform.consistency}}
\label{app:proof.tau.uniform.consistency}

We prove the result by comparing the plug-in block estimator with an oracle block estimator for the local-inverse score process.
Let
\[
d_\omega
=
\sum_{j=1}^p |S_{0,j}|
\le
ps_0
\]
be the number of local-inverse coordinates.
Define the sample mean $\bar{\mathcal X}^\omega_{ij} = n^{-1} \sum_{t=1}^n \mathcal X^\omega_{t,ij}$
and the oracle centered block estimator as an intermediate quantity,
\[
\bar L^\omega_{ij}
=
\frac{1}{n-l+1}
\sum_{r=l}^n
\left(
l^{-1/2}
\sum_{t=r-l+1}^r
(
\mathcal X^\omega_{t,ij}-\bar{\mathcal X}^\omega_{ij}
)
\right)^2.
\]
The first step of the proof is to show that $\max_{1\le j\le p}\max_{i\in S_{0,j}} |\bar L^\omega_{ij}-L^\omega_{ij}|=o_\P(1)$.
Too see this, for $i\in S_{0,j}$ and $i'\in S_{0,j'}$, write $\gamma^\omega_{ij,i'j'}(h) = \cov(\mathcal X^\omega_{0,ij},\mathcal X^\omega_{h,i'j'})$.
By Isserlis' formula,
\[
\gamma^\omega_{ij,i'j'}(h)
=
(v_{0,ij}^{\top}\Gamma_h v_{0,i'j'})
(\omega_{0,j}^{\top}\Gamma_h\omega_{0,j'})
+
(v_{0,ij}^{\top}\Gamma_h\omega_{0,j'})
(\omega_{0,j}^{\top}\Gamma_h v_{0,i'j'}).
\]
Under Conditions \ref{cond:Sigma.Omega.1norm} and \ref{cond:local.inverse}, the vectors $v_{0,ij}$ and $\omega_{0,j}$ have uniformly bounded $\ell_1$-norms.
Combining this fact with Condition \ref{cond.A} gives
\[
\sup_{1\le j,j'\le p}
\sup_{i\in S_{0,j},i'\in S_{0,j'}}
|\gamma^\omega_{ij,i'j'}(h)|
\le
C(1 \vee |h|)^{-2[\beta \wedge (2\beta-1)]} \log^2(e+|h|).
\]
Since $\beta>3/4$, the summation of the right-hand side over $h\in\mathbb{N}$ is finite, and thus the long-run variance $L^\omega_{ij}$ is well defined.
Let $Z^0_{r,ij} = l^{-1/2} \sum_{t=r-l+1}^r \mathcal X^\omega_{t,ij}$.
By stationarity,
\[
\E[(Z^0_{r,ij})^2]
=
\sum_{|h|<l}
\left(1-\frac{|h|}{l}\right)
\gamma^\omega_{ij,ij}(h).
\]
By simple calculation, the deterministic bias of the block variance satisfies
\[
\max_{1\le j\le p}\max_{i\in S_{0,j}}
\left|
\E[(Z^0_{r,ij})^2]-L^\omega_{ij}
\right|
\le
C l^{-[(4\beta-3) \wedge 1]}\log^3(l).
\]
It remains to control the stochastic fluctuation of the overlapping block average.
The variables $(Z^0_{r,ij})^2-\E[(Z^0_{r,ij})^2] := W_{r,ij}$ are centered Gaussian chaos of order at most four.
By hypercontractivity for Gaussian chaoses, for every $q\ge2$,
\[
\left\|
    \frac{1}{n-l+1}
    \sum_{r=l}^n W_{r,ij}
\right\|_q
\le
Cq^2
\left\|
    \frac{1}{n-l+1}
    \sum_{r=l}^n W_{r,ij}
\right\|_2 .
\]
It remains to bound the right-hand side.
By Isserlis' formula and the covariance decay bound derived above, the covariance sequence of the block-chaos variables has total mass of order $l$, uniformly over $i,j$.
More precisely,
\[
\sup_{1\le j\le p}\sup_{i\in S_{0,j}}
\sum_{s=l}^n
|\cov(W_{r,ij},W_{s,ij})|
\le
Cl,
\qquad l\le r\le n.
\]
Indeed, there are $O(l)$ indices $s$ for which the two blocks overlap or are within distance $l$, and for these indices the covariance is bounded by a constant.
For the remaining indices, every cross-block covariance factor contains the separation between the two blocks, and Isserlis' formula reduces the bound to the summable autocovariance envelope of the local-inverse score process.
Consequently,
\[
\var\left(
\frac{1}{n-l+1}
\sum_{r=l}^n W_{r,ij}
\right)
\le
\frac{1}{(n-l+1)^2}
\sum_{r,s=l}^n
|\cov(W_{r,ij},W_{s,ij})| \le
C\frac ln,
\]
and therefore
\[
\max_{1\le j\le p}\max_{i\in S_{0,j}}
\left\|
    \frac{1}{n-l+1}
    \sum_{r=l}^n W_{r,ij}
\right\|_q
\le
Cq^2\sqrt{\frac{l}{n}}.
\]
Taking $q=2\log(ps_0n)$, using $s_0\le p$, and applying Markov's inequality gives
\[
\max_{1\le j\le p}\max_{i\in S_{0,j}}
\left|
    \frac{1}{n-l+1}
    \sum_{r=l}^n
    \left\{
        (Z^0_{r,ij})^2-\E[(Z^0_{r,ij})^2]
    \right\}
\right|
=
O_\P\left(
    \sqrt{\frac{l}{n}}\log^2(pn)
\right).
\]

We now account for the sample centering by $\bar{\mathcal X}^\omega_{n,ij}=n^{-1}\sum_{t=1}^n\mathcal X^\omega_{t,ij}$.
The Gaussian approximation argument used for the local-inverse score vector gives
\[
\mathfrak A_n := \max_{1\le j\le p}\max_{i\in S_{0,j}}
|\bar{\mathcal X}^\omega_{n,ij}|
=
O_\P
\left(
\sqrt{\frac{\log(pn)}{n}}
\right).
\]
For each $r$, we have
\[
l^{-1/2}
\sum_{t=r-l+1}^r
(\mathcal X^\omega_{t,ij}-\bar{\mathcal X}^\omega_{n,ij})
=
Z^0_{r,ij}
-
\sqrt l\,\bar{\mathcal X}^\omega_{n,ij}.
\]
Therefore,
\[
\begin{aligned}
&\frac{1}{n-l+1}
\sum_{r=l}^n
\left[
\left(
l^{-1/2}
\sum_{t=r-l+1}^r
(\mathcal X^\omega_{t,ij}-\bar{\mathcal X}^\omega_{n,ij})
\right)^2
-
(Z^0_{r,ij})^2
\right] \\
&\quad =
-2\sqrt l\,\bar{\mathcal X}^\omega_{n,ij}
\left(
\frac{1}{n-l+1}
\sum_{r=l}^n Z^0_{r,ij}
\right)
+
l(\bar{\mathcal X}^\omega_{n,ij})^2 .
\end{aligned}
\]
By Cauchy--Schwarz,
\[
\left|
\frac{1}{n-l+1}
\sum_{r=l}^n Z^0_{r,ij}
\right|
\le
\left(
\frac{1}{n-l+1}
\sum_{r=l}^n (Z^0_{r,ij})^2
\right)^{1/2}.
\]
The bias and stochastic fluctuation bounds above, together with the uniform upper bound in Condition~\ref{cond:long-run.nondegeneracy}, imply
\[
\mathfrak B_n
:=
\max_{1\le j\le p}\max_{i\in S_{0,j}}
\left(
\frac{1}{n-l+1}
\sum_{r=l}^n
(Z^0_{r,ij})^2
\right)^{1/2}
=
O_\P(1).
\]
Consequently,
\[
\begin{aligned}
&\max_{1\le j\le p}\max_{i\in S_{0,j}}
\left|
\frac{1}{n-l+1}
\sum_{r=l}^n
\left[
\left(
l^{-1/2}
\sum_{t=r-l+1}^r
(\mathcal X^\omega_{t,ij}-\bar{\mathcal X}^\omega_{n,ij})
\right)^2
-
(Z^0_{r,ij})^2
\right]
\right| \\
&\qquad
\le
2\sqrt l \mathfrak A_n \mathfrak B_n+l \mathfrak A_n^2  \\
&\qquad
=
O_\P\left(
\sqrt{\frac{l\log(pn)}{n}}
+
\frac{l\log(pn)}{n}
\right).
\end{aligned}
\]
Combining the bias, stochastic fluctuation, and sample-centering terms gives
\[
\max_{1\le j\le p}\max_{i\in S_{0,j}}
|\bar L^\omega_{ij}-L^\omega_{ij}|
=
O_\P\left(
l^{-[(4\beta-3)\wedge 1]}\log^3(l)
+
\sqrt{\frac{l}{n}}\log^2(pn)
\right)
=o_\P(1).
\]
This finishes the first step of the proof.
We now move on to the second step, i.e., comparing the plug-in block estimator $\hat L^\omega_{ij}$ with the oracle block estimator $\bar L^\omega_{ij}$.
Define
\[
Z_{r,ij}
=
l^{-1/2}
\sum_{t=r-l+1}^r
(\mathcal X^\omega_{t,ij}-\bar{\mathcal X}^\omega_{n,ij}),
\qquad
\hat Z_{r,ij}
=
l^{-1/2}
\sum_{t=r-l+1}^r
\hat{\mathcal X}^\omega_{t,ij}.
\]
Then
\[
\bar L^\omega_{ij}
=
\frac{1}{n-l+1}
\sum_{r=l}^n Z_{r,ij}^2,
\qquad
\hat L^\omega_{ij}
=
\frac{1}{n-l+1}
\sum_{r=l}^n \hat Z_{r,ij}^2.
\]

On the event $2C_As_0|\hat\Sigma_n-\Sigma|_\infty\le 1$, the same inverse perturbation argument as in Proposition \ref{prop:deterministic.reduction} gives
\[
\max_{1\le j\le p}|\hat\omega_{0,j}-\omega_{0,j}|_1
\le
Cs_0|\hat\Sigma_n-\Sigma|_\infty,
\]
and
\[
\max_{1\le j\le p}\max_{i\in S_{0,j}}
|\hat v_{0,ij}-v_{0,ij}|_1
\le
Cs_0|\hat\Sigma_n-\Sigma|_\infty.
\]
Moreover,
\[
\mathcal X^\omega_{t,ij}-\bar{\mathcal X}^\omega_{n,ij}
=
v_{0,ij}^{\top}(X_tX_t^{\top}-\hat\Sigma_n)\omega_{0,j}.
\]
Therefore,
\[
\hat{\mathcal X}^\omega_{t,ij}
-
(\mathcal X^\omega_{t,ij}-\bar{\mathcal X}^\omega_{n,ij})
=
(\hat v_{0,ij}-v_{0,ij})^{\top}
(X_tX_t^{\top}-\hat\Sigma_n)\hat\omega_{0,j}
+
v_{0,ij}^{\top}
(X_tX_t^{\top}-\hat\Sigma_n)
(\hat\omega_{0,j}-\omega_{0,j}).
\]
Using the uniform $\ell_1$-bounds on $v_{0,ij}$, $\omega_{0,j}$, $\hat v_{0,ij}$, and $\hat\omega_{0,j}$, together with
\[
|X_tX_t^{\top}-\hat\Sigma_n|_\infty
\le
C(M_n^2+1),
\]
we obtain
\[
\max_{1\le t\le n}
\max_{1\le j\le p}
\max_{i\in S_{0,j}}
\left|
\hat{\mathcal X}^\omega_{t,ij}
-
(\mathcal X^\omega_{t,ij}-\bar{\mathcal X}^\omega_{n,ij})
\right|
\le
Cs_0|\hat\Sigma_n-\Sigma|_\infty(M_n^2+1).
\]
Consequently,
\[
\max_{l\le r\le n}
\max_{1\le j\le p}
\max_{i\in S_{0,j}}
|\hat Z_{r,ij}-Z_{r,ij}|
\le
C\sqrt l s_0|\hat\Sigma_n-\Sigma|_\infty(M_n^2+1).
\]
By the Gaussian approximation bound in Theorem 1 of \citet{zhai2026simultaneous} and the Gaussian maximal tail bound,
\[
|\hat\Sigma_n-\Sigma|_\infty
=
O_\P
\left(
\sqrt{\frac{\log(pn)}{n}}
\right),
\qquad
M_n^2+1
=
O_\P(\log(pn)).
\]
Hence
\[
\max_{l\le r\le n}
\max_{1\le j\le p}
\max_{i\in S_{0,j}}
|\hat Z_{r,ij}-Z_{r,ij}|
=
O_\P
\left(
s_0\sqrt{\frac{l}{n}}\log^2(pn)
\right)
=o_\P(1),
\]
where the last step follows from $\mathfrak R_{n,l}(\epsilon)\to0$.

Since the oracle bound above and Condition \ref{cond:long-run.nondegeneracy} imply
\[
\max_{1\le j\le p}\max_{i\in S_{0,j}}
\frac{1}{n-l+1}
\sum_{r=l}^n
Z_{r,ij}^2
=
O_\P(1),
\]
we have
\[
\begin{aligned}
&\max_{1\le j\le p}\max_{i\in S_{0,j}}
|\hat L^\omega_{ij}-\bar L^\omega_{ij}| \\
&\quad\le
2
\max_{1\le j\le p}\max_{i\in S_{0,j}}
\left(
\frac{1}{n-l+1}
\sum_{r=l}^n
Z_{r,ij}^2
\right)^{1/2}
\max_{1\le j\le p}\max_{i\in S_{0,j}}
\left(
\frac{1}{n-l+1}
\sum_{r=l}^n
(\hat Z_{r,ij}-Z_{r,ij})^2
\right)^{1/2} \\
&\qquad+
\max_{1\le j\le p}\max_{i\in S_{0,j}}
\frac{1}{n-l+1}
\sum_{r=l}^n
(\hat Z_{r,ij}-Z_{r,ij})^2  \\
&\quad=o_\P(1).
\end{aligned}
\]
This finishes the second step of the proof.
Combining this result with the first step's result gives
\[
\max_{1\le j\le p}\max_{i\in S_{0,j}}
|\hat L^\omega_{ij}-L^\omega_{ij}|
=
o_\P(1).
\]
Finally, Condition \ref{cond:long-run.nondegeneracy} gives $0<c_\omega\le L^\omega_{ij}\le C_\omega<\infty$ uniformly over $1\le j\le p$ and $i\in S_{0,j}$.
Since the square-root map is Lipschitz on compact subsets of $(0,\infty)$,
\[
\max_{1\le j\le p}\max_{i\in S_{0,j}}
\left|
\frac{\hat\tau_{ij}}{\tau_{ij}}-1
\right|
=
\max_{1\le j\le p}\max_{i\in S_{0,j}}
\left|
\sqrt{\frac{\hat L^\omega_{ij}}{L^\omega_{ij}}}-1
\right|
=
o_\P(1).
\]
This finishes the proof of Lemma~\ref{lem:tau.uniform.consistency}.

\subsection{Proof of Lemma~\ref{lem:local.inverse.screening}}
\label{app:proof.local.inverse.screening}

We prove the two screening claims separately for any sequence satisfying $\delta_n\to\infty$ and $\delta_n+\sqrt{\log p}=o(\sqrt n)$, which includes the stated choice $\delta_n=\sqrt{\log n}$.
The key point is that a nonzero graph-compatible coordinate is much larger than its long-run standard error, while a zero coordinate remains on the stochastic standard-error scale.


By Condition \ref{cond:long-run.nondegeneracy}, $\tau_{ij}\asymp n^{-1/2}$ uniformly over $1\le j\le p$ and $i\in S_{0,j}$.
The same Gaussian approximation argument used in Theorem~\ref{thm:ga.direct}, applied to the local-inverse score vector $(\mathcal X^\omega_{t,ij})_{i\in S_{0,j},1\le j\le p}$, gives
\[
\max_{1\le j\le p}
\max_{i\in S_{0,j}}
\left|
\frac{\hat\omega_{0,ij}-\omega_{0,ij}}{\tau_{ij}}
\right|
=
O_{\P}(\sqrt{\log p}).
\]
Together with Lemma~\ref{lem:tau.uniform.consistency}, the same bound holds after replacing $\tau_{ij}$ by $\hat\tau_{ij}$.

We first work under $H_0:E=E_0$.
By the fixed-signal condition, there exists a constant $c_\ast>0$ such that $\min_{1\le j\le p} \min_{i\in S_{0,j}} |\omega_{0,ij}| \ge c_\ast$.
Since $\tau_{ij}\asymp n^{-1/2}$, this implies
\[
\min_{1\le j\le p}
\min_{i\in S_{0,j}}
\frac{|\omega_{0,ij}|}{\tau_{ij}}
\asymp
\sqrt n.
\]
Under $\mathfrak R_{n,l}(\epsilon)\to0$, we have $\log p=o(n)$.
By the assumed threshold condition, $\sqrt{\log p}+\delta_n = o(\sqrt n)$.
It follows that
\[
\min_{1\le j\le p}
\min_{i\in S_{0,j}}
\frac{|\hat\omega_{0,ij}|}{\hat\tau_{ij}}
>
\delta_n
\]
with probability tending to one.
That si equivalent to $\hat{\bm \Delta}_j=0$ for every $1\le j\le p$ with probability tending to one.

Now consider a fixed included alternative $E\subsetneq E_0$.
Because the true graph is contained in the proposed graph, the restricted population inverse $\bm \omega_{0,j}$ coincides with the corresponding true precision column.
Because the inclusion is strict, there exist $j_0$ and $i_0\in S_{0,j_0}$ such that $\omega_{0,i_0j_0}=0$.
For this fixed coordinate, the local-inverse expansion and the variance condition (Condition \ref{cond:long-run.nondegeneracy}) give $\hat\omega_{0,i_0j_0} / \tau_{i_0j_0} = O_{\P}(1)$, and furthermore by Lemma~\ref{lem:tau.uniform.consistency},
\[
\frac{\hat\omega_{0,i_0j_0}}{\hat\tau_{i_0j_0}}
=
O_{\P}(1).
\]
Since $\delta_n\to\infty$,
\[
\P
\left(
\frac{|\hat\omega_{0,i_0j_0}|}{\hat\tau_{i_0j_0}}
\le
\delta_n
\right)
\to1.
\]
Thus $\hat\Delta_{i_0j_0}=C_\Delta\sqrt{\log p}$ with probability tending to one.
Equivalently, $\hat{\bm \Delta}_{j_0}\ne0$ with probability tending to one.
This finishes the proof of Lemma~\ref{lem:local.inverse.screening}.

\subsection{Proof of Theorem~\ref{thm:ce.direct.formal}}
\label{app:proof.ce.direct.formal}

We first combine the uniform consistency of the long-run standard errors in Lemma~\ref{lem:tau.uniform.consistency} with the screening result in Lemma~\ref{lem:local.inverse.screening} to establish null validity, and then control the pumped statistic under fixed alternatives to prove power consistency.

We first focus on the null behavior.
The first assertion of the theorem is directly implied by Lemma~\ref{lem:local.inverse.screening}, where we have
\[
\P_{H_0}
\left(
\hat{\bm \Delta}_j=0
\text{ for all }
1\le j\le p
\right)
\to1.
\]
On this event, $\tilde{\bm \omega}_{0,j}=\hat{\bm \omega}_{0,j}$ for every $j$.
For every $i\in S_{0,j}$, the sample normal equations give $e_i^{\top}\hat\Sigma_n\hat{\bm \omega}_{0,j} - \one{i=j} = 0$.
Therefore, the maximum in $\tilde T_n$ over all $1\le i,j\le p$ is equivalent to the maximum over $(i,j)\in\mathcal J_0$.
Thus $\tilde T_n=\hat T_n$ with probability tending to one under $H_0$.
Consequently, the Gaussian approximation and block-bootstrap validity for $\tilde T_n$ now follow directly from Theorems~\ref{thm:ga.direct} and~\ref{thm:block.bootstrap.direct}.
Indeed,
\[
\sup_{u\in\R}
\left|
\P_{H_0}(\tilde T_n\le u)
-
\P(|G|_\infty\le u)
\right|
\le
\sup_{u\in\R}
\left|
\P_{H_0}(\hat T_n\le u)
-
\P(|G|_\infty\le u)
\right|
+
\P_{H_0}(\tilde T_n\ne\hat T_n),
\]
and the right-hand side converges to zero.
Similarly,
\[
\sup_{u\in\R}
\left|
\hat F_{n,l}(u)
-
\P_{H_0}(\tilde T_n\le u)
\right|
\le
\sup_{u\in\R}
\left|
\hat F_{n,l}(u)
-
\P_{H_0}(\hat T_n\le u)
\right|
+
\P_{H_0}(\tilde T_n\ne\hat T_n),
\]
which converges to zero in probability.
The asymptotic size statement follows from the same quantile-bracketing argument used for the direct statistic.

It remains to prove the power guarantee.
Recall the simplified notation that we defined earlier, $A_j=B_{0,j}^{\top}\Sigma B_{0,j}$ and $\hat A_j=B_{0,j}^{\top}\hat\Sigma_n B_{0,j}$.
Consider the event $\mathcal G_n = \{ 2C_A s_0|\hat\Sigma_n-\Sigma|_\infty\le1 \}$.
By the covariance concentration bound used in Theorems~\ref{thm:ga.direct} and~\ref{thm:block.bootstrap.direct}, we have $\P(\mathcal G_n)\to1$.
On this event,
\[
|A_j^{-1}(\hat A_j-A_j)|_1
\le
|A_j^{-1}|_1|\hat A_j-A_j|_1
\le
C_A s_0|\hat\Sigma_n-\Sigma|_\infty
\le
\frac{1}{2}.
\]
Hence $\hat A_j$ is invertible, and the Neumann-series bound gives
\[
|\hat A_j^{-1}|_1
\le
\frac{|A_j^{-1}|_1}{1-|A_j^{-1}(\hat A_j-A_j)|_1}
\le
2C_A.
\]
Since $\hat A_j$ is symmetric positive definite on $\mathcal G_n$,
\[
\lambda_{\min}(\hat A_j)
=
|\hat A_j^{-1}|_{\mathrm{op}}^{-1}
\ge
|\hat A_j^{-1}|_1^{-1}
\ge
\frac{1}{2C_A}.
\]
Suppose that $\mathcal G_n$ holds and that $\hat{\bm\Delta}_j\ne0$ for some $j$.
Let $x_j=B_{0,j}^{\top}\hat{\bm\Delta}_j$.
Then $x_j\ne0$, and every nonzero coordinate of $x_j$ is equal to $C_n$.
By the identity $B_{0,j}^{\top} ( \hat\Sigma_n\hat{\bm\omega}_{0,j}-e_j ) = 0$, we have
\[
B_{0,j}^{\top}
(
\hat\Sigma_n\tilde{\bm\omega}_{0,j}-e_j
)
=
B_{0,j}^{\top}\hat\Sigma_n\hat{\bm\Delta}_j
=
\hat A_jx_j.
\]
It follows that
\[
\begin{aligned}
\max_{1\le i\le p}
\left|
e_i^{\top}\hat\Sigma_n\tilde{\bm\omega}_{0,j}
-
\one{i=j}
\right|
&\ge
|\hat A_jx_j|_\infty
\\
&\ge
\frac{x_j^{\top}\hat A_jx_j}{|x_j|_1}
\\
&\ge
\lambda_{\min}(\hat A_j)
\frac{|x_j|_2^2}{|x_j|_1}
\\
&\ge
\frac{C_n}{2C_A}.
\end{aligned}
\]
Consequently, on the event $\mathcal G_n$,
\begin{equation}\label{eq:ce.perturb.lower.formal}
\hat{\bm\Delta}_j\ne0
\text{ for some }j
\quad
\Longrightarrow
\quad
\tilde T_n
\ge
\frac{\sqrt n C_n}{2C_A}.
\end{equation}

Now we are ready to apply Lemma~\ref{lem:local.inverse.screening} to prove the power assertion.
First consider a fixed included alternative $E\subsetneq E_0$.
By Lemma~\ref{lem:local.inverse.screening}, with probability tending to one, there exists $j$ such that $\hat{\bm \Delta}_j\ne0$.
Combining this event with \eqref{eq:ce.perturb.lower.formal} gives
\[
\tilde T_n
\ge
\frac{\sqrt n C_n}{2C_A}
=
\frac{\sqrt n C_\Delta\sqrt{\log p}}{2C_A}
\]
with probability tending to one.
Next, we claim that $\hat q_{1-\alpha}=O_{\P}(\sqrt{\log (pn)})$.
To see this, let $\hat Z_r = l^{-1/2} \left| \sum_{t=r-l+1}^r \hat{\mathcal X}_t \right|_\infty$ for $l\le r\le n$.
Then $\hat q_{1-\alpha}$ is the empirical $(1-\alpha)$-quantile of $\{\hat Z_r:l\le r\le n\}$.
Let $Z_r^{\mathrm L} = l^{-1/2} \left| \sum_{t=r-l+1}^r (\mathcal X_t-\bar{\mathcal X}_n) \right|_\infty$.
By Proposition~\ref{prop:score.perturb} and the covariance maximal bound,
\[
\max_{l\le r\le n}
|\hat Z_r-Z_r^{\mathrm L}|
\le
C\sqrt l s_0|\hat\Sigma_n-\Sigma|_\infty(M_n^2+1)
=
o_\P(1).
\]
The oracle block statistic $Z_r^{\mathrm L}$ is a maximum of selected covariance-score block sums of the same derived Gaussian linear process used in Theorem~\ref{thm:block.bootstrap.direct}.
The Gaussian maximal tail bound for these block sums, uniformly over $r$, gives
\[
\P(Z_r^{\mathrm L}>K\sqrt{\log(pn)})\to0
\]
for a sufficiently large constant $K$.
By Markov's inequality,
\[
\P\left(
\frac{1}{n-l+1}
\sum_{r=l}^n
\one{Z_r^{\mathrm L}>K\sqrt{\log(pn)}}
>\alpha
\right)
\le
\frac{1}{\alpha(n-l+1)}
\sum_{r=l}^n
\P\left(
Z_r^{\mathrm L}>K\sqrt{\log(pn)}
\right).
\]
The same selected-coordinate block-tail bound applies uniformly in $r$, so the right-hand side converges to zero.
Therefore we arrive at
\[
\hat q_{1-\alpha}
=
O_\P(\sqrt{\log(pn)}).
\]
Therefore
\[
\frac{
\sqrt n C_\Delta\sqrt{\log p}/(2C_A)
}{
\hat q_{1-\alpha}
}
\to\infty
\]
in probability for every fixed $C_\Delta>0$.
Therefore
\[
\P_{H_1}
(
\tilde T_n>\hat q_{1-\alpha}
)
\to1
\]
for every fixed included alternative $E\subsetneq E_0$.

There is, however, another case for $E\ne E_0$ not covered by $E\subsetneq E_0$.
That is, a fixed alternative not included in the null graph, namely $E\not\subseteq E_0$.
In this case, we cannot directly apply the second assertion of Lemma \ref{lem:local.inverse.screening}.
Instead, we may control the power in the following way.
Define
\[
d_0
=
\max_{(i,j)\in\mathcal J_0}
|e_i^{\top}\Sigma\bm \omega_{0,j}|.
\]
We first show that $d_0>0$.
Suppose otherwise that $d_0=0$.
Then $e_i^{\top}\Sigma\bm \omega_{0,j}=0$ for every $i\notin S_{0,j}$ and every $1\le j\le p$.
Meanwhile, by the definition of the population restricted inverse,
\[
B_{0,j}^{\top}\Sigma\bm \omega_{0,j}
=
B_{0,j}^{\top}e_j.
\]
Combining the inside-support equation with the off-support equation gives
\[
\Sigma\bm \omega_{0,j}=e_j.
\]
Since $\Sigma$ is invertible, $\bm \omega_{0,j}$ must be the $j$-th column of the true precision matrix $\Omega$.
However, $\bm \omega_{0,j}$ is supported on $S_{0,j}$ by construction.
Thus the $j$-th column of $\Omega$ is supported on $S_{0,j}$ for every $j$, which implies $E\subseteq E_0$.
This contradicts $E\not\subseteq E_0$.
Therefore $d_0>0$.

Next, recall the notation that $\Delta=\hat\Sigma_n-\Sigma$ and $b_j=B_{0,j}^{\top}e_j$ for each $j$.
Since $\bm \omega_{0,j}=B_{0,j}A_j^{-1}b_j$ and $\hat{\bm \omega}_{0,j}=B_{0,j}\hat A_j^{-1}b_j$, we obtain $|\bm \omega_{0,j}|_1\le C_A$ and $|\hat{\bm \omega}_{0,j}-\bm \omega_{0,j}|_1 \le 2C_A^2s_0|\Delta|_\infty$.
Therefore, uniformly over $(i,j)\in\mathcal J_0$,
\[
\begin{aligned}
\left|
e_i^{\top}\hat\Sigma_n\hat{\bm \omega}_{0,j}
-
e_i^{\top}\Sigma\bm \omega_{0,j}
\right|
&\le
\left|
e_i^{\top}\Delta\bm \omega_{0,j}
\right|
+
\left|
e_i^{\top}\hat\Sigma_n
(
\hat{\bm \omega}_{0,j}-\bm \omega_{0,j}
)
\right|
\\
&\le
C_A|\Delta|_\infty
+
\left(
|e_i^{\top}\Sigma B_{0,j}|_1
+
|e_i^{\top}\Delta B_{0,j}|_1
\right)
2C_A^2s_0|\Delta|_\infty
\\
&\le
C_A|\Delta|_\infty
+
\left(
C_\Sigma+s_0|\Delta|_\infty
\right)
2C_A^2s_0|\Delta|_\infty
\\
&\le
C s_0|\Delta|_\infty,
\end{aligned}
\]
where $C>0$ depends only on $C_A$ and $C_\Sigma$.
Since $\P_{H_1}(\mathcal G_n)\to1$ and $s_0|\hat\Sigma_n-\Sigma|_\infty=o_{\P_{H_1}}(1)$, it follows that
\[
\max_{(i,j)\in\mathcal J_0}
\left|
e_i^{\top}\hat\Sigma_n\hat{\bm \omega}_{0,j}
-
e_i^{\top}\Sigma\bm \omega_{0,j}
\right|
=
o_{\P_{H_1}}(1).
\]
Consequently, with the event
$\mathcal D_n = \{ \max_{(i,j)\in\mathcal J_0} | e_i^{\top}\hat\Sigma_n\hat{\bm \omega}_{0,j} - e_i^{\top}\Sigma\bm \omega_{0,j} | \le d_0/2 \}$,
we have $\P_{H_1}(\mathcal D_n)\to1$.
Also recall the aforementioned event $\mathcal G_n$, and that $\P_{H_1}(\mathcal G_n^c)\to0$.
Therefore, it suffices work on the lower bound of $\tilde T_n$ on the high-probability event $\mathcal D_n \cap \mathcal G_n$, which has probability tending to one.
We further split this event by $\mathcal P_n = \{ \hat{\bm\Delta}_j\ne0 \text{ for some }1\le j\le p \}$.

On the event $\mathcal P_n$, which covers $\mathcal P_n \cap \mathcal D_n \cap \mathcal G_n$, we have already established that
\[
\tilde T_n \ge \frac{\sqrt n C_\Delta \sqrt{\log p}}{2C_A}.
\]
Meanwhile, on the event $\mathcal P_n^c\cap\mathcal D_n$, which covers $\mathcal P_n^c \cap \mathcal D_n \cap \mathcal G_n$, we have $\tilde{\bm \omega}_{0,j}=\hat{\bm \omega}_{0,j}$ for every $j$.
Since $\mathcal J_0$ contains only off-support coordinates, we have $\one{i=j}=0$ for every $(i,j)\in\mathcal J_0$.
Thus, on $\mathcal P_n^c \cap \mathcal D_n \cap \mathcal G_n$,
\[
\begin{aligned}
\tilde T_n
&\ge
\sqrt n
\max_{(i,j)\in\mathcal J_0}
|e_i^{\top}\hat\Sigma_n\hat{\bm \omega}_{0,j}|
\\
&\ge
\sqrt n
\left[
\max_{(i,j)\in\mathcal J_0}
|e_i^{\top}\Sigma\bm \omega_{0,j}|
-
\max_{(i,j)\in\mathcal J_0}
\left|
e_i^{\top}\hat\Sigma_n\hat{\bm \omega}_{0,j}
-
e_i^{\top}\Sigma\bm \omega_{0,j}
\right|
\right]
\\
&\ge
\frac{\sqrt n d_0}{2}.
\end{aligned}
\]
Combining the two sub-events, we obtain that on $\mathcal G_n\cap\mathcal D_n$,
\[
\tilde T_n
\ge
\min
\left\{
\frac{\sqrt n C_\Delta\sqrt{\log p}}{2C_A},
\frac{\sqrt n d_0}{2}.
\right\}
\]
Therefore,
\[
\begin{aligned}
\P_{H_1}(\tilde T_n\le \hat q_{1-\alpha})
&\le
\P_{H_1}(\mathcal G_n^c)
+
\P_{H_1}(\mathcal D_n^c)
\\
&\quad+
\P_{H_1}
\left[
\hat q_{1-\alpha}
\ge
\min
\left\{
\frac{\sqrt n C_n}{2C_A},
\frac{\sqrt n d_0}{2}
\right\}
\right].
\end{aligned}
\]
The first two terms converge to zero.
For the last term, the bootstrap critical value satisfies $\hat q_{1-\alpha}=O_{\P_{H_1}}(\sqrt{\log (pn)})$ under the preceding conditions.
Hence the third term also converges to zero
Therefore, we have
\[
\P_{H_1}
(
\tilde T_n>\hat q_{1-\alpha}
)
\to1
\]
for every fixed alternative satisfying $E\not\subseteq E_0$.
Together with the included alternative case $E\subsetneq E_0$ handled above, this proves the power guarantee for every fixed alternative $E\ne E_0$.
This finishes the proof of Theorem~\ref{thm:ce.direct.formal}.

\end{document}